\documentclass[11pt,letterpaper]{article}

\usepackage{graphicx}
\usepackage{rotating}
\usepackage{algorithm}
\usepackage{color}
\usepackage{amsfonts,amsmath,amssymb,amsthm}
\usepackage{bm}
\usepackage{tikz}
\usepackage{xcolor}
\usepackage{ifthen}
\usepackage{pgf}
\usepackage{setspace}
\usepackage[authoryear,longnamesfirst,round]{natbib}
\usetikzlibrary{arrows,shapes}
\usepackage{paralist}
\usepackage{xr}
\newcommand{\wh}{\widehat} 		

\newcommand{\beq}{\begin{equation}}
\newcommand{\eeq}{\end{equation}}
\newcommand{\bea}{\begin{eqnarray}}
\newcommand{\eea}{\end{eqnarray}}
\newcommand{\ba}{\begin{array}}
\newcommand{\ea}{\end{array}}
\newcommand{\bi}{\begin{itemize}}
\newcommand{\ei}{\end{itemize}}
\newcommand{\ben}{\begin{enumerate}}
\newcommand{\een}{\end{enumerate}}

\renewcommand{\r}{\right}
\renewcommand{\l}{\left}

\theoremstyle{plain}
\newtheorem{theorem}{Theorem}
\newtheorem{lemma}{Lemma}
\newtheorem{assumption}{Assumption}

\newcommand{\blind}{0}

\begin{document}

\def\spacingset#1{\renewcommand{\baselinestretch}%
{#1}\small\normalsize} \spacingset{1}


\if0\blind
{
  \title{\bf Determining the dimension of factor structures in non-stationary large datasets}
  \author{Matteo Barigozzi
  \thanks{
    We wish to thank Giuseppe Cavaliere and Laura Coroneo for helpful comments. We are grateful to the participants to the 1st Italian Workshop on Econometrics and Empirical Economics (Milan, 25-26 January, 2018); the Financial Econometrics Seminar at CREST (Paris, 15 February, 2018); and the Economics Seminar Series at Aarhus University (Aarhus, 1 March, 2018); the Workshop on Big Data in Financial Markets at Cambridge University (Cambridge, 24-25 May, 2018); the Workshop on Macroeconomic and Financial Time Series at Lancaster University (Lancaster, 31 May-1 June, 2018). 
}\hspace{.2cm}\\
    Department of Statistics, LSE\\
     \\
    Lorenzo Trapani \\
    School of Economics, University of Nottingham}
  \maketitle
} \fi

\if1\blind
{
  \bigskip
  \bigskip
  \bigskip
  \begin{center}
    {\LARGE\bf \bf Determining the dimension of factor structures in non-stationary large datasets}
\end{center}
  \medskip
} \fi

\bigskip
\begin{abstract}

We propose a procedure to
determine the dimension of the common factor space in a large,
possibly non-stationary, dataset. Our procedure is designed to
determine whether there are (and how many) common factors (i) with
linear trends, (ii) with stochastic trends, (iii) with no trends, i.e. stationary. 
Our analysis is based on the fact
that the largest eigenvalues of a suitably scaled covariance matrix of
the data (corresponding to the common factor part) diverge, as the dimension $N$ of the dataset diverges, whilst the
others stay bounded. Therefore, we propose a class of randomised test
statistics for the null that the $p$-th eigenvalue diverges, based directly on the estimated eigenvalue. The tests only
requires minimal assumptions on the data, and no restrictions on the relative rates of divergence of $N$ and $T$ are imposed.
Monte Carlo evidence shows that our procedure
has very good finite sample properties, clearly dominating competing
approaches when no common factors are present.
We illustrate our methodology through an application to US bond yields
with different maturities observed over the last 30 years. A common
linear trend and two common stochastic trends are found and identified
as the classical level, slope and curvature factors.
\end{abstract}

\noindent%
\vfill

\newpage
\section{Introduction and main ideas \label{intro}}

In this paper, we propose a methodology to estimate the dimension of the
space spanned by the common (non-stationary) factors in the large approximate
factor model%
\begin{equation}
X_{t}=\Lambda \mathcal{F}_{t}+u_{t},  \label{model}
\end{equation}%
where $\mathcal{F}_{t}$ is the $r\times 1$ vector of common factors and $\Lambda $ is an $N\times r$ matrix of factor loadings. We will also make use of the scalar version of (\ref%
{model})%
\begin{equation}
X_{i,t}=\lambda _{i}^{\prime }\mathcal{F}_{t}+u_{i,t},  \label{modelscalar}
\end{equation}%
with $1\leq i\leq N$ and $1\leq t\leq T$. Although the relevant assumptions
are detailed in the remainder of the paper, in (\ref{model}) we are
assuming that there are three possible categories of common factors in the
vector $\mathcal{F}_{t}$: factors with a linear trend and an additional,
either an $I\left( 1\right) $ or $I\left( 0\right) $, zero mean component;
pure, zero mean $I\left( 1\right) $ factors with no trends; and, finally,
stationary common factors. Each group may well have dimension zero, e.g. factors with linear trends may not be present,
etc. We also assume, throughout the paper, that the idiosyncratic terms $u_{i,t}$ are $I\left(
0\right) $ for each $i$. Based on the classification above, we develop a technique to
estimate the number of common factors which have linear trends, and,
separately, the ones which are zero mean $I\left( 1\right) $. In particular, we use the eigenvalues of
the second moment matrix of the data, checking whether they
diverge to infinity as $\min \left( N,T\right) \rightarrow \infty $, due to the presence of common factors, or whether they are bounded. In order to construct tests for the asymptotic behaviour of the eigenvalues, (i) we derive bounds on their divergence rates, and (ii) based on those bounds we propose a randomisation procedure which produces a statistic for which we are able to derive the asymptotic behaviour under the null and the alternative hypotheses.

Determining the presence (or not), and the number of non-stationary factors in \eqref{model} can be useful in a variety of applications. 
First, we can assess the presence of unit roots in a large panel $X_{t}$ - see e.g. \citet{mp04}, \citet{baing04}, \citet{bai2006estimation}, \citet{bai2009panel}, \citet{kapetanios2011panels}, and \citet{pesaran2013panel}. Second, 
it is easy to see that $I\left( 1\right) $\ common factors in (%
\ref{model}) entail the presence of common trends and therefore cross-unit cointegration 
among the components of $X_{t}$ - see e.g. \citet{EP1994}, \citet{SW1988}, \citet{gengenbach2009}, and \citet{zyr18}. 
Indeed, if in \eqref{model} there are, say, $r$ common $I(1)$ factors (common trends), then there are $(N-r)$ cointegration relations - see also \citet{onatski2016} for an alternative approach to cointegration in large VARs. Empirical applications considering common $I(1)$ factors include: \citet{bai04} on employment
fluctuations across 60 industries in the US, \citet{moon2007} on a panel of interest rates at different maturities in the US\ and
Canada, and \citet{engel2015}, who use
(\ref{model}) as part of a strategy to develop a forecasting technique applied to a panel of bilateral US dollar
rates against 17 OECD countries. Panel models with linear trends have also been employed in the context of modelling macro-econocmic data - see \citet{maciejowska2010} - and have also proven helpful in
modelling US temperature data - see \citet{chenwu}. In all these applications, the first step of the
analysis would be the determination of the number of non-stationary and stationary common factors.

Starting at least from \citet{chamberlainrothschild83}, the literature has
developed a plethora of contributions to determine the number of common
factors in a large panel. Most methodologies focus on the case of
stationary datasets, and existing approaches can be broadly
grouped into two categories. Several studies rely on
setting a threshold for the eigenvalues of the covariance or of other second moment matrices of the $%
X_{i,t}$s - see \citet{baing02}, \citet{hallinliska07}, and \citet{ABC10}. In addition to this, the literature
has explored the possibility of using ratios of adjacent eigenvalues -
see \citet{onatski09,onatski10}, \citet{lam2012}, and \citet{ahnhorenstein13}.
Although the two approaches have different merits, the rationale
underpinning them is the same: in the presence of, say, $r$ common factors,
the largest $r$ eigenvalues of the covariance matrix of the $X_{i,t}$s\
diverge to infinity as $\min \left( N,T\right) \rightarrow \infty $, whilst
all the remaining eigenvalues stay bounded.

Fewer contributions are available to deal directly with factor models for
non-stationary data. In particular, developing an inferential theory for $\Lambda $
and $\mathcal{F}_{t}$ in a model similar to (\ref{model}) has been paid
significant attention by the statistical literature: examples include \citet{baing04}, \citet{bai04}, \citet{penaponcela}, \citet{zhangpangao}, and \citet{zyr18}. 
More specifically, \citet{bai04} extends the results by \citet{penaponcela} to the large dimensional setting, i.e. letting $N\to\infty$, and develops the inferential theory and a criterion, from which it is
possible to estimate the number of common stationary and non-stationary factors. \citet{zyr18} propose a method based on the ratio of eigenvalues of a transformation of the long-run covariance matrix to find the number of $I(d)$ factors for $d\ge 0$, but imposing the constraint $\frac N{T^\kappa}\to c\in(0,\infty)$, for $\kappa\in\left(0,\frac 12\right)$, as $\min\left( N,T\right) \to \infty$. Note that none of these two approaches considers the case of common linear trends.
To this end, \citet{maciejowska2010} develops an inferential theory for
estimated common factors and loadings in a set-up like (\ref{model}), where also linear trends are allowed, 
but no criteria for the determination of the number of common factors are
proposed. Moreover, none of these contributions deals with the case in which there are no common (non-stationary or stationary) factors: thus, these approaches cannot detect whether $X_t$ is stationary, i.e. it has no common trends, or not. Finally, \citet{baing04} propose a method to assess the presence or not of stochastic trends in large panels. However, in their setup it is assumed that at least one stationary factor is always present. Moreover, in order to assess the presence of non-stationary factors their procedure requires to estimate first the number of factors and the factors themselves using differenced data and then to test for unit roots or cointegration in the cumulated estimated factors. Although asymptotically valid, such an approach might suffer of efficiency loss due to its two-steps.

\bigskip
Our paper fills the gaps mentioned above.  Although the main arguments are laid out in the
remainder of the paper, here we present a heuristic preview of how the procedure
works. 

To begin with, in the presence of linear trends, it can be expected that the
sample second moment matrix of $X_{t}$ will diverge as fast as $T^{3}$.
Also, due to the well known eigenvalue separation property of large factor
models, it can be expected that the eigenvalues corresponding to common
factors should diverge as fast as $N$ (see \citealp{maciejowska2010}). This suggests considering the
eigenvalues of $T^{-3}\sum_{t}X_{t}X_{t}^{\prime }$ (denoted as, say, $\nu
_{1}^{\left( p\right) }$) to decide between%
\begin{equation*}
\left\{ 
\begin{tabular}{l}
$H_{0,1}^{(p)}:\nu _{1}^{\left( p\right) }\rightarrow \infty $, \\ 
$H_{A,1}^{(p)}:\nu _{1}^{\left( p\right) }<\infty $,%
\end{tabular}%
\right.
\end{equation*}%
as $\min \left( N,T\right) \rightarrow \infty $; the test can be carried out
for $p=1,2,...$, stopping as soon as the null is rejected. Similarly,
considering the zero mean, $I\left( 1\right) $ common factors, the Functional Central Limit Theorem (FCLT)
suggests that the second moment matrix of $X_{t}$ will diverge as fast as $%
T^{2}$, again with the eigenvalues corresponding to the common factors
diverging as fast as $N$ (see \citealp{bai04}). Thus, one could study the eigenvalues of $%
T^{-2}\sum_{t}X_{t}X_{t}^{\prime }$ (denoted as, say, $\nu _{2}^{\left(
p\right) }$), and decide between%
\begin{equation*}
\left\{ 
\begin{tabular}{l}
$H_{0,2}^{(p)}:\nu _{2}^{\left( p\right) }\rightarrow \infty $, \\ 
$H_{A,2}^{(p)}:\nu _{2}^{\left( p\right) }<\infty $,%
\end{tabular}%
\right.
\end{equation*}%
as $\min \left( N,T\right) \rightarrow \infty $, carrying out the test as
above. The output of these two steps is an estimate of the number of common
factors which have a linear trend and of those which are genuinely zero mean 
$I\left( 1\right) $ processes, respectively. Note that, in both steps, if we reject the null-hypothesis when $p=1$, we are in fact saying that there are no common factors. This approach could be complemented
by using $%
T^{-1}\sum_{t}\Delta X_{t}\Delta X_{t}^{\prime }$ and determining the number
of total common factors as suggested in \citet{trapani17}, which would provide an indirect estimate of the number of common stationary
factors. 

From a technical point of view, the implementation of the algorithm
described above presents one difficulty: we are unable to construct test
statistics which converge to a distributional limit under the null
hypotheses, and the best result we can obtain are rates. Thus, we base our
tests on randomising the test statistic. This approach builds on an idea of %
\citet{pearson50}, and it has been exploited in numerous contexts - see e.g. \citet{corradi2006}, \citet{bandi2014} and \citet{trapani17}. A major advantage of this procedure is that only (strong)
rates are needed, and these can be derived under quite general assumptions.
In particular, we derive our rates (and, thus, we are able to apply our
test) under no restrictions on the relative rates of divergence of $N$ and $%
T $ as they pass to infinity, which can be compared with the standard
restriction that as $\min \left (N,T\right) \rightarrow \infty $, $\frac{N%
}{T}\rightarrow c\in \left( 0,\infty \right) $, often assumed in random
matrix theory (see also \citealp{onatski2016}, where a similar restriction is needed); this entails that our procedure can be applied to virtually any dataset, being
particularly useful when either dimension is much bigger than the other. Also, our
theory requires milder restrictions on the finiteness of moments than other
contributions in the literature - see e.g. \citet{bai04} - and allows for
arbitrary levels of (weak) cross-correlation among the idiosyncratic errors $%
u_{i,t}$. 

The remainder of the paper is organised as follows. In Section \ref{theory},
we spell out the main assumptions and (in Section \ref{results}) we study
the strong rates of convergence of the eigenvalues of various rescalings of
the second moment matrix of $X_{t}$. The testing algorithm is presented in
Section \ref{test}. Numerical evidence from simulations is in Section \ref%
{numerics}, where we also report an empirical illustration. Finally, Section %
\ref{conclusions} concludes. Proofs and technical results are in appendix.

NOTATION. We define the Euclidean norm of a vector $a=\left[ a_{1},...,a_{n}%
\right] $ as $\left\Vert a\right\Vert =\left( \sum_{i=1}^{n}a_{i}^{2}\right)
^{1/2}$; \textquotedblleft a.s.\textquotedblright\ stands for
\textquotedblleft almost surely\textquotedblright , with orders of magnitude
for an a.s. convergent sequence (say $s_{T}$) being denoted as $O_{a.s.}\left(
T^{\varsigma }\right) $ and $o_{a.s.}\left( T^{\varsigma }\right) $ when, for some $%
\epsilon >0$ and $\tilde{T}<\infty $, $P\left[ \left\vert T^{-\varsigma
}s_{T}\right\vert <\epsilon \text{ for all }T\geq \tilde{T}\right] =1$ and $%
T^{-\varsigma }s_{T}\rightarrow 0$ a.s., respectively; $I_{A}\left( x\right) 
$ is the indicator function of a set $A$; finally, $C_{0}$, $C_{1}$, etc...
denote positive, finite constants whose value may differ from line to line.
Other relevant notation is introduced later on in the paper.

\section{Theory\label{theory}}

In this section, we \textit{(a)} lay out our main model - equation (\ref%
{model}) - in more precise terms and spell out the relevant assumptions
(Section \ref{assumption}), and \textit{(b)} present the main results on the
eigenvalues of various rescaled versions of the sample second moment matrix
of $X_{t}$ (Section \ref{results}).

\subsection{Model and assumptions\label{assumption}}

Recall the scalar version of our model (\ref{modelscalar}):
\begin{equation}
X_{i,t}=\lambda _{i}^{\prime }\mathcal{F}_{t}+u_{i,t}, \quad 1\le i\le N,  \label{modelscalar2}
\end{equation}
where $\lambda _{i}$ and $\mathcal{F}_{t}$ are $r\times 1$ vectors.
We begin with a representation result which, essentially, states that the
number of common factors with a linear trend (and, possibly, further
components which may be $I\left( 0\right) $ or $I\left( 1\right) $) can be
either zero - no common factors with linear trends - or $1$. This result is
originally due to \citet{maciejowska2010}, and we report it hereafter, as a
lemma, for convenience. We assume that%
\begin{equation}
\mathcal{F}_{t}=A\left( d_{1}t\right) +B\psi _{t},  \label{f-tilde}
\end{equation}%
where $A$ is a non-zero $r\times 1$ vector, $B$ an $r\times r$ matrix, and, more 
importantly, in (\ref{f-tilde}) $d_{1}$ is a dummy variable, which has the purpose
 to entertain the possibility that there are linear trends or not, according as $d_1=1$ or $0$, respectively. As far
as the $r$-dimensional vector $\psi _{t}$ is concerned, its components are
allowed to be a mixture of $I\left( 0\right) $ and $I\left( 1\right) $ processes, with
no linear trends.

We consider the following assumption, which ensures that the $\mathcal{F}_{t}$s are fully identified.

\begin{assumption}
\label{maciejowska}It holds that: (i) $A$ is non-zero; (ii) $rank\left(
B\right) =r$; (iii) the vector $\psi _{t}$ can be rearranged and partitioned
as $\left[ \psi _{at}^{\prime },\psi _{bt}^{\prime }\right] ^{\prime }$,
where $\psi _{at}\sim I\left( 1\right) $ has dimension $r_{2}+d_{2}$ and $%
\psi _{bt}\sim I\left( 0\right) $ has dimension $r_{3}+\left( 1-d_{2}\right) 
$, where $d_{2}$ is a dummy variable.
\end{assumption}

By part \textit{(ii)} of Assumption \ref{maciejowska}, $B$ has full rank,
which ensures the identification of the vector $\mathcal{F}_{t}$
irrespective of whether there is a trend or not. When there are trends, that
is when $d_{1}=1$, part \textit{(i)} of the assumption ensures that they do
have an impact on $\mathcal{F}_{t}$. Finally, by part \textit{(iii)} there
could be both $I\left( 1\right) $ and $I\left( 0\right) $ factors in the
vector $\psi _{t}$, sorted in no particular order.

\begin{lemma}
\label{maciej}Under Assumption \ref{maciejowska}, model (\ref{modelscalar2}) can be
equivalently represented as%
\begin{equation}
X_{i,t}=\lambda _{i}^{\left( 1\right) }f_{t}^{\left( 1\right) }+\lambda
_{i}^{\left( 2\right) \prime }f_{t}^{\left( 2\right) }+\lambda _{i}^{\left(
3\right) \prime }f_{t}^{\left( 3\right) }+u_{i,t}, \quad 1\le i\le N,  \label{model-scalar}
\end{equation}%
where $\lambda _{i}^{\left( 1\right) }$ and  $f_{t}^{\left( 1\right) }$ are $r_1\times 1$ with $0\le r_1\le 1$, $\lambda
_{i}^{\left( 2\right)}$ and $f_{t}^{\left( 2\right) }$ are $r_2\times 1$ vectors with $0\le r_2\le \min(N,T)$, $\lambda
_{i}^{\left( 3\right)}$ and $f_{t}^{\left( 3\right) }$ are $r_3\times 1$ vectors with $0\le r_3\le \min(N,T)$, and such that $r=r_1+r_2+r_3$ and $\lambda_{i}'=(\lambda_{i}^{(1)'}\lambda_{i}^{(2)'}\lambda_{i}^{(3)'})$ for all $i$.\\ 
Moreover, the common non-stationary factors are defined by the following equations%
\begin{eqnarray}
f_{t}^{\left( 1\right) } &=&d_{1}t+d_{2}f_{t}^{\left( 1\right) \dag }+\left(
1-d_{2}\right) g_{t},  \label{factors-a} \\
f_{t}^{\left( 1\right) \dag } &=&f_{0}^{\left( 1\right) \dag
}+\sum_{j=1}^{t}e_{t}^{\left( 1\right) },  \label{factors-a-1} \\
f_{t}^{\left( 2\right) } &=&f_{0}^{\left( 2\right)
}+\sum_{j=1}^{t}e_{t}^{\left( 2\right) }, \label{factors-b}
\end{eqnarray}%
where in (\ref{factors-a})-(\ref{factors-b}): $f_{t}^{\left( 1\right)\dag }$, $g_t$ and $e_t^{(1)}$ are $r_{1}\times 1$ vectors, $e_t^{(2)}$ is an $r_{2}\times 1$
vector, $e_{t}^{\left( 1\right) }$, $e_{t}^{\left( 2\right) }$, $g_{t}$ and $f_{t}^{\left( 3\right) }$ are $I\left( 0\right) $, and $d_1$ and $d_2$ are dummy variables.
\end{lemma}

Lemma \ref{maciej} states that the number of linear trends is either zero or
one: if an identified $k$-dimensional vector of common factors has linear
trends, this is tantamount to an identified $k$-dimensional vector of common
factors where only the first factor has a linear trend. When $r_{1}=1$ and $d_1=1$, we show in Theorem \ref{eigenvalues} below, that it does not matter whether the
remainder $d_{2}f_{t}^{\left( 1\right) \dag }+\left( 1-d_{2}\right) g_{t}$
is $I\left( 1\right) $ or $I\left( 0\right) $: the trend component is the
one that dominates. When $r_{1}=0$, there are no linear trends in the factor
structure; in this case, $f_{t}^{\left( 1\right) }$ can be $I\left( 1\right) 
$ or $I\left( 0\right) $, according as $d_{2}=1$ or $0$.

Let us denote as $r^*$ the number of non-stationary factors, and as $r$ the total number of factors. Then, based on (\ref{factors-a})-(\ref{factors-b}), the numbers of common factors in $X_{i,t}$ are summarised in the table below. 
\begin{center}
{
\begin{tabular}{l|l}
\hline
\textit{Factor type} & \textit{Number} \\ \hline
With linear trend & $r_{1}d_{1}$ \\ 
Zero mean, $I\left( 1\right) $ & $r_{2}+r_{1}\left( 1-d_{1}\right) d_{2}$ \\ 
Zero mean, $I\left( 0\right) $ & $r_{3}+r_{1}\left( 1-d_{1}\right) \left(
1-d_{2}\right) $ \\ \hline
Total non-stationary & $r^{\ast }=r_{1}d_{1}+r_{2}+r_{1}\left( 1-d_{1}\right)
d_{2}$ \\ 
Total number of common factors & $r=r^{\ast }+r_{3}+r_{1}\left(
1-d_{1}\right) \left( 1-d_{2}\right)=r_{1}+r_{2}+r_{3}$\\
\hline
\end{tabular}
}
\end{center}
\medskip

Recall that we allow for the possibility of having any of the numbers $r_1$, $r_2$, $r_3$, $r^*$, or even $r$, to be equal to zero. On the other hand, if there is no linear trend ($d_1=0$), we have at most $r_1+r_2$ zero-mean $I(1)$ factors and $r_1+r_3$ zero-mean $I(0)$ factors, while if there is a linear trend ($d_1=0$), we have at most $r_2$ zero-mean $I(1)$ factors and $r_3$ zero-mean $I(0)$ factors.

We now spell out the main assumptions. Consider the vector of zero-mean $I(1)$ factors: $f_{t}^{\ast }$, where $f_{t}^{\ast
}=\left[ f_{t}^{\left( 1\right) \dag },f_{t}^{\left( 2\right) \prime }\right]
^{\prime }$, and consider the $I(0)$ vector $e_{t}$, where $e_{t}=\left[ e_{t}^{\left( 1\right)
},e_{t}^{\left( 2\right) \prime }\right] ^{\prime }$. Both $f_{t}^{\ast }$ and $e_{t}$ are  $[r_2+r_1(1-d_1)d_2]\times 1$ vectors.\footnote{Note that if $d_1=1$ or $d_2=0$ then these vectors have dimension $r_2$ and are given by  $f_{t}^{\ast }=f_{t}^{\left( 2\right) }$ and $e_{t}=e_{t}^{\left(2\right) }$; on the other hand if $d_1=0$ and $d_2=1$ then the vectors have dimension $r_1+r_2$, thus become scalars if $r_2=0$ and $r_1=1$.}

We define the long-run
covariance matrix associated with $f_{t}^{\ast }$ as%
\begin{equation}
\Sigma _{\Delta f^*}=\lim_{T\rightarrow \infty }Var\left(
T^{-1/2}\sum_{t=1}^{T}e_{t}\right) .  \label{sigma-df}
\end{equation}

\begin{assumption}
\label{as-1}Let $\kappa >0$. It holds that (i) $E\left\Vert e_{t}\right\Vert
^{4+\kappa }<\infty $ for all $t$; (ii) $E\left\vert f_{0}^{*
}\right\vert ^{4+\kappa }<\infty $; (iii) $\Sigma _{\Delta f^*}$ is positive
definite; (iv) there exists, on a suitably enlarged probability space, an $%
\left( r_{2}+d_{2}\right) $-dimensional standard Wiener process $W\left(
t\right) $ such that, for some $\epsilon >0$,%
\begin{equation*}
\sup_{1\leq j\leq t}\left\Vert f_{j}^{\ast }-\Sigma _{\Delta f^*}^{1/2}W\left(
j\right) \right\Vert =O_{a.s.}\left( t^{1/2-\epsilon }\right);
\end{equation*}
 (v) $E\left\Vert \sum_{t=1}^{T}e_{t}\right\Vert
^{2+\kappa }\leq C_{0}\left( \sum_{t=1}^{T}E\left\Vert e_{t}\right\Vert
^{2}\right) ^{\frac{2+\kappa }{2}}$; (vi) $E\left\Vert
\sum_{t=1}^{T}f_{t}^{\ast }f_{t}^{\ast \prime }\right\Vert ^{2}$ $\leq $ $%
C_{0}T^{4}$.
\end{assumption}

Assumption \ref{as-1} poses some restrictions on the common $I(1)$ factors.
Parts \textit{(i)} and \textit{(ii)} require the existence of at least the
second moment of the innovation $e_{t}$ and of the initial condition $%
f_{0}^{\ast }$ respectively. Part \textit{(iii) }is a standard requirement, which rules out that the common, zero mean $I(1)$ factors are cointegrated: in essence, this ensures that the number of $I(1)$ common factors is
genuinely $r_{2}+d_{2}$. Part \textit{(iv)} states that a strong
approximation exists for the partial sums process $f_{t}^*$. Although this is
a high-level assumption, we prefer to write it in this form as opposed to
spelling out more primitive assumptions, since this makes the set-up more
general. Part \textit{(v) }is a Burkholder-type inequality (see
e.g. \citealp{linbai}, p. 108). Finally, part \textit{(vi) }can be verified e.g. under independence and finite fourth moments.

An important implication is that $e_{t}$ is allowed to be (weakly) dependent
over time. Considering part \textit{(iv)} in particular, starting from the
seminal paper by \citet{berkes1979}, the literature has developed several
refinements of the Strong Invariance Principle (SIP) for random vectors. In
particular, \citet{liu2009} derive the SIP for stationary
causal processes, a wide class which includes e.g.
conditional heteroskedasticity models, Volterra series, and data generated by
dynamical systems - see \citet{wu07}. Thus, part \textit{(iv)} of the assumption accommodates
for a wide variety of commonly considered DGPs. 

\begin{assumption}
\label{as-2}It holds that: (i) (a) $\max_{1\leq i\leq N,1\leq t\leq
T}E\left\vert u_{i,t}\right\vert ^{4}<\infty $; (b) $\max_{1\leq t\leq
T}E\left\Vert f_{t}^{\left( 3\right) }\right\Vert ^{4}$ $<$ $\infty $; and
(c) $\max_{1\leq t\leq T}E\left\vert g_{t}\right\vert ^{4}$ $<$ $\infty $;
(ii) (a) $\max_{1\leq i\leq N}E\left\Vert \sum_{t=1}^{T}f_{t}^{\ast
}u_{i,t}\right\Vert ^{2}$ $\leq $ $C_{0}T^{2}$; (b) $E\left\Vert
\sum_{t=1}^{T}f_{t}^{\ast }f_{t}^{\left( 3\right) \prime }\right\Vert ^{2}$ $%
\leq $ $C_{0}T^{2}$; and (c) $E\left\Vert \sum_{t=1}^{T}f_{t}^{\ast
}g_{t}\right\Vert ^{2}$ $\leq $ $C_{0}T^{2}$; (iii) $E\left\Vert
\sum_{t=1}^{T}tf_{t}^{\ast }\right\Vert ^{2}$ $\leq $ $C_{0}T^{5}$; (iv) (a) 
$\max_{1\leq i\leq N}E\left\vert \sum_{t=1}^{T}tu_{i,t}\right\vert ^{2}$ $%
\leq $ $C_{0}T^{3}$; (b) $E\left\Vert \sum_{t=1}^{T}tf_{t}^{\left( 3\right)
}\right\Vert ^{2}$ $\leq $ $C_{0}T^{3}$; and (c) $E\left\vert
\sum_{t=1}^{T}tg_{t}\right\vert ^{2}$ $\leq $ $C_{0}T^{3}$; (v) $E\left\Vert
\sum_{t=1}^{T}f_{t}^{\ast }f_{t}^{\ast \prime }\right\Vert ^{2}$ $\leq $ $%
C_{0}T^{4}$. 
\end{assumption}

Assumption \ref{as-2} deals with the idiosyncratic terms $u_{i,t}$ and the stationary factors. Part \textit{(i)} requires the existence of the $4$-th moments,
which is a milder assumption than the customary $8$-th moment existence
requirement - see \citet{bai04}. Part \textit{(ii)} could be shown from more
primitive assumptions; indeed, a prototypical assumption would require $%
e_{t} $ and $u_{i,t}$ to be independent of each other and \textit{i.i.d.}
over time - in such a case, explicit calculations would yield part \textit{%
(ii)(a)}. Parts \textit{(iii)} and \textit{(iv)} could again be shown from more
primitive assumptions; for example, part \textit{(iv)(a)} would automatically
follow if $Eu_{i,t}^{2}<\infty $ and $u_{i,t}$ is \textit{i.i.d.} across
time. Similarly, it could be verified that part \textit{(iii)} holds
whenever $E\left\Vert e_{t}\right\Vert ^{2}<\infty $ and $e_{t}$ is \textit{%
i.i.d.} across $t$.

We now spell out the assumptions for the $N\times r$ loadings matrix $\Lambda =[\lambda
_{1}|...|\lambda _{N}]^{\prime }$.

\begin{assumption}
\label{as-3}The loadings $\Lambda $ are non-stochastic with (i) $\max_{1\leq
i\leq N}\left\Vert \lambda _{i}\right\Vert <\infty $; (ii) $%
\lim_{N\rightarrow \infty }\frac{\Lambda ^{\prime }\Lambda }{N}\rightarrow
\Sigma _{\Lambda }$, where the matrix $\Sigma _{\Lambda }$ is positive
definite.
\end{assumption}

Assumption \ref{as-3} is standard in this literature - see e.g. \citet{bai04}. One consequence of
part \textit{(ii)} and Lemma \ref{maciej} is that every diagonal block of $\Sigma _{\Lambda }$, defined by the loadings of $f_t^{(1)}$, $f_t^{(2)}$ or $f_t^{(3)}$, is
also positive definite. Note that the assumption requires the loadings to be
non-stochastic; however, this could be relaxed to the case of random loadings, with no changes to the main
arguments in the paper.

Another, important consequence of Assumption \ref{as-3} is that the common
factors belonging in each category are \textquotedblleft
strong\textquotedblright\ or \textquotedblleft pervasive\textquotedblright .
We postpone a discussion of this aspect, and of the possibility of extending this
set-up, until Section \ref{weak}.

\subsection{Asymptotic behavior of eigenvalues\label{results}}

We base inference on the two matrices%
\begin{eqnarray}
\Sigma _{1} &=&\frac{1}{T^{3}}\sum_{t=1}^{T}X_{t}X_{t}^{\prime },
\label{sig-1} \\
\Sigma _{2} &=&\frac{1}{T^{2}}\sum_{t=1}^{T}X_{t}X_{t}^{\prime }.
\label{sig-22}
\end{eqnarray}%
We denote the $p$-th largest eigenvalues of $\Sigma _{1}$\ and $\Sigma _{2}$%
\ as $\nu _{1}^{\left( p\right) }$and $\nu _{2}^{\left( p\right) }$
respectively. The the asymptotic behaviour of those eigenvalues is studied in the following Theorem.

\begin{theorem}
\label{eigenvalues} Under Assumptions \ref{as-1}-\ref{as-3}, it holds that,
for every positive, bounded constants $C_{p}$, there is a slow varying sequence \begin{equation*}
l_{N,T}=\left( \ln N\right) ^{1+\epsilon }\left( \ln T\right) ^{\frac{3}{2}+\epsilon },
\end{equation*}
with $\epsilon >0$, and some random $N_{0}$
and $T_{0}$ such that, for all $N\geq N_{0}$ and $T\geq T_{0}$,
\begin{eqnarray}
\nu _{1}^{\left( p\right) } &\geq &C_{p}N,\quad\quad\quad\quad\quad\quad\; \text{ for }p\leq r_{1},
\label{lambda-x-trend} \\
\nu _{1}^{\left( p\right) } &=&O_{a.s.}\left( \frac{N}{\sqrt{T}}l_{N,T}\right),\quad  \text{ for }p>r_{1},  \label{lambda-x-notrend}
\end{eqnarray}
and
\begin{eqnarray}
\nu _{2}^{\left( p\right) } &\geq &C_{p}\frac{N}{\ln \ln T}, \quad\quad\quad\quad\; \text{ for }1\leq
p\leq r_{2}+\max \left\{ r_{1},d_{2}\right\} ,  \label{lambda-x-large} \\
\nu _{2}^{\left( p\right) } &=&O_{a.s.}\left( \frac{N}{\sqrt{T}}%
l_{N,T}\right),\quad \text{ for }p>r_{2}+\max \left\{ r_{1},d_{2}\right\} .
\label{lambda-x-small}
\end{eqnarray}
\end{theorem}

\noindent
Theorem \ref{eigenvalues} is a separation result for the eigenvalues
corresponding to common factors in $\Sigma _{1}$ and $\Sigma _{2}$ and is our first contribution.

Equations (\ref{lambda-x-trend}) and (\ref{lambda-x-notrend}) refer to the eigenvalues of $\Sigma _{1}$. The results state that the first $r_{1}$ eigenvalues diverge to infinity at a rate $N$; conversely, the remaining eigenvalues have a
smaller magnitude. We pose no restrictions on the relative rate of
divergence between $N$ and $T$ as they pass to infinity. Thus, the magnitude
of $\nu _{1}^{\left( p\right) }$, when $p>r_{1}$, may be very large; it is
however smaller - by a factor $T^{-1/2}$ - compared to that of $\nu
_{1}^{\left( p\right) }$ when $p\leq r_{1}$. In the definition of $\Sigma
_{1}$, there is a denominator given by $T^{3}$: intuitively, this is due to
the fact that the presence of a drift in the common factor $f_{t}^{\left(
1\right) }$ creates a linear trend. Norming by $T^{3}$ is needed in
order to make the trend component converge.

Equations (\ref{lambda-x-large}) and (\ref{lambda-x-small}) refer to the eigenvalues of $\Sigma_{2}$. This matrix is normalised by $T^{2}$: the main idea is that we wish to separate the eigenvalues corresponding to non-stationary factors from the
other ones. The partial sums of $f_{t}^{\ast }f_{t}^{\ast \prime }$ should
grow at least as fast as $T^{2}$ by the CLT in functional spaces; the result
in (\ref{lambda-x-large}) follows from this intuition, although, since we
need an a.s. rate, it is based on the Law of the Iterated Logarithm (see %
\citealp{donsker1977}). Similarly to $\Sigma _{1}$, the remaining
eigenvalues may also diverge, but this will happen at a slower rate. Equation (\ref{lambda-x-small}) illustrates the separation result, through the $T^{-1/2}$ term. Following the proof of the theorem, it could be readily shown that, if the idiosyncratic components $u_{i,t}$ were $I(1)$, the upper bound for $\nu
_{2}^{\left( p\right) }$ when $p>r_{2}+\max \left\{r_{1},d_{2} \right\}$ would be $O_{a.s.}\left(Nl_{N,T} \right)$ - in essence, in this case a separation result could not be shown, whence the need to assume that the $u_{i,t}$s are $I(0)$. On the other hand, one could envisage a situation where only a fraction of the $u_{i,t}$s are $I(1)$ - say $O\left( N^{\alpha_{0}}\right)$, with $\alpha_{0}<1$. In such a case, by adapting the proof of Lemma \ref{remainder} it can be shown that the upper bound in (\ref{lambda-x-small}) would become $O_{a.s.}\left(N^{\alpha_{0}}l_{N,T} \right)+O_{a.s.}\left(\frac{N}{\sqrt{T}}l_{N,T} \right)$, and thus a separation result would obtain. 

Note that Theorem \ref{eigenvalues} provides only rates: no distributional
results are available. When data are stationary, \citet{wang16} derive an
asymptotic distribution for the estimates of the diverging eigenvalues of the sample covariance matrix. We
do not know, however, if this can also be done for the $\nu _{1}^{\left(
p\right) }$s and the $\nu _{2}^{\left( p\right) }$s. Hence, in what follows
we will rely only on rates.

Finally, in order to construct the relevant test statistics, we will also
make use of the first differenced version of (\ref{modelscalar2}):%
\begin{equation}
\Delta X_{i,t}=\lambda_i' \Delta \mathcal F_{t}+\Delta u_{i,t}, \quad 1\le i\le N.  \label{first-diff}
\end{equation}

\begin{assumption}
\label{as-4}It holds that: (i) $E\left( \Delta \mathcal F_{j,t}\Delta u_{i,t}\right)
=0$ for $1\leq j\leq r$ and $1\leq i\leq N$; (ii) $\max_{1\leq i\leq N,1\leq
t\leq T}E\left\vert \Delta X_{i,t}\right\vert ^{4}\leq C_{0}$; (iii) $%
E\max_{1\leq \widetilde{t}\leq T}\left\vert \sum_{t=1}^{\widetilde{t}}\Delta
X_{h,t}\Delta X_{j,t}-E\left( \Delta X_{h,t}\Delta X_{j,t}\right)
\right\vert ^{2}\leq C_{0}$; (iv) (a) $T^{-1}\sum_{t=1}^{T}E\left( \Delta
\mathcal F_{t}\Delta \mathcal F_{t}^{\prime }\right) $ is a positive definite matrix; (b) the
largest eigenvalue of $T^{-1}\sum_{t=1}^{T}E\left( \Delta u_{t}\Delta
u_{t}^{\prime }\right) $ is finite; (c) $T^{-1}\sum_{t=1}^{T}E\left( \Delta
u_{t}\Delta u_{t}^{\prime }\right) $ is a positive definite matrix. 
\end{assumption}

Assumption \ref{as-4} is the same as Assumptions 1-3 in \citet{trapani17},
and we refer to that paper for examples in which the assumption is
satisfied. Essentially, these are the same examples that hold for $e_t$ in Assumption \ref{as-2}.


\section{Estimating the number of common factors\label{test}}

We now present the algorithm to estimate the dimension of the factor space.
We begin by determining the presence or not of a common linear trend by estimating $r_{1}$ based on $\nu _{1}^{\left( p\right) }$, and then we determine the presence or not of zero-mean $I(1)$ common factors by estimating $r^*$, based on $\nu _{2}^{\left( p\right) }$.

\subsection{Preliminary definitions}
Consider the notation $\beta =\frac{\ln N}{\ln T}$, and define%
\begin{equation}
\delta \left\{ 
\begin{tabular}{ll}
$>0$ & $\text{when }\beta <\frac{1}{2}$, \\ 
$>1-\frac{1}{2\beta }$ & $\text{when }\beta \geq \frac{1}{2}$.%
\end{tabular}%
\right.   \label{delta}
\end{equation}%
The role played by $\delta$ is the following. In view of Theorem \ref%
{eigenvalues}, the largest eigenvalues are (modulo some slowly varying
functions) proportional to $N$; the others, to $NT^{-1/2}$. When
premultiplying eigenvalues by $N^{-\delta }$, the former will be
proportional to $N^{1-\delta }$, thereby still diverging; the latter will be
proportional to $N^{1-\delta }T^{-1/2}$, which, by construction,
will drift to zero. Note that \eqref{delta} provides a general rule to set $\delta$, and we discuss specific choices in Section \ref{numerics}.

In order to construct our test statistics, we make use the eigenvalues of the matrix 
\begin{equation}\label{sig-3}
\Sigma_3 =\frac 1 T \sum_{t=1}^T\Delta X_t\Delta X_t',
\end{equation}
 which with the
same notation as before, are denoted as $\nu _{3}^{\left( p\right) }$ in decreasing order. In particular, when running our procedure for the $p$-th largest eigenvalues of $\Sigma_1$ or $\Sigma_2$, we will extensively use the quantities
\begin{equation} 
{\overline{\nu }}_{3,p}(k)=\frac{1}{4\left( N- k+1 \right) 
}\sum_{h=k}^{N}\nu _{3}^{\left( h\right) },  \label{v3-rescale}
\end{equation}
for different values of $k$. Essentially, ${\overline{\nu }}_{3,p}(k)$ is the average of all (or some) eigenvalues of $\Sigma_3$ and will be employed in order to rescale the estimated eigenvalues, so as to render all our test statistics scale invariant. In the numerical analysis of Section \ref{numerics}, we consider rescaling schemes with $k=1$, $k=p$, or $k=(p+1)$ and we discuss the impact of these choices on our results. For simplicity in the rest of this section we do not make explicit the dependence of \eqref{v3-rescale} on $k$. Finally, note the division by $4$ in \eqref{v3-rescale}, which is done, heuristically, since it is possible that $\Delta X_{i,t}$ could inflate the variance by over-differencing, and the factor $4$ represents the
largest inflation factor possible. 
%
%
%

\subsection{Determining the presence of factors with linear trends}

Consider first $\Sigma_1$ defined in \eqref{sig-1}, and its
eigenvalues $\nu _{1}^{\left( p\right) }$. Based on (\ref{lambda-x-trend})-(%
\ref{lambda-x-notrend}), the first $r_{1}$ eigenvalues of $\Sigma_1$ should diverge to positive
infinity, as $\min\left( N,T\right) \rightarrow \infty $, at a faster rate than
the $(N-r_1)$ remaining ones. Thus, the cornerstone of the algorithm to determine $%
r_{1}$ is based on checking whether $\nu _{1}^{\left( p\right) }$ diverges
sufficiently fast. In particular, as suggested by Theorem \ref{eigenvalues}, we want to construct a test for%
\begin{equation}
\left\{ 
\begin{tabular}{l}
$H_{0,1}^{\left( p\right) }:\nu _{1}^{\left( p\right) }\geq C_{p}N$, \\ 
$H_{A,1}^{\left( p\right) }:\nu _{1}^{\left( p\right) }\leq C_{p}\frac{N}{\sqrt{T}}l_{N,T}$,
\end{tabular}
\right.   \label{test1}
\end{equation}
for some positive bounded constant $C_p$. Thus, given $r_1$ we have that $H_{0,1}^{\left( p\right) }$ holds true for $p\le r_1$, while $H_{A,1}^{\left( p\right) }$ holds true for $p> r_1$.

Consider the following transformation of $\nu
_{1}^{\left( p\right) }$ 
\begin{equation}
\phi _{1}^{\left( p\right) }=\exp \left\{ N^{-\delta }\frac{\nu _{1}^{\left(
p\right) }}{{\overline{\nu }}_{3,p}}\right\} .  \label{phi-1}
\end{equation}
Then, based on \eqref{test1}, equations (\ref{lambda-x-trend}) and (\ref{lambda-x-notrend}), and given the definition \eqref{delta} of $\delta$, we have that 
\begin{equation*}
\begin{tabular}{ll}
$\lim_{\min \left( N,T\right) \rightarrow \infty }\phi _{1}^{\left( p\right)}=\infty$, &$\text{under } H_{0,1}^{(p)} \text{ i.e. for }p\leq r_{1},$ \\ 
$\lim_{\min \left( N,T\right) \rightarrow \infty }\phi _{1}^{\left( p\right)}=1,$ & $\text{under } H_{A,1}^{(p)}\text{ i.e. for }p>r_{1}.$
\end{tabular}
\end{equation*}
In principle, we could then use $\phi _{1}^{\left( p\right) }$ to test $H_{0,1}^{\left( p\right) }$. However, since $\phi _{1}^{\left( p\right) }$ either diverges to infinity or not, it does not have any randomness. Therefore, we propose to use the following randomisation
algorithm - note that other randomisations schemes would also be possible,
in principle; the one we propose, however, has been often considered in this
type of literature (see e.g. \citealp{corradi2006}, and \citealp{trapani17}).

\begin{description}
\item \textit{\underline{Step A1.1.}} Generate an \textit{i.i.d.} sample $\left\{ \xi _{1,j}^{\left( p\right) }\right\} _{j=1}^{R_{1}}$ from a common
distribution $G_{1}$, independently across $p$.

\item \textit{\underline{Step A1.2.}} For any $u$ drawn from a distribution $F_{1}\left( u\right) $, define, for  $1\le j\le R_1$,
$$
\zeta _{1,j}^{\left( p\right) }\left( u\right) =I\left[ \phi _{1}^{\left(p\right) }\times \xi _{1,j}^{\left( p\right) }\leq u\right]
$$

\item \textit{\underline{Step A1.3.}} Compute%
\begin{equation*}
\vartheta _{1}^{\left( p\right) }\left( u\right) =\frac{1}{\sqrt{R_{1}}}%
\sum_{j=1}^{R_{1}}\frac{\zeta _{1,j}^{\left( p\right) }\left( u\right)
-G_{1}\left( 0\right) }{\sqrt{G_{1}\left( 0\right) \left[ 1-G_{1}\left(
0\right) \right] }}.
\end{equation*}

\item \textit{\underline{Step A1.4.}} Compute%
\begin{equation*}
\Theta _{1}^{\left( p\right) }=\int_{-\infty }^{+\infty }\left\vert
\vartheta _{1}^{\left( p\right) }\left( u\right) \right\vert
^{2}dF_{1}\left( u\right) .
\end{equation*}
\end{description}

The intuition for considering this approach is the following. Under the null, we know that $\phi _{1}^{\left(
p\right) }$ diverges; thus, we can expect $\zeta _{1,j}^{\left( p\right)
}\left( u\right) $ to be an \textit{i.i.d.} Bernoulli sequence with expected
value exactly equal to $G_{1}\left( 0\right) $, and variance $G_{1}\left(
0\right) \left[ 1-G_{1}\left( 0\right) \right] $. 
In such case, a CLT should ensure that $\vartheta _{1}^{\left( p\right) }\left( u\right) $
follows a Normal distribution, and consequently $\Theta _{1}^{\left(
p\right) }$ should be expected to follow a Chi-squared distribution. By the
same token, under the alternative $\phi _{1}^{\left( p\right) }$ is finite,
and therefore $\zeta _{1,j}^{\left( p\right) }\left( u\right) $ should be an 
\textit{i.i.d.} Bernoulli sequence with expected value different from $%
G_{1}\left( 0\right) $; thus, $\vartheta _{1}^{\left( p\right) }\left(
u\right) $ should diverge as fast as $\sqrt{R_{1}}$ by the LLN, and
consequently $\Theta _{1}^{\left( p\right) }$ should also diverge at a rate $%
R_{1}$. The random variable $\Theta _{1}^{\left( p\right) }$ is then the statistic that we are going to use.

In order to derive the asymptotic behavior of $\Theta _{1}^{\left( p\right) }$, we need some regularity
conditions on the distributions $G_{1}(\cdot)$ and $F_{1}(\cdot)$ - see Section \ref{numerics} for a choice of these functions and of $R_1$.

\begin{assumption}
\label{fg-1}It holds that: (i) (a) $G_{1}(\cdot)$ has a bounded density function; (b) $G_{1}\left( 0\right) \neq 0$ and $G_{1}\left( 0\right) \neq 1$; (ii) $\int_{-\infty}^{\infty}
u^{2}dF_{1}\left( u\right) <\infty $. 
\end{assumption}

Let $P^{\ast }$ denote the
conditional probability with respect to \{$X_{i,t}, 1 \leq t \leq T, 1 \leq i \leq N$\}%
; we use the notation \textquotedblleft $\overset{D^{\ast }}{\rightarrow }$%
\textquotedblright\ and \textquotedblleft $\overset{P^{\ast }}{\rightarrow }$%
\textquotedblright\ to define, respectively, conditional convergence in
distribution and in probability according to $P^{\ast }$. 
It holds that

\begin{theorem}
\label{test-1}Consider $H_{0,1}^{\left( p\right) }$ and $H_{A,1}^{\left( p\right) }$ defined in (\ref{test1}). Under Assumptions \ref{as-1}-\ref{fg-1}, if
\begin{equation}
\lim_{\min \left(N,R_{1}\right) \rightarrow \infty }\sqrt{R_{1}}\exp
\left\{ -N^{1-\delta }\right\} =0,  \label{r1}
\end{equation}
then, for almost all realisations of $\left\{ e_{t},u_{i,t},1\leq i\leq N,1\leq
t\leq T\right\} $ and for all $p$, as $\min \left( N,T,R_{1}\right) \rightarrow \infty$, under $H_{0,1}^{\left( p\right) }$ it holds that
\begin{eqnarray}
&&\Theta _{1}^{\left( p\right) }\overset{D^{\ast }}{\rightarrow }\chi
_{1}^{2}, \label{null-1} 
\end{eqnarray}
and under $H_{A,1}^{\left( p\right) }$ it holds that
\begin{eqnarray}
&&\frac{1}{R_{1}}\, \frac{\int_{-\infty }^{\infty }\left[ G_{1}\left( u\right)
-G_{1}\left( 0\right) \right] ^{2}dF_{1}\left( u\right) }{G_{1}\left(
0\right) \left[ 1-G_{1}\left( 0\right) \right] }\,\Theta _{1}^{\left( p\right)
}\overset{P^{\ast }}{\rightarrow }1.
\label{alt-1}
\end{eqnarray}
 
\end{theorem}

The determination of $r_{1}$ follows from an algorithm which is based on a single step.

\begin{description}
\item \textit{\underline{Step T1.1.}} Set $p=1$ and run the test for $H_{0,1}^{(1)}:\nu_{1}^{\left( 1\right) }=\infty 
$ based on $\Theta _{1}^{\left( 1\right) }$. If the null is rejected, set $%
\widehat{r}_{1}=0$\ and stop, otherwise set $\widehat{r}_{1}=1$.
\end{description}

The output of this step is $\widehat{r}_{1}$, which is an estimate of $r_{1}$. As discussed above, $r_{1}$ can be either $0$ or $1$, whence the test
being stopped at $p=2$. The procedure based on the single \textit{Step T1.1} can therefore be viewed as a test for the presence of a common factor with a linear trend.

As can be expected, in order to ensure
that $\widehat{r}_{1}$ is consistent, a pivotal role is played by the level
of the test, $\alpha _{1}:=P^*( \Theta _{1}^{\left( p\right) }>c_{\alpha ,1})$, through the relevant critical value
denoted as $c_{\alpha ,1}$.

\begin{lemma}
\label{test-1-algo} Under the assumptions of Theorem \ref{test-1}, as $\min \left( N,T,R_{1}\right) \rightarrow \infty $, if $c_{\alpha
,1}\rightarrow \infty $ with $c_{\alpha ,1}=o\left( R_{1}\right) $, then it
holds that $P^*\left( \widehat{r}_{1}=r_{1}\right) =1$, for almost all
realisations of $\left\{ e_{t},u_{i,t},1\leq i\leq N,1\leq t\leq T\right\} $. 
\end{lemma}

Requiring that $c_{\alpha ,1}\rightarrow \infty $ is necessary in order to
have asymptotically zero Type I error probability, which ensures the
consistency result in the lemma; an immediate implication of $c_{\alpha
,1}\rightarrow \infty $ is that the level of the test is such that
\begin{equation}\label{size-1}
\lim_{\min \left( N,T,R_{1}\right) \rightarrow \infty }P^*\left( \Theta
_{1}^{\left( p\right) }>c_{\alpha ,1}\right) =0.
\end{equation}
The fact that $c_{\alpha ,1}$ diverges has also an interesting consequence on the
interpretation of the outcome of our testing procedure. It is well-known
that randomised tests will yield different results for different researchers
when applied to the same data, since the added randomness does not vanish
asymptotically. However, this is not the case with our procedure, since, when $c_{\alpha ,1}\rightarrow \infty $, \eqref{size-1} holds under $H_{0,1}^{(p)}$. 
Further, we show in the proof that having $c_{\alpha ,1}=o\left( R_{1}\right) $ affords that the probability of a Type II error is asymptotically zero, thus ensuring consistency. 

Looking at this from a different angle, the results in Lemma \ref{test-1-algo} are guaranteed by letting the level of the test $\alpha_1\to 0$ as $\min \left( N,T,R_{1}\right) \rightarrow \infty $ and we refer to Section \ref{numerics} for the choice of $\alpha_1$.
%
%

\subsection{Determining the number of non-stationary common
factors}

Consider the matrix $\Sigma_2$ defined in \eqref{sig-22} and its eigenvalues $\nu _{2}^{\left( p\right) }$. Based on Theorem \ref%
{eigenvalues}, the $r^{\ast }$ largest eigenvalues of $\Sigma_2$ should diverge to positive infinity, as $\min\left( N,T\right) \rightarrow \infty $, at a faster rate than the $(N-r^*)$ remaining ones. Therefore, we can construct a the test for
\begin{equation}
\left\{ 
\begin{tabular}{l}
$H_{0,2}^{\left( p\right) }:\nu _{2}^{\left( p\right) }\geq C_{p}\frac{N}{\ln\ln T}$, \\ 
$H_{A,2}^{\left( p\right) }:\nu _{2}^{\left( p\right) }\leq C_{p}\frac{N}{\sqrt{T}}l_{N,T}$,
\end{tabular}
\right.   \label{test2}
\end{equation}
for some positive bounded constant $C_p$. Thus, given $r^*$ we have that $H_{0,2}^{\left( p\right) }$ holds true for $p\le r^*$, while $H_{A,2}^{\left( p\right) }$ holds true for $p> r^*$.

We exploit this fact, as in the above, by considering the following
transformation of $\nu _{2}^{\left( p\right) }$ 
\begin{equation}
\phi _{2}^{\left( p\right) }=\exp \left\{ N^{-\delta }\left( \ln \ln
T\right) \frac{\nu _{2}^{\left( p\right) }}{{\overline{\nu }}_{3,p}}%
\right\} ,  \label{phi-2}
\end{equation}%
which is very similar to (\ref{phi-1}) except for the presence of the logarithmic term, which is a consequence of (\ref{lambda-x-large}).
%
Then, based on \eqref{test1}, equations (\ref{lambda-x-large}) and (\ref{lambda-x-small}), and given the definition \eqref{delta} of $\delta$, we have
that
\begin{equation*}
\begin{tabular}{ll}
$\lim_{\min \left( N,T\right) \rightarrow \infty }\phi _{2}^{\left( p\right)}=\infty $, &$\text{under } H_{0,2}^{(p)} \text{ i.e. for }p\leq r^*,$ \\  
$\lim_{\min \left( N,T\right) \rightarrow \infty }\phi _{2}^{\left( p\right)}=1,$ & $\text{under } H_{A,2}^{(p)}\text{ i.e. for }p>r^*.$
\end{tabular}
\end{equation*}
We consider the following randomisation procedure.

\begin{description}
\item \textit{\underline{Step A2.1}} Generate an \textit{i.i.d.} sample $\left\{ \xi _{2,j}^{\left( p\right) }\right\} _{j=1}^{R_{2}}$ from a common
distribution $G_{2}$, independently across $p$ and of $\left\{ \xi_{2,j}^{\left( p^{\prime }\right) }\right\} _{j=1}^{R_{2}}$ for all $p^{\prime } \neq p$.

\item \textit{\underline{Step A2.2}} For any $u$ drawn from a distribution $F_{2}\left( u\right) $, define, for  $1\le j\le R_2$,
\begin{equation*}
\zeta _{2,j}^{\left( p\right) }\left( u\right) =I\left[ \phi _{2}^{\left(
p\right) }\times \xi _{2,j}^{\left( p\right) }\leq u\right] .
\end{equation*}

\item \textit{\underline{Step A2.3.}} Compute%
\begin{equation*}
\vartheta _{2}^{\left( p\right) }\left( u\right) =\frac{1}{\sqrt{R_{2}}}%
\sum_{j=1}^{R_{2}}\frac{\zeta _{2,j}^{\left( p\right) }\left( u\right)
-G_{2}\left( 0\right) }{\sqrt{G_{2}\left( 0\right) \left[ 1-G_{2}\left(
0\right) \right] }}.
\end{equation*}

\item \textit{\underline{Step A2.4.}} Compute%
\begin{equation*}
\Theta _{2}^{\left( p\right) }=\int_{-\infty }^{+\infty }\left\vert
\vartheta _{2}^{\left( p\right) }\left( u\right) \right\vert
^{2}dF_{2}\left( u\right) .
\end{equation*}
\end{description}

The same comments as in the previous algorithm apply: in
essence, the procedure exploits the fact that under the null and the
alternative, $\phi _{2}^{\left( p\right) }$ diverges or drifts to zero
respectively: the former feature ensures (asymptotic) normality of $%
\vartheta _{2}^{\left( p\right) }\left( u\right) $, whereas the latter
entails that $\vartheta _{2}^{\left( p\right) }\left( u\right) $ diverges
under the alternative.

\begin{assumption}
\label{fg-2}It holds that: (i) (a) $G_{2}$ has a bounded density function; (b) $G_{2}\left( 0\right) \neq 0$ and $G_{2}\left( 0\right) \neq 1$; (ii) (a) $\int_{-\infty}^{\infty}
u^{2}dF_{2}\left( u\right) <\infty $.
\end{assumption}

It holds that

\begin{theorem}
\label{test-2} Consider $H_{0,2}^{\left( p\right) }$ and $H_{A,2}^{\left( p\right) }$ defined in (\ref{test2}). Under Assumptions \ref{as-1}-\ref{as-4}
and \ref{fg-2}, if
\begin{equation}
\lim_{\min \left( N,R_{2}\right ) \rightarrow \infty }\sqrt{R_{2}}\exp
\left\{ -N^{1-\delta }\right\} =0,  \label{r2}
\end{equation}%
then, for almost all realisations of $\left\{ e_{t},u_{i,t},1\leq i\leq N,1\leq t\leq T\right\} $ and for all $p$, as $\min \left( N,T,R_{2}\right) \rightarrow \infty$, under $H_{0,2}^{\left( p\right) }$ it holds that
\begin{eqnarray}
&&\Theta _{2}^{\left( p\right) }\overset{D^{\ast }}{\rightarrow }\chi
_{1}^{2}, \label{null-2} 
\end{eqnarray}
and under $H_{A,2}^{\left( p\right) }$ it holds that
\begin{eqnarray}
&&\frac{1}{R_{2}}\,\frac{\int_{-\infty }^{\infty }\left( G_{2}\left( u\right)
-G_{2}\left( 0\right) \right) ^{2}dF_{1}\left( u\right) }{G_{2}\left(
0\right) \left( 1-G_{2}\left( 0\right) \right) }\,\Theta _{2}^{\left( p\right)
}\overset{P^{\ast }}{\rightarrow }1\text{ under }H_{1}^{\left( 2\right) }.
\label{alt-2}
\end{eqnarray}
\end{theorem}

Note that, conditionally on the sample, the sequence $\left\{\Theta
_{2}^{\left( p\right) }\right\}_{p=1}^N$ is independent across $p$. We recommend
the following algorithm for the determination of $r^{\ast }$.

\begin{description}
\item \textit{\underline{Step T2.1.}} Run the test for $H_{0,2}^{(1)}:\nu _{2}^{\left( 1\right) }=\infty $
based on $\Theta _{2}^{\left( 1\right) }$. If the null is rejected, set $%
\widehat{r}^{\ast }=0$\ and stop, otherwise go to the next step.

\item \textit{\underline{Step T2.2.}} Starting from $p=1$, run the test for $H_{0,2}^{(p+1)}:\nu _{2}^{\left(
p+1\right) }=\infty $ based on $\Theta _{2}^{\left( p+1\right) }$,
constructed using an artificial sample $\left\{ \xi _{2,j}^{\left(
p+1\right) }\right\} _{j=1}^{R_{2}}$ generated independently of $\left\{ \xi
_{2,j}^{\left( 1\right) }\right\} _{j=1}^{R_{2}}$, ..., $\left\{ \xi
_{2,j}^{\left( p\right) }\right\} _{j=1}^{R_{2}}$. If the null is rejected,
set $\widehat{r}^{\ast }=p$\ and stop; otherwise repeat the step until the
null is rejected (or until a pre-specified maximum number of factors, say $r_{\max
}^{\ast }$, is reached).
\end{description}

As can be expected, in this context a pivotal role is played by the level of
the individual tests, which should be chosen so that $\widehat{r}^{\ast }$
is a good approximation of $r^{\ast }$, at least asymptotically. Similarly
to the previous case, let $c_{\alpha ,2}$ denote the critical value of the
test at each step.

\begin{lemma}
\label{test-2-algo} Under the assumptions of Theorem \ref{test-2}, as $\min \left(N,T,R_{2}\right) \rightarrow \infty $, if $r^*_{\max
}\geq r^*$ and $c_{\alpha ,2}\rightarrow \infty $ with $c_{\alpha
,2}=o\left( R_{2}\right) $, then it holds that $P\left( \widehat{r}^{\ast
}=r^{\ast }\right) =1$ for almost all realisations of $\left\{
e_{t},u_{i,t},1\leq i\leq N,1\leq t\leq T\right\} $.
\end{lemma}

This lemma has the same interpretation - especially when it comes to the
condition that $c_{\alpha ,2}\rightarrow \infty $ - as Lemma \ref
{test-1-algo}.

\subsection{Determining the number of zero-mean $I(1)$ and $I(0)$ factors}
After estimating $r^{\ast }$, it is possible to estimate the number of
common, zero-mean $I\left( 1\right) $ factors by subtracting the number of those with a linear trend from the total number of non-stationary factors, i.e. as $(\widehat{r}^{\ast }-\widehat{r}_{1})$. Under the conditions of Lemmas \ref{test-1-algo} and \ref{test-2-algo}, it is immediate to verify that
\begin{equation*}
P^*\left[ \widehat{r}^{\ast }-\widehat{r}_{1}=r_{2}+r_{1}\left( 1-d_{1}\right)d_{2}\right] =1.
\end{equation*}%

As a final remark, on the grounds of Assumption \ref{as-4} it is possible to use the algorithm proposed in \citet{trapani17} to estimate
the total number of common factors. The algorithm - based on first-differenced data - uses the eigenvalues $\nu_3^{(p)}$ of $\Sigma_3$ defined in \eqref{sig-3} in a similar way to the algorithms above. Denoting the estimate of the total number of factors as $\widehat{r}$,
the number of common $I(0)$ factors can be estimated as $\widehat{r}-\widehat{r}^{\ast }$. Under the conditions in \citet{trapani17} and of Lemma \ref{test-2-algo} above, it follows that
\begin{equation*}
P^*\left[ \widehat{r}-\widehat{r}^{\ast} = r_{3}+r_{1}\left( 1-d_{1}\right) \left(
1-d_{2}\right)\right] =1.
\end{equation*}%

\subsection{Determining the presence of weak factors\label{weak}}

By Assumption \ref{as-3}, all the common factors are assumed to be strongly pervasive.
This is a direct consequence of having $\left\Vert \Lambda \right\Vert
^{2}=O\left( N\right) $. It is however possible to imagine a situation in
which some of the common factors are \textquotedblleft weak\textquotedblright , or
\textquotedblleft less pervasive\textquotedblright : this can
arise from e.g. having genuinely weak factors, or from having strong factors
which impact only on a small number of units - see, for example, \citet{onatski12} and the references therein.

In this section, we report some heuristic arguments (similar to \citealp{trapani17}), on the ability of our procedure to determine weak
factors. For the sake of a concise discussion, but with no loss of generality, we
consider the case where all $r$ factors are zero-mean $I(1)$, and $\Lambda ^{\prime }\Lambda $ is diagonal, with diagonal elements $%
c_{p}(N)$ given by%
\begin{equation*}
c_{p}(N)=\left\{ 
\begin{array}{cc}
N & \text{for }1\leq p\leq p^{\prime } \\ 
N^{1-\kappa _{p}} & \text{for }p^{\prime }<p\leq r%
\end{array}%
\right. .
\end{equation*}%
Allowing for $\kappa _{p}\in \left( 0,1\right) $ corresponds to the case of
having weak factors, and the larger $\kappa _{p}$ the weaker the
corresponding factor. Suppose that the researcher is using $\Sigma _{2}$ and
its eigenvalues $\nu_2^{(p)}$ in order to determine $r$. Repeating exactly the same arguments in the proof of Theorem %
\ref{eigenvalues}, it can be shown that
\begin{equation}
\nu _{2}^{\left( p\right) }\geq C_{0}\frac{c_{p}(N)}{\ln \ln T}.  \label{weak-p}
\end{equation}
Equation (\ref{weak-p}) entails that, whenever $p^{\prime }<p\leq r$,
\begin{equation}
\nu _{2}^{\left( p\right) }\geq C_{0}\frac{N^{1-\kappa _{p}}}{\ln \ln T}.
\label{weak-pp}
\end{equation}
Recall that, our procedure, essentially, is based on testing whether, as $\min \left (N,T\right) \rightarrow \infty$
\begin{equation*}
\left\{ 
\begin{tabular}{l}
$H_{0,2}^{\left( p\right) }:({\ln \ln T})N^{-\delta } \nu _{2}^{\left( p\right) }\rightarrow
\infty $ \\ 
$H_{A,2}^{\left( p\right) }:({\ln \ln T})N^{-\delta } \nu _{2}^{\left( p\right) }\rightarrow
0$
\end{tabular}
\right. ,
\end{equation*}
with $\delta $ selected as per (\ref{delta}). Thus, based on (\ref{weak-pp}%
), weak factors can be determined if
\begin{equation*}
\lim_{\min \left( N,T\right) \rightarrow \infty }{N^{1-\kappa
_{p}-\delta }}\rightarrow \infty ,
\end{equation*}%
which requires%
\begin{equation}
\kappa _{p}<1-\delta .  \label{kappa-p}
\end{equation}%
On the grounds of (\ref{delta}), the constraint in (\ref{kappa-p}) explains up to which extent
weak factors can be detected. When $\beta \leq \frac{1}{2}$, that is $\frac{N}{\sqrt T}=O(1)$, then $\delta =0$, 
and we need $\kappa _{p}<1$. This entails that, when $N$ is
much smaller than $T$, our procedure is able to detect even very weak factors.
Conversely, when $\beta >\frac{1}{2}$, that is $\frac{\sqrt T}{N}=o(1)$, it is required that $\kappa _{p}<1-\frac{1}{2\beta }$: as $\beta $
increases, i.e. $N$ increases, the test is less and less able to detect weak factors. Note that
when $N$ and $T$ have the same order of magnitude, and thus $\beta =1$, weak
factors can be detected as long as $\kappa _{p}<\frac{1}{2}$ - that is, when
the eigenvalues associated with that factor diverge to infinity a bit faster
than $\sqrt N$.

\section{Monte Carlo and empirical evidence\label{numerics}}

In our experiments, we use data generated as 
\begin{align}
X_{i,t}&=\lambda_i^{(1)}f_t^{(1)}+ \lambda_i^{(2)'} f_t^{(2)}+ \lambda_i^{(3)'}   f_t^{(3)}+\sqrt \theta u_{i,t},\quad 1\le i \le N, \; 1\le t \le T,\label{dgp1} \\
f_t^{(1)}&=1+f_{t-1}^{(1)} + \epsilon_{t}^{(1)}, \quad\quad\quad\quad\quad\quad\quad\quad\quad\quad\quad \label{dgp2} \\
f_{j,t}^{(2)}&= f_{j,t-1}^{(2)} +  e_{j,t}^{(2)}, \quad e_{j,t}^{(2)}=\rho_j  e_{j,t-1}^{(2)} +  \epsilon_{j,t}^{(2)},\quad\quad j=1,\ldots, r_2,\label{dgp3} \\
f_{j,t}^{(3)}&= \alpha_j f_{j,t-1}^{(3)} +  \epsilon_{j,t}^{(3)}, \quad\quad\quad\quad\quad\quad\quad\quad\quad\quad\quad j=1,\ldots, r_3,\label{dgp4} \\
u_{i,t}&=a_i u_{i,t-1} + v_{i,t} + b_i \sum_{|k|\le C_i, k\neq 0} v_{i+k,t},\label{dgp5}
\end{align}
The loadings in \eqref{dgp1} are simulated such that each entry is distributed as $\mathcal N(0,1)$ and such that the matrix $\Lambda$ satisfies the normalization constraint $\Lambda'\Lambda=N I_r$. 
In \eqref{dgp3} and \eqref{dgp4}, we use $\rho_j\sim U[0, \bar \rho]$ with $\bar \rho\in\{0,0.4,0.8\}$, and $\alpha_j\sim U[-0.5,0.5]$ respectively. The vector $\epsilon_t=(\epsilon_{t}^{(1)}\epsilon_{1,t}^{(2)} \ldots \epsilon_{r_2,t}^{(2)}\,\epsilon_{1,t}^{(3)}\ldots\epsilon_{r_3,t}^{(3)})$ is simulated from $\mathcal N(0,\Gamma)$ independently at each $t$, with $\Gamma$ diagonal and such that 
\[
\frac {1}{NT}\sum_{i=1}^N\sum_{t=1}^T(\lambda_i^{(1)}\Delta f_t^{(1)})^2=\frac {1}{NT}\sum_{i=1}^N\sum_{t=1}^T(\lambda_i^{(2)'}\Delta  f_t^{(2)})^2=\frac {1}{NT}\sum_{i=1}^N\sum_{t=1}^T(\lambda_i^{(3)'} f_t^{(3)})^2,
\]
so that  in first differences each factor component has, on average, the same weight. In \eqref{dgp5} we allow both for serial and cross-sectional dependence in the idiosyncratic errors and for all $1\le i\le N$. We fix $a_i=0.5$, $b_i=0.5$ and $C_i=\min\l(\l\lfloor \frac N{20}\r\rfloor,10\r)$, and  the errors $v_{i,t}$ are simulated from $\mathcal N(0,1)$. Note that this model for the idiosyncratic component is the same as in \citet{ahnhorenstein13}. Last, we set the noise-to-signal as
\[
\theta = 0.5\,\frac{\sum_{i=1}^N\sum_{t=1}^T(\lambda_i^{(1)}\Delta f_t^{(1)}+\lambda_i^{(2)'}\Delta  f_t^{(2)}+\lambda_i^{(3)'}\Delta  f_t^{(3)})^2}{\sum_{i=1}^N\sum_{t=1}^T(\Delta u_{i,t})^2}.
\]

We consider the following cases:

\begin{enumerate}
\item we fix $r_3=0$ and we let $r_1\in \{0,1\}$ and $r_2\in\{0,1,2\}$ and we
use  the test based on $\phi _{1}^{\left( p\right) }$ to compute $\wh{r}_1$
(see Table \ref{tab:case1});

\item we fix $r_1=0$ and we let $r_2\in\{0,1,2\}$ and $r_3\in\{0,1,2\}$ and we
use the test based on $\phi _{2}^{\left( p\right) }$ to compute $\wh{r}^*=\wh{r}_2$ (see Table \ref{tab:case2});

\item we fix $r_1=1$ and we let $r_2\in\{0,1,2\}$ and $r_3\in\{0,1,2\}$ and we
use the test based on $\phi _{2}^{\left( p\right) }$ to compute $\wh{r}^*=\wh{r}_2+1$
(see Table \ref{tab:case3}).

\end{enumerate}

For each case, we set $N\in\{50,100,200\}$ and $T\in\{100,200,500\}$, and we simulate model \eqref{dgp1}-\eqref{dgp5} 500 times, reporting the average value of $\wh{r}_1$ or $\wh{r}_2=(\wh{r}^*-r_1)$, across simulations. Moreover, when computing $\wh{r}_2$ we compare our results with the Information Criteria by \citet{bai04}, denoted as $IC$ - this corresponds to $IC3$ in the original paper; we note that the other criteria, known as $IC_1$ and $IC_2$, deliver a similar (or worse) performance and are therefore not reported. 

Our tests are run as follows. When computing $\phi_1^{(p)}$ and $\phi_2^{(p)}$, we rescale the $p$-th eigenvalue as (see \eqref{v3-rescale})
\[
\frac{\nu_i^{(p)}}{\bar{\nu}_{3,p}(k)}=\frac{\nu_i^{(p)}}{\frac 1{4(N-k+1)}\sum^{N}_{h= k} \nu_3^{(h)}}, \qquad i=1,2.
\]
For a given $p$, we consider three different rescaling schemes corresponding to three different choices for $k$:
\begin{itemize}
\item [\underline{${BT1}$}:] when $k=1$, i.e. $\bar{\nu}_{3,p}(k) = \frac 1{4N}\sum^{N}_{h= 1} \nu_3^{(h)}$; 
\item [\underline{${BT2}$}:] when $k=p$, i.e. $\bar{\nu}_{3,p}(k) = \frac 1{4(N-p+1)}\sum^{N}_{h= p} \nu_3^{(h)}$; 
\item [\underline{${BT3}$}:] when $k=(p+1)$, i.e. $\bar{\nu}_{3,p}(k) = \frac 1{4(N-p)}\sum^{N}_{h= {p+1}} \nu_3^{(h)}$.
\end{itemize} 
We then divide the eigenvalues by $N^{\delta}$, where (see \eqref{delta})
\begin{equation}
\delta=\left\{
\begin{array}{ll}
\delta^*, & \text{when }\frac{\ln N}{\ln T} <\frac{1}{2}, \\ 
1-\frac{1}{2\beta }+\delta^*, & \text{when }\frac{\ln N}{\ln T} \geq \frac{1}{2},
\end{array}
\right. 
\end{equation}
with $\delta^* = 10^{-5}$. Thence, for each $p$, in the first step of the randomisation algorithm, $\{\xi_{1,j}^{(p)}\}_{j=1}^{R_1}$ and $\{\xi_{2,j}^{(p)}\}_{j=1}^{R_2}$ are generated from a standard normal distribution, with $R_1 = N$ and $R_2 = N$, if $p=1$ or $R_2=\l\lfloor \frac N 3\r\rfloor$, for $p>1$. In the second step of the randomisation algorithm, we set $u=\pm \sqrt 2$. In the Appendix, we provide an analysis of our results when varying $R_1$, $R_2$, $\delta^*$ and $u$, showing that results are robust to these specifications. All tests are carried out at a significance level $\alpha_1=\alpha_2=\frac{0.05}{\min(N,T)}$, which corresponds to critical values growing logarithmically with $N$ or $T$, hence satisfying the conditions in Lemmas \ref{test-1-algo} and \ref{test-2-algo}. 

To save space, here we report only the results when, in \eqref{dgp3}, we set $\bar \rho=0.4$ - results for $\bar \rho=0$ and $\bar \rho=0.8$ are in the Appendix. As an overall comment, results are in general unaffected but for two cases. The first case is when $r_1=0$ and $r_2=1$ and we compute $\wh{r}_1$; in this case, we find that lower values of $\bar \rho$ improve the results. Conversely, higher values of $\bar \rho$ make the innovations of the zero-mean $I(1)$ factors more persistent, thus making the associated eigenvalues larger: in this case, we are therefore more likely to falsely detect trends. The second case arises when $r_2=0$, and we compute $\wh{r}_2$. In this case, we find the exact opposite. This can be explained upon noting that, for lower values of $\bar\rho$, the two $I(1)$ factors become closer to two pure random walks which are highly collinear, thus making the second eigenvalue $\nu_2^{(2)}$ much smaller than the first one $\nu_2^{(1)}$: thus, in this case, we are less likely to detect the second factor. For the same reason, higher values of $\bar\rho$ make the two factors less collinear, so that then the second factor is detected more easily.

\begin{table}[t!]
\centering
\caption{Average estimated number of factors with linear trend, $\wh{r}_1$.}\label{tab:case1}
\vskip .2cm
\small{
\begin{tabular}{ll | ccc | ccc | ccc}
\hline
\hline
&& \multicolumn{3}{|c|}{$N=50$, $T=100$}&\multicolumn{3}{c|}{$N=100$, $T=100$}&\multicolumn{3}{c}{$N=200$, $T=100$}\\
$r_1\,$ & $r_2$ & $BT1$ & $BT2$ & $BT3$& $BT1$ & $BT2$ & $BT3$& $BT1$ & $BT2$ & $BT3$\\
\hline
0$\,$		&	0	&	0.00	&	0.00	&	0.00	&	0.00	&	0.00	&	0.00	&	0.00	&	0.00	&	0.00	\\
0$\,$		&	1	&	0.20	&	0.18	&	0.19	&	0.11	&	0.10	&	0.41	&	0.04	&	0.04	&	0.31	\\
0$\,$		&	2	&	0.04	&	0.05	&	0.06	&	0.02	&	0.02	&	0.13	&	0.01	&	0.00	&	0.06	\\
1$\,$		&	0	&	1.00	&	1.00	&	1.00	&	1.00	&	1.00	&	1.00	&	1.00	&	1.00	&	1.00	\\
1$\,$		&	1	&	1.00	&	1.00	&	1.00	&	1.00	&	1.00	&	1.00	&	1.00	&	1.00	&	1.00	\\
1$\,$		&	2	&	1.00	&	1.00	&	1.00	&	1.00	&	1.00	&	1.00	&	1.00	&	1.00	&	1.00	\\
\hline
\hline
&& \multicolumn{3}{|c|}{$N=100$, $T=200$}&\multicolumn{3}{c|}{$N=200$, $T=200$}&\multicolumn{3}{c}{$N=200$, $T=500$}\\
$r_1\,$ & $r_2$ & $BT1$ & $BT2$ & $BT3$& $BT1$ & $BT2$ & $BT3$& $BT1$ & $BT2$ & $BT3$\\
\hline
0$\,$		&	0	&	0.00	&	0.00	&	0.00	&	0.00	&	0.00	&	0.00	&	0.00	&	0.00	&	0.00	\\
0$\,$		&	1	&	0.06	&	0.06	&	0.30	&	0.05	&	0.04	&	0.25	&	0.02	&	0.01	&	0.14	\\
0$\,$		&	2	&	0.00	&	0.00	&	0.05	&	0.00	&	0.00	&	0.03	&	0.00	&	0.00	&	0.01	\\
1$\,$		&	0	&	1.00	&	1.00	&	1.00	&	1.00	&	1.00	&	1.00	&	1.00	&	1.00	&	1.00	\\
1$\,$		&	1	&	1.00	&	1.00	&	1.00	&	1.00	&	1.00	&	1.00	&	1.00	&	1.00	&	1.00	\\
1$\,$		&	2	&	1.00	&	1.00	&	1.00	&	1.00	&	1.00	&	1.00	&	1.00	&	1.00	&	1.00	\\
\hline
\hline
\end{tabular}
}
\end{table}

\begin{table}[t!]
\caption{Average estimated number of zero-mean $I(1)$ factors, $\wh{r}_2$, when $r_1=0$.}\label{tab:case2}
\centering
\vskip .2cm
\small{
\begin{tabular}{ll |cccc | cccc | cccc  }
\hline
\hline
&& \multicolumn{4}{|c}{$N=50$, $T=100$}& \multicolumn{4}{|c}{$N=100$, $T=100$}& \multicolumn{4}{|c}{$N=200$, $T=100$}\\
$r_2\,$ & $r_3$ & $BT1$ & $BT2$ & $BT3$& $IC$& $BT1$ & $BT2$ & $BT3$& $IC$& $BT1$ & $BT2$ & $BT3$& $IC$\\
\hline
0	&	0	&	0.00	&	0.00	&	0.00	&	1.00	&	0.00	&	0.00	&	0.00	&	1.00	&	0.00	&	0.00	&	0.00	&	1.00	\\
0	&	1	&	0.00	&	0.00	&	0.00	&	1.00	&	0.00	&	0.00	&	0.00	&	1.00	&	0.00	&	0.00	&	0.00	&	1.00	\\
0	&	2	&	0.00	&	0.00	&	0.00	&	1.00	&	0.00	&	0.00	&	0.00	&	1.00	&	0.00	&	0.00	&	0.00	&	1.00	\\
1	&	0	&	1.00	&	1.00	&	1.00	&	1.00	&	1.00	&	1.00	&	1.00	&	1.00	&	1.00	&	1.00	&	1.00	&	1.00	\\
1	&	1	&	1.00	&	1.00	&	1.00	&	1.00	&	1.00	&	1.00	&	1.00	&	1.00	&	1.00	&	1.00	&	1.00	&	1.00	\\
1	&	2	&	1.00	&	1.00	&	1.00	&	1.00	&	1.00	&	1.00	&	1.00	&	1.00	&	1.00	&	1.00	&	1.00	&	1.00	\\
2	&	0	&	1.99	&	1.99	&	1.99	&	2.00	&	1.98	&	2.00	&	1.99	&	2.00	&	1.97	&	1.99	&	1.99	&	2.00	\\
2	&	1	&	1.98	&	1.97	&	1.97	&	2.00	&	1.97	&	2.00	&	2.00	&	2.00	&	1.94	&	1.98	&	1.99	&	2.00	\\
2	&	2	&	1.98	&	1.98	&	1.98	&	2.00	&	1.99	&	2.00	&	2.00	&	2.00	&	1.98	&	2.00	&	2.00	&	2.00	\\
\hline
\hline
&& \multicolumn{4}{|c}{$N=100$, $T=200$}& \multicolumn{4}{|c}{$N=200$, $T=200$}& \multicolumn{4}{|c}{$N=200$, $T=500$}\\
$r_2\,$ & $r_3$ & $BT1$ & $BT2$ & $BT3$& $IC$& $BT1$ & $BT2$ & $BT3$& $IC$& $BT1$ & $BT2$ & $BT3$& $IC$\\
\hline
0	&	0	&	0.00	&	0.00	&	0.00	&	1.00	&	0.00	&	0.00	&	0.00	&	1.00	&	0.00	&	0.00	&	0.00	&	1.00	\\
0	&	1	&	0.00	&	0.00	&	0.00	&	1.00	&	0.00	&	0.00	&	0.00	&	1.00	&	0.00	&	0.00	&	0.00	&	1.00	\\
0	&	2	&	0.00	&	0.00	&	0.00	&	1.00	&	0.00	&	0.00	&	0.00	&	1.00	&	0.00	&	0.00	&	0.00	&	1.00	\\
1	&	0	&	1.00	&	1.00	&	1.00	&	1.00	&	1.00	&	1.00	&	1.00	&	1.00	&	1.00	&	1.00	&	1.00	&	1.00	\\
1	&	1	&	1.00	&	1.00	&	1.00	&	1.00	&	1.00	&	1.00	&	1.00	&	1.00	&	1.00	&	1.00	&	1.00	&	1.00	\\
1	&	2	&	1.00	&	1.00	&	1.00	&	1.00	&	1.00	&	1.00	&	1.00	&	1.00	&	1.00	&	1.00	&	1.00	&	1.00	\\
2	&	0	&	2.00	&	2.00	&	2.00	&	2.00	&	1.99	&	2.00	&	2.00	&	2.00	&	2.00	&	2.00	&	2.00	&	2.00	\\
2	&	1	&	2.00	&	2.00	&	2.00	&	2.00	&	2.00	&	2.00	&	2.00	&	2.00	&	2.00	&	2.00	&	2.00	&	2.00	\\
2	&	2	&	2.00	&	2.00	&	2.00	&	2.00	&	2.00	&	1.99	&	2.00	&	2.00	&	2.00	&	2.00	&	2.00	&	2.00	\\
\hline
\hline
\end{tabular}
}
\end{table}

\begin{table}[t!]
\centering
\caption{Average estimated number of zero-mean $I(1)$ factors, $\wh{r}_2$, when $r_1=1$.}\label{tab:case3}
\vskip .2cm
\small{
\begin{tabular}{ll |cccc | cccc | cccc  }
\hline
\hline
&& \multicolumn{4}{|c}{$N=50$, $T=100$}& \multicolumn{4}{|c}{$N=100$, $T=100$}& \multicolumn{4}{|c}{$N=200$, $T=100$}\\
$r_2\,$ & $r_3$ & $BT1$ & $BT2$ & $BT3$& $IC$& $BT1$ & $BT2$ & $BT3$& $IC$& $BT1$ & $BT2$ & $BT3$& $IC$\\
\hline
0	&	0	&	0.00	&	0.00	&	0.00	&	0.00	&	0.00	&	0.00	&	0.00	&	0.00	&	0.00	&	0.00	&	0.00	&	0.00	\\
0	&	1	&	0.00	&	0.00	&	0.00	&	0.00	&	0.00	&	0.00	&	0.00	&	0.00	&	0.00	&	0.00	&	0.00	&	0.00	\\
0	&	2	&	0.00	&	0.00	&	0.00	&	0.00	&	0.00	&	0.00	&	0.00	&	0.00	&	0.00	&	0.00	&	0.00	&	0.00	\\
1	&	0	&	1.00	&	1.00	&	1.00	&	1.00	&	1.00	&	0.99	&	1.00	&	1.00	&	1.00	&	1.00	&	1.00	&	1.00	\\
1	&	1	&	1.00	&	1.00	&	1.00	&	1.00	&	1.00	&	0.99	&	1.00	&	1.00	&	1.00	&	1.00	&	1.00	&	1.00	\\
1	&	2	&	0.99	&	0.98	&	1.00	&	1.00	&	1.00	&	1.00	&	1.00	&	1.00	&	1.00	&	0.99	&	0.99	&	1.00	\\
2	&	0	&	1.90	&	1.97	&	1.99	&	1.97	&	1.91	&	1.95	&	1.98	&	1.99	&	1.85	&	1.95	&	1.98	&	1.99	\\
2	&	1	&	1.64	&	1.80	&	1.91	&	1.85	&	1.64	&	1.85	&	1.94	&	1.95	&	1.55	&	1.76	&	1.91	&	1.97	\\
2	&	2	&	1.62	&	1.77	&	1.87	&	1.80	&	1.65	&	1.76	&	1.87	&	1.89	&	1.54	&	1.67	&	1.83	&	1.94	\\
\hline
\hline
&& \multicolumn{4}{|c}{$N=100$, $T=200$}& \multicolumn{4}{|c}{$N=200$, $T=200$}& \multicolumn{4}{|c}{$N=200$, $T=500$}\\
$r_2\,$ & $r_3$ & $BT1$ & $BT2$ & $BT3$& $IC$& $BT1$ & $BT2$ & $BT3$& $IC$& $BT1$ & $BT2$ & $BT3$& $IC$\\
\hline
0	&	0	&	0.00	&	0.00	&	0.00	&	0.00	&	0.00	&	0.00	&	0.00	&	0.00	&	0.00	&	0.00	&	0.00	&	0.00	\\
0	&	1	&	0.00	&	0.00	&	0.00	&	0.00	&	0.00	&	0.00	&	0.00	&	0.00	&	0.00	&	0.00	&	0.00	&	0.00	\\
0	&	2	&	0.00	&	0.00	&	0.00	&	0.00	&	0.00	&	0.00	&	0.00	&	0.00	&	0.00	&	0.00	&	0.00	&	0.00	\\
1	&	0	&	0.99	&	1.00	&	1.00	&	1.00	&	1.00	&	1.00	&	1.00	&	1.00	&	1.00	&	1.00	&	1.00	&	1.00	\\
1	&	1	&	1.00	&	1.00	&	1.00	&	1.00	&	1.00	&	1.00	&	1.00	&	1.00	&	1.00	&	1.00	&	1.00	&	1.00	\\
1	&	2	&	1.00	&	1.00	&	1.00	&	1.00	&	1.00	&	1.00	&	1.00	&	1.00	&	1.00	&	1.00	&	1.00	&	1.00	\\
2	&	0	&	1.98	&	2.00	&	1.99	&	1.99	&	1.96	&	1.99	&	2.00	&	2.00	&	1.99	&	2.00	&	2.00	&	2.00	\\
2	&	1	&	1.89	&	1.98	&	1.99	&	1.98	&	1.87	&	1.95	&	1.99	&	2.00	&	1.98	&	1.99	&	2.00	&	2.00	\\
2	&	2	&	1.89	&	1.95	&	1.98	&	1.96	&	1.85	&	1.93	&	1.96	&	1.99	&	1.97	&	1.99	&	1.99	&	2.00	\\
\hline
\hline
\end{tabular}
}
\end{table}

The tables lend themselves to drawing some general conclusions about the
main features of our methodology. First, $BT1$ and $BT2$ are usually very good at finding
no common factors - whether with a linear trend or genuinely $I\left(
1\right) $ with zero mean - when there are no common factors (see Tables \ref{tab:case1} and \ref{tab:case2}), which is also
consistent with the results in \citet{trapani17}. In the case of detecting the presence of genuine zero-mean $I(1)$ common factors, these results can be compared with the ones obtained using $IC$ which invariably finds one common $I(1)$ factor even when such factors are not present (see Table \ref{tab:case2}). Few exceptions are found in Table \ref{tab:case1}, in the case where there is no common factor with a linear trend but there is one zero-mean $I(1)$ common factor. Even in this case, both $BT1$ and $BT2$ work extremely well as $T$ increases. Note that $BT3$ also works very well in this case, at least when no common factors with linear trends are present. Conversely, when there is one common factor with a linear trend, $BT3$ tends to overestimated more when one $I(1)$ factor is present, even when $T$ is large. Second, in general our criteria tend to understate, albeit slightly, as opposed to overstate the true number of common factors, this is particularly true when estimating the number of zero-mean $I(1)$ factors in presence of linear trends (see Table \ref{tab:case3}). In any case, the bias is of the same order as the bias of $IC$ and tends to vanish as $T$ increases. Overall, the performance of all our criteria improves dramatically as $T$
increases: although results are usually good whenever $T=100$, they markedly
improve when $T\geq 200$ for all cases considered. The impact of $N$ is, in general, less clear.
%

\section{An empirical investigation of the dimensions of the yield curve}

In this section, we illustrate our methodology through an application to the High Quality Market (HQM) Corporate Bond Yield Curve, available from the Federal Reserve Economic Data (FRED)\footnote{\tt{https://fred.stlouisfed.org}.} - details on the construction of the yield curves are available from the US Department of Treasury.\footnote{\tt{https://www.treasury.gov/resource-center/economic-policy/corp-bond-yie}.}
We use monthly data on HQM Corporate Bonds with maturities from  6 months up to 100 years ($N=196$), and spanning the period from January 1985 to September 2017 ($T=393$). The data are shown in Figure \ref{fig:data}, which shows evidence of non-stationarity and co-movements both cross-sectionally and across time.
\begin{figure}[t!]
\begin{center}
\caption{HQM Corporate Bond Yield Curve}\label{fig:data}
\includegraphics[width=.95\textwidth]{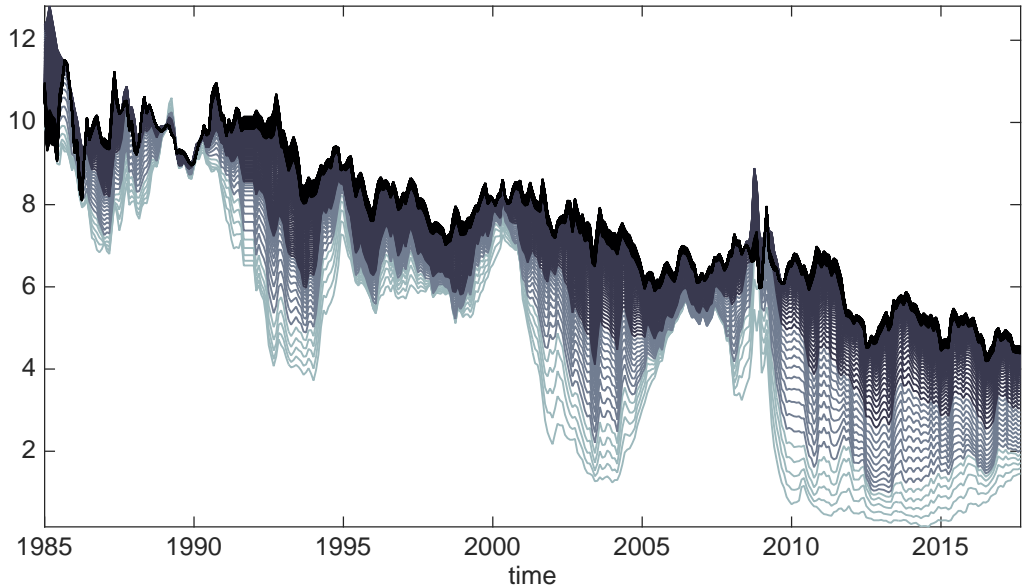}\\
\end{center}
\vskip -.5cm
\end{figure}

We use the same settings as in Section \ref{numerics}. In particular, when computing $\wh{r}_1$, we set $R_1 = N$, while for $\wh{r}^*$ we set $R_2 = N$ if $p=1$ and $R_2=\lfloor N/3\rfloor$ for $p>1$. The significance level is $\frac{0.05}{\min(N,T)}=0.0002551$. Finally, we note that when computing $\wh r$, $BT1$ is equivalent to the test by \citet{trapani17}. 

Results are in Table \ref{tab:emp}, where we have reported our three criteria, and, as a term of comparison, the information criterion $IC$ - when computing $\wh r$, this is equivalent to $IC3$ in \citet{baing02}. Based on our findings, there is borderline evidence of a common factor with a linear trend - indeed, this is picked up by $BT3$. Note that $BT3$, in view of our simulations, may have a tendency to overstate the presence of a common factor with a linear trend in small samples when there is one $I$(1), zero mean common factor. However, in our case the sample sizes are sufficiently large, and there is clear evidence of having several common factors, which suggests that $BT3$ may be correct in indicating the presence of a common factor with a linear trend. As far as the other factors are concerned, both $BT2$ and $BT3$ indicate that there are zero mean $I$(1) common factors; based on the discrepancy between these two criteria, it may be argued that two of such factors may be only borderline non-stationary. This evidence is in line with the findings from $IC$; conversely, $BT1$ seems to suggest only one common factor, which is at odds with the stylised factors in this literature where, usually, at least three factors are identified. To sum up, the results in Table \ref{tab:emp} indicate the presence of five common factors, which we estimate as the principal components of $X_t$, using the covariance $T^{-2}\sum_{t=1}^T  X_t X_t'$ and imposing the identifying constraint $\Lambda' \Lambda= N I_r$, \citep[see][]{bai04,maciejowska2010}. 

\begin{table}[t!]
\centering
\caption{Estimated number of factors in the HQM Corporate Bond Yield Curve}\label{tab:emp}
\vskip .3cm
\small{
\begin{tabular}{ l c|c|c|c|c}
\hline
\hline
&& $BT1$ & $BT2$ & $BT3$ & $IC$\\
\hline
with linear trend 	& $\wh{r}_1$ 	& 0 	& 0 	& 1 	& n.a.	\\
non-stationary  	& $\wh{r}^*$ 	& 1	& 3 	& 5 	& 5 		\\
zero-mean, $I(1)$ 	& $\wh{r}_2$ 	& 1	& 3	& 4	& n.a.	\\
all factors 				& $\wh{r}$ 	& 1	& 5	& 5	& 5		\\
zero-mean, $I(0)$  	& $\wh{r}_3$ 	& 0	& 2	& 0	& 0		\\
\hline
\hline
\end{tabular}
}
\end{table}

The estimated factors are shown in Figure \ref{fig:ident} (solid red lines). The first three factors appear to be non-stationary; in particular, as indicated by $BT3$, the first one does seem to be driven by a linear trend. This evidence is consistent with our findings in Table \ref{tab:emp} (save for $BT1$), and it implies that the first three factors are highly persistent. 
In Figure \ref{fig:acf} we report the autocorrelation of each estimated factor and the median, 5th and 95th percentiles of the autocorrelations of the idiosyncratic errors together with 95\% confidence bands (dashed lines) computed as $\pm \frac{1.96}{\sqrt T}$. These results suggest that the fourth and fifth factor are nearly stationary, whilst the idiosyncratic component is clearly stationary since it shows no residual autocorrelation. The presence of common unit roots, and the stationarity of the idiosyncratic error imply cointegration, which in turn implies the factor structure in bond yields - see \citet{dungey}. 

 \begin{figure}[t!]
 \centering
 \caption{Estimated and identified common factors $\wh{\mathcal F}_{j,t}$ with proxies.}\label{fig:ident}
 \begin{tabular}{cc}
\includegraphics[width=.45\textwidth]{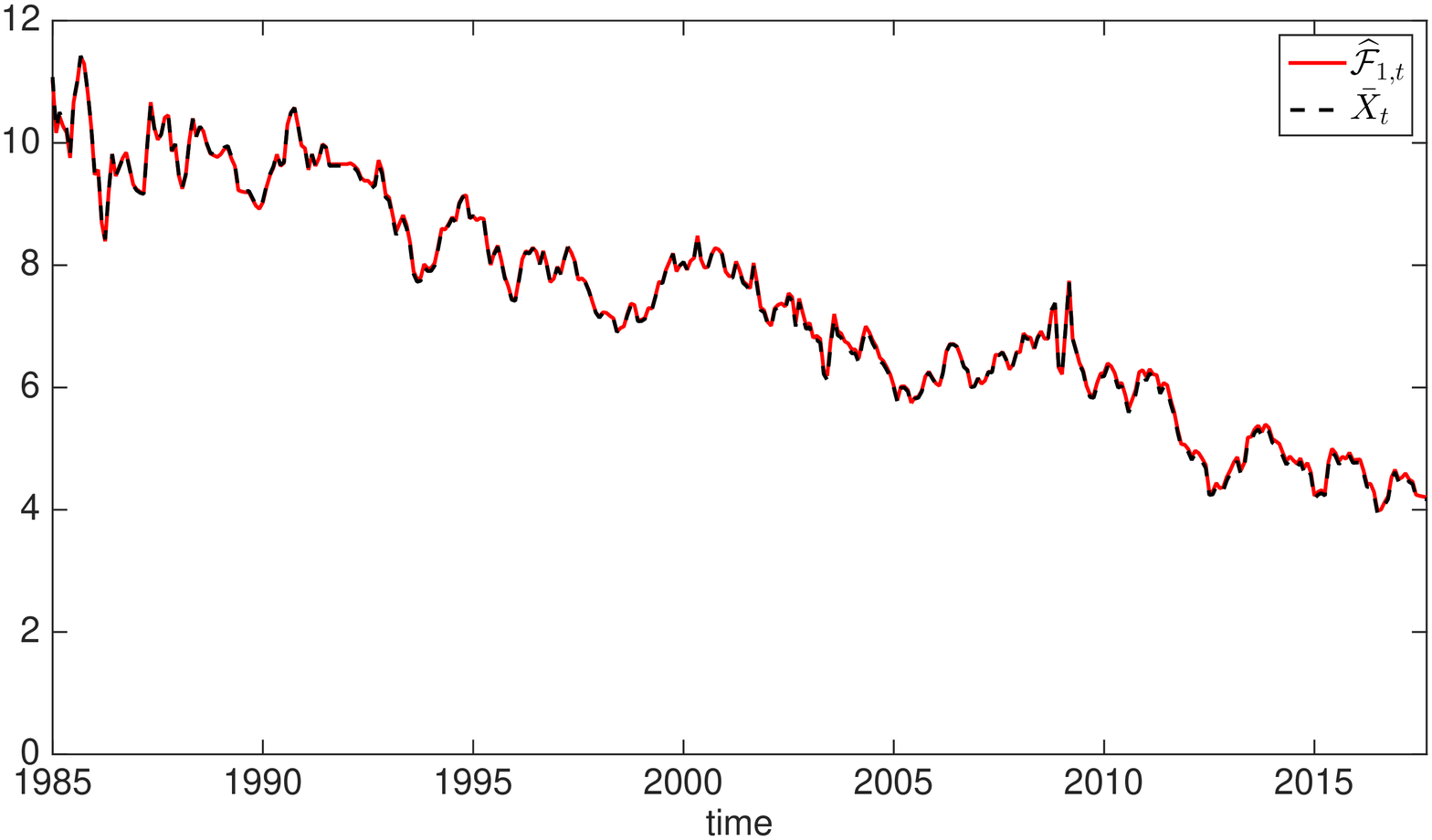}&\includegraphics[width=.45\textwidth]{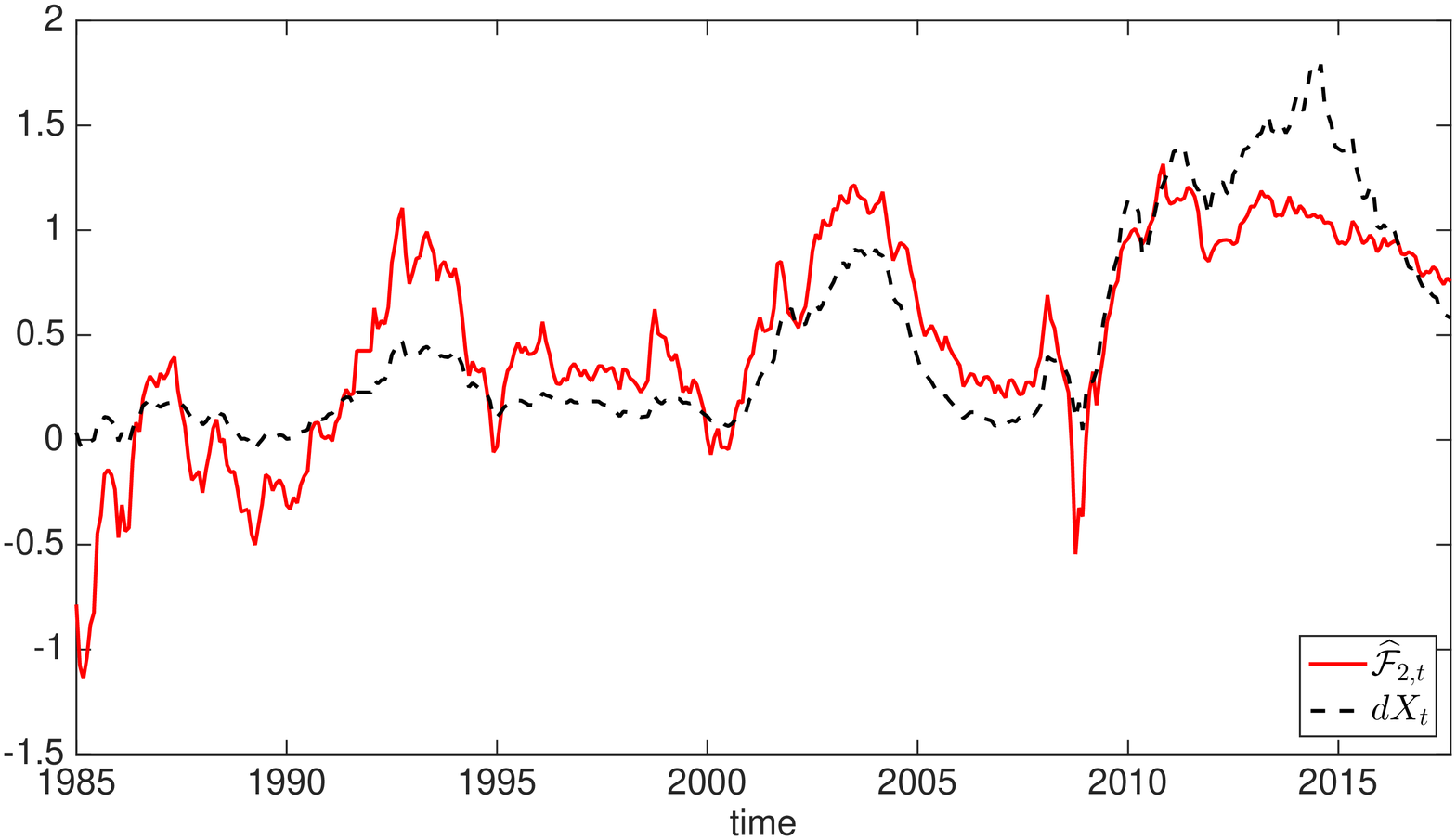}\\
\small{$\wh{\mathcal F}_{1,t}$}&\small{$\wh{\mathcal F}_{2,t}$}\\
\includegraphics[width=.45\textwidth]{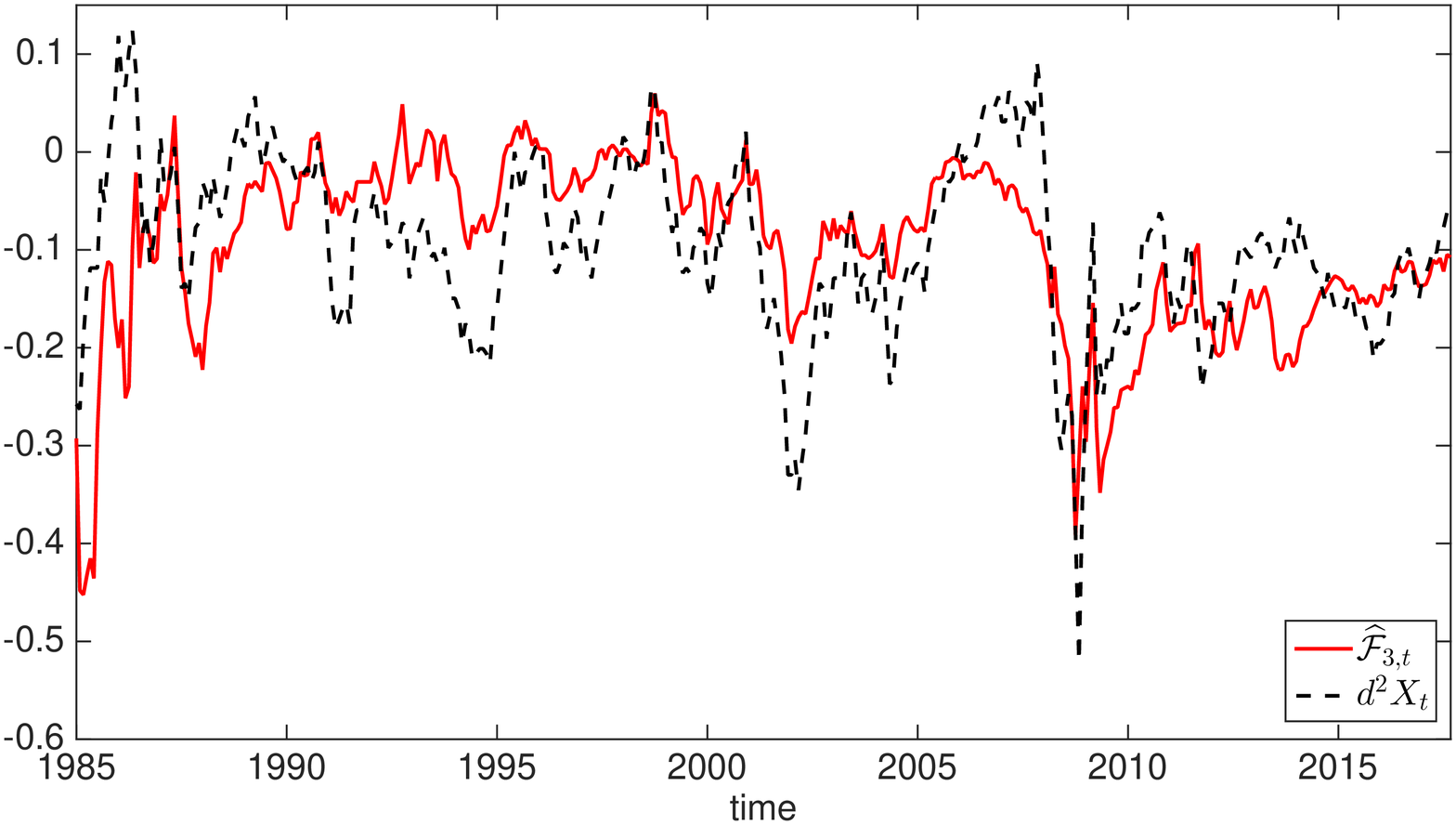}&\includegraphics[width=.45\textwidth]{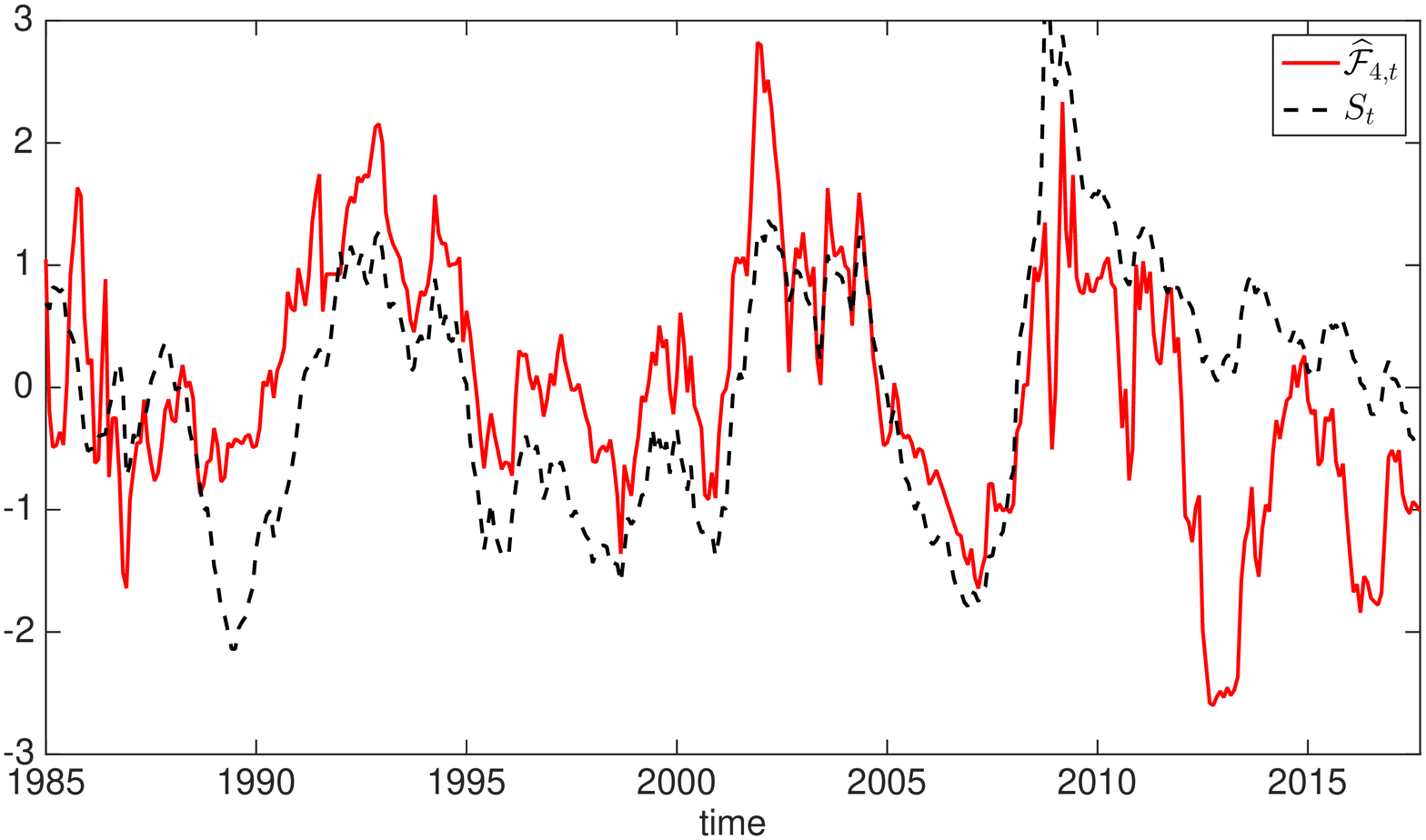}\\
\small{$\wh{\mathcal F}_{3,t}$}&\small{$\wh{\mathcal F}_{4,t}$}\\
 \end{tabular} 
 \begin{tabular}{c}
\includegraphics[width=.45\textwidth]{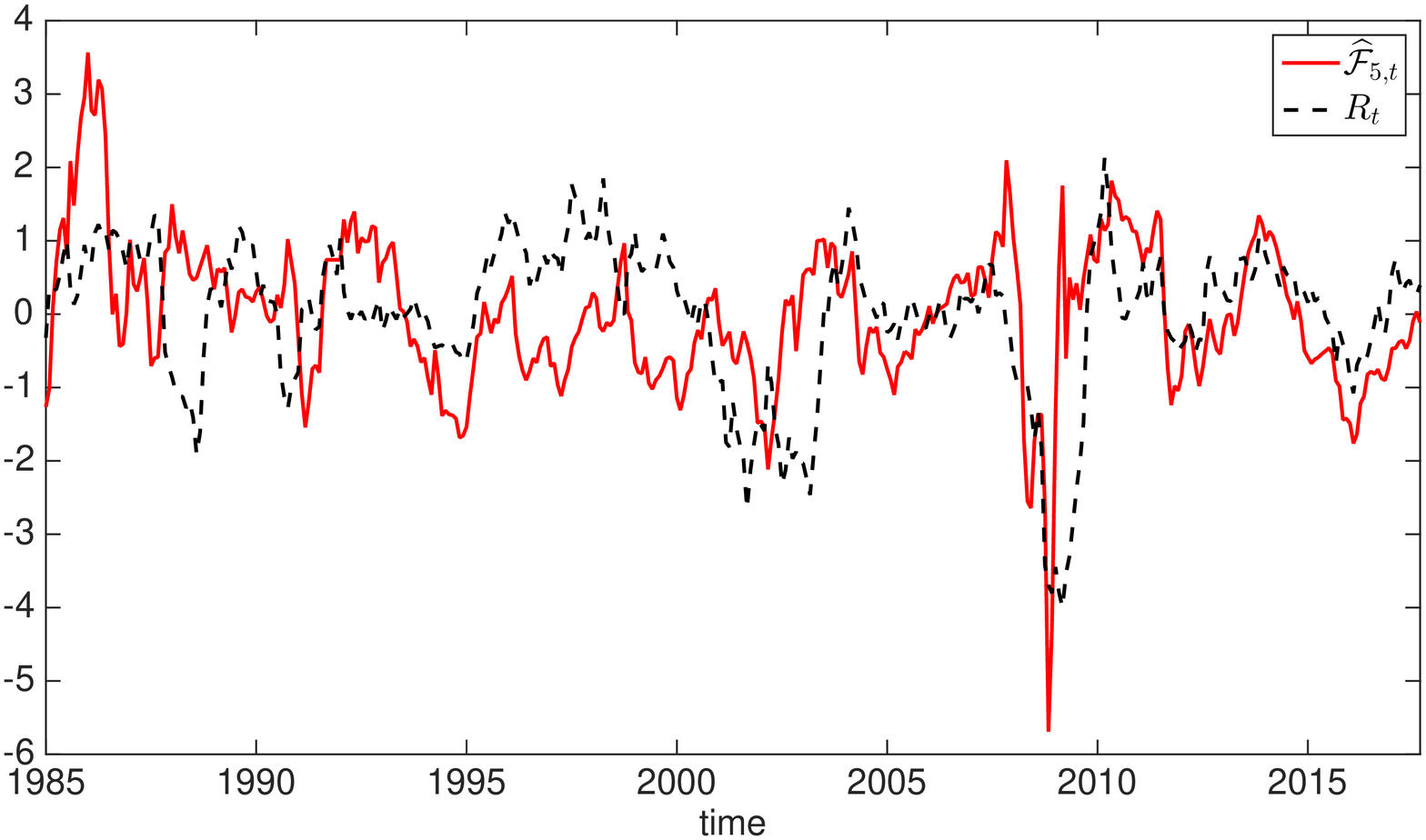}\\
\small{$\wh{\mathcal F}_{5,t}$}
 \end{tabular} 
 
 \end{figure}

 \begin{figure}[t!]
 \centering
 \caption{Autocorrelation of estimated common factors $\wh{\mathcal F}_{j,t}$ and idiosyncratic errors $\wh{u}_{i,t}$.}\label{fig:acf}
 \begin{tabular}{cc}
\includegraphics[width=.45\textwidth]{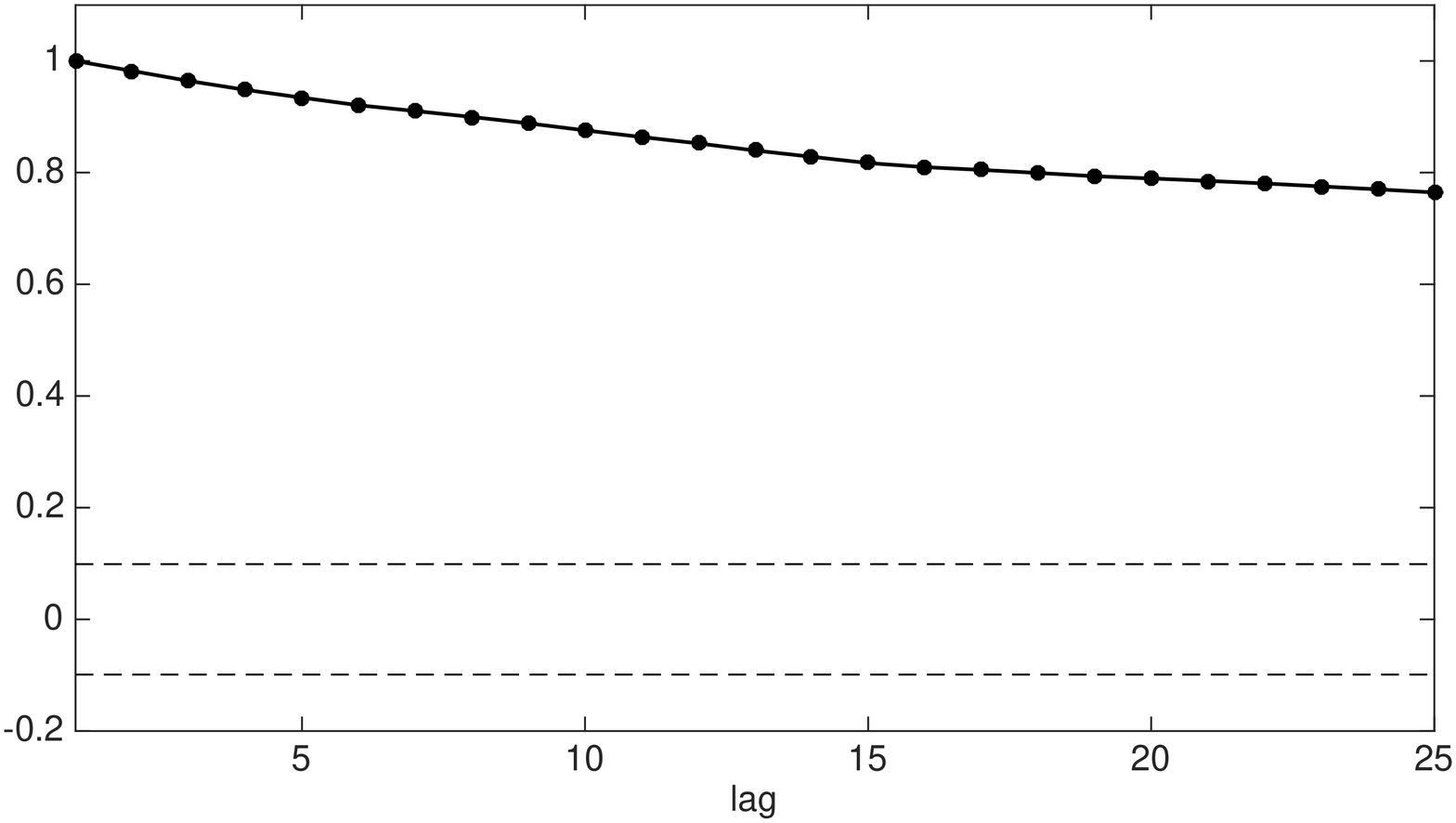}&\includegraphics[width=.45\textwidth]{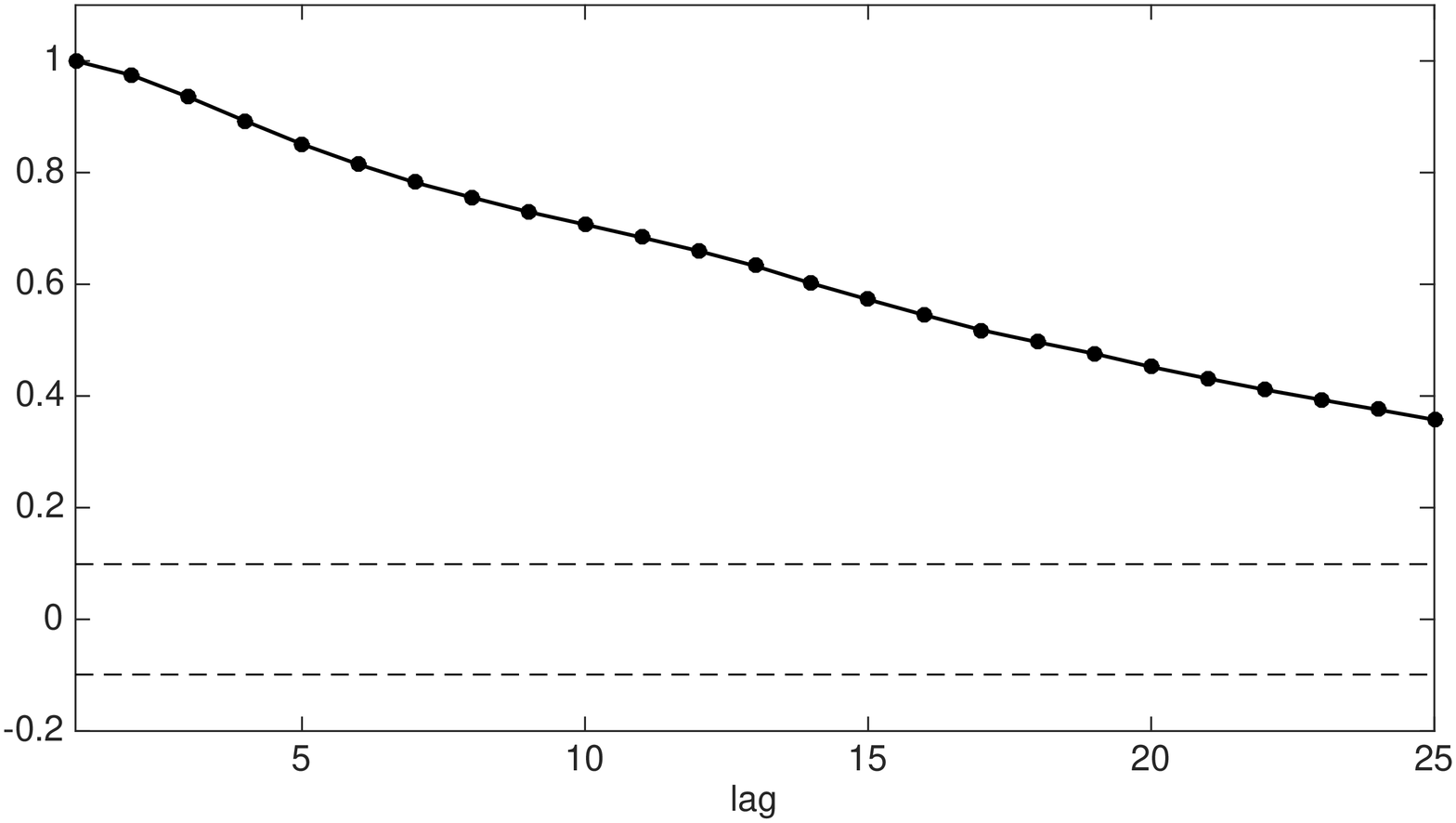}\\
\small{$\wh{\mathcal F}_{1,t}$}&\small{$\wh{\mathcal F}_{2,t}$}\\
\includegraphics[width=.45\textwidth]{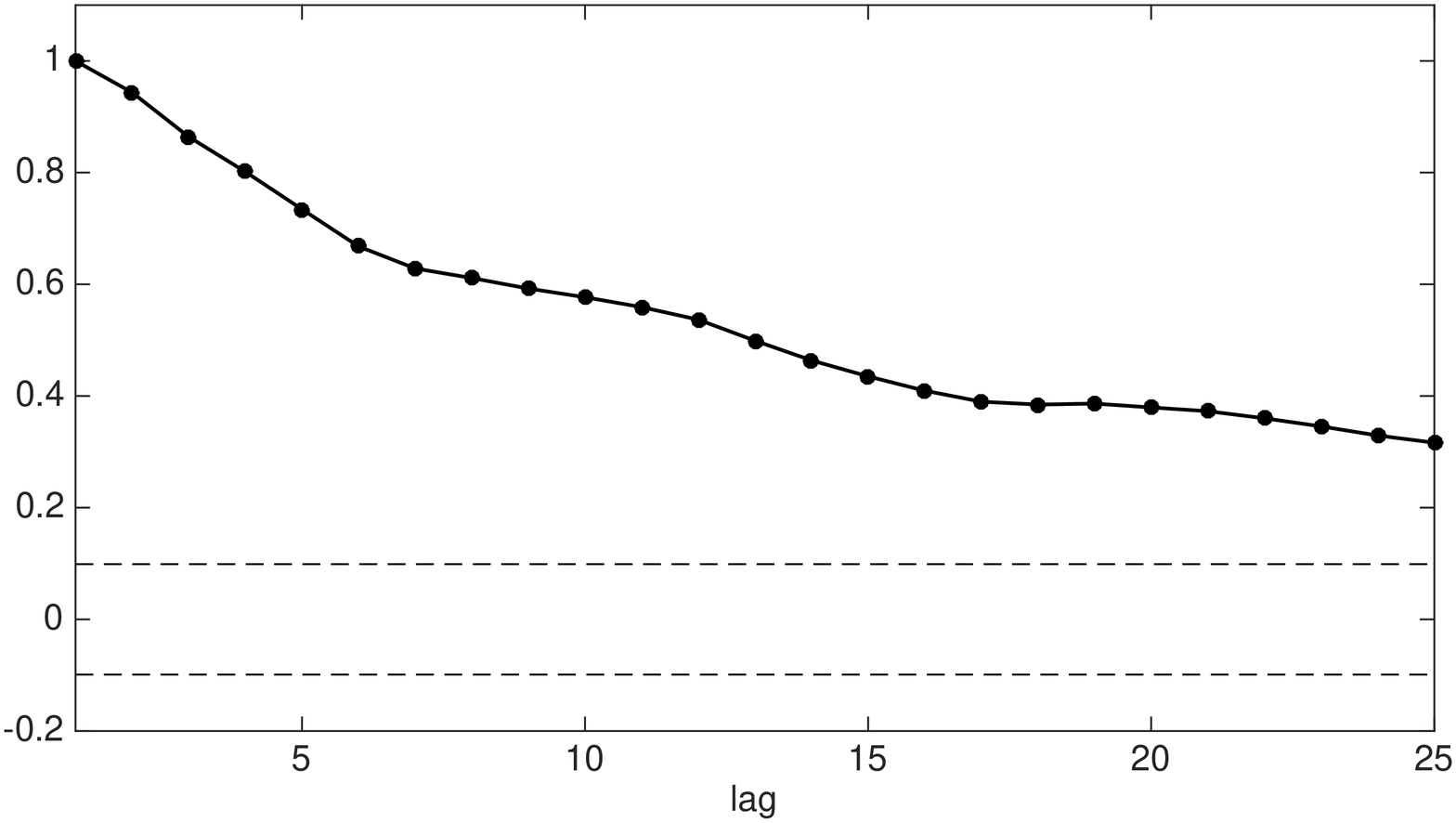}&\includegraphics[width=.45\textwidth]{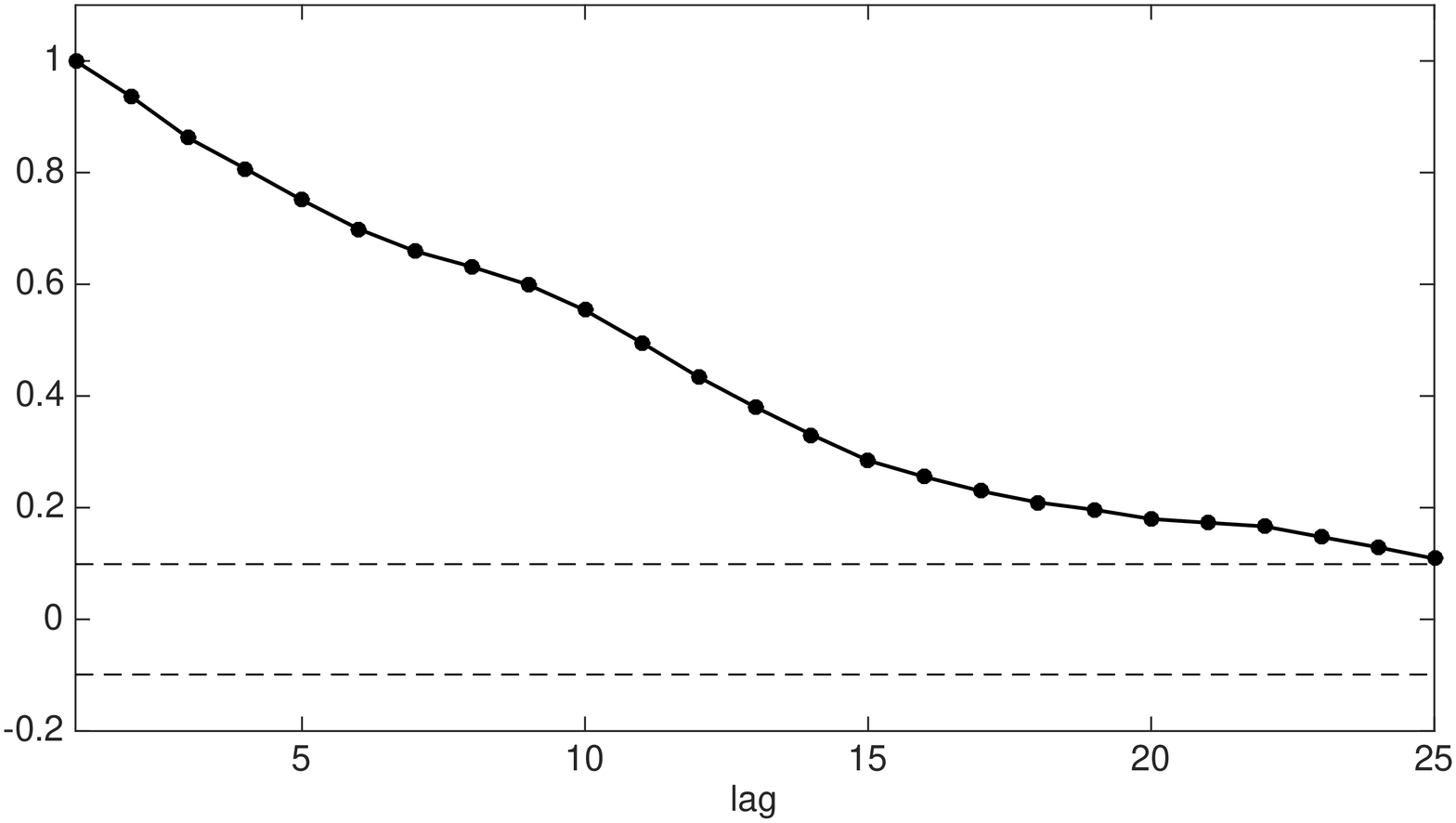}\\
\small{$\wh{\mathcal F}_{3,t}$}&\small{$\wh{\mathcal F}_{4,t}$}\\
\includegraphics[width=.45\textwidth]{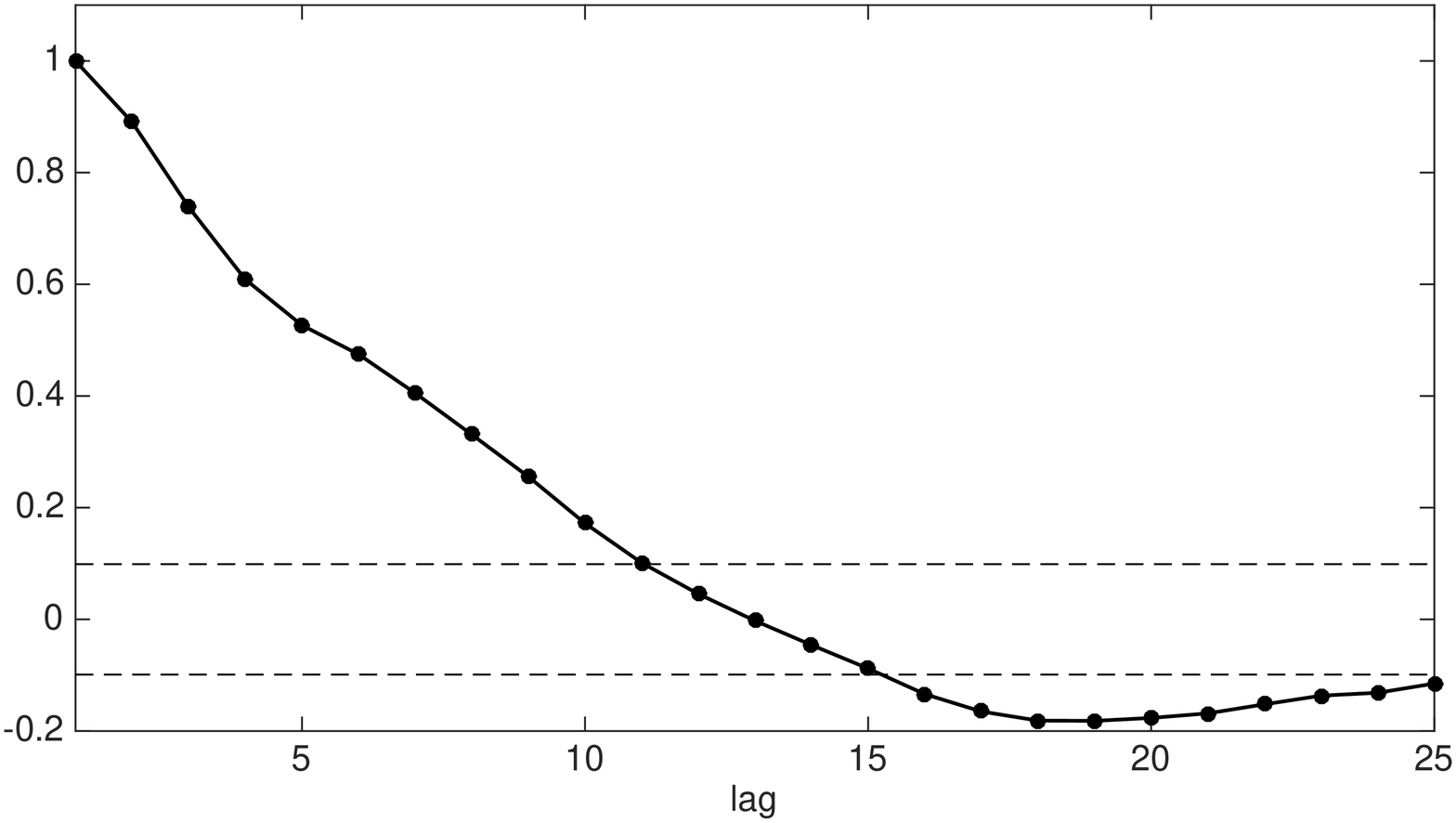}&\includegraphics[width=.45\textwidth]{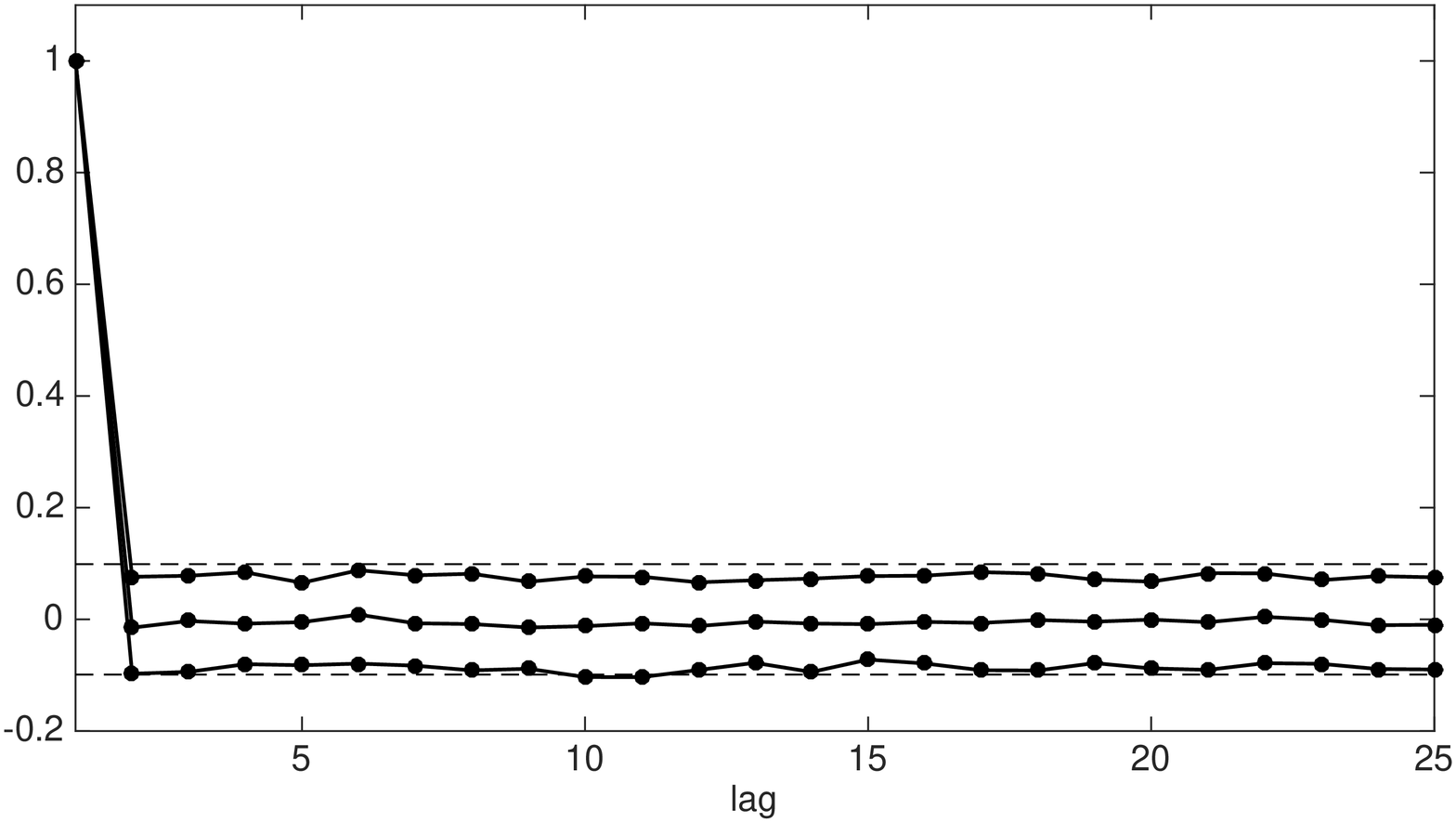}\\
\small{$\wh{\mathcal F}_{5,t}$}&\small{$\wh{u}_{i,t}$}
 \end{tabular}  
 \end{figure}

Our findings can be contrasted with the stylised facts which are typically found in this literature. In particular, following \citet{nelson1987}, it is common to model yield curves by means of three common factors, which are usually interpreted as the level, slope, and curvature of the yield curve in a given time period $t$ -- see for example \citet{dai2000} and \citet{diebold2006}. Moreover, when considering corporate bonds it common to find additional factors beyond the classical first three -- see for example \citet{duffie1999modeling}, \citet{duffie2007multi}, and \citet{clopez2008}.

First we analyse the first three estimated common factors. Throughout we assume that at each point in time $t$, the $N$ elements of $X_t$ are ordered according to their maturity, thus $X_{1,t}$ is the shortest maturity (6 months), while $X_{N,t}$ is the longest maturity (100 years). We compare each estimated factor with a standard proxy as specified by \citet{diebold2006b}. Results are in the first three panels of Figure \ref{fig:ident}, where we show both the estimated factors  (solid red lines) and the proxies (dashed black lines). In particular, in order to identify $\wh{\mathcal{F}}_{1,t}$, we consider the proxy $\bar X_t=N^{-1}\sum_{i=1}^NX_{i,t}$; we found that $\mbox{Corr}(\bar X_t,\wh{\mathcal {F}}_{1,t})\simeq 1$, which strongly suggests that $\wh{\mathcal {F}}_{1,t}$ can be viewed as the \textit{level} of the curve. Turning to $\wh{\mathcal {F}}_{2,t}$, we use, as a proxy for the slope, $dX_t = N^{-1}\sum_{i=2}^N(\ln X_{i,t}-\ln X_{i-1,t}) = N^{-1}(\ln X_{N,t}-\ln X_{1,t})$. We find that $\mbox{Corr}(dX_t,\wh{\mathcal {F}}_{2,t})=.82$, which suggests that $\wh{\mathcal {F}}_{2,t}$ can be interpreted as the \textit{slope} of the term structure. Finally, we compare $\wh{\mathcal F}_{3,t}$ to $d^2X_t = (N-2)^{-1}\sum_{i=2}^{N-1}( X_{i+1,t}-2 X_{i,t}+ X_{i-1,t})$ as a proxy for the curvature; we find $\mbox{Corr}(d^2X_t,\wh{\mathcal {F}}_{3,t})=.53$, which shows some evidence that $\wh{\mathcal F}_{3,t}$ can be interpreted as the \textit{curvature}. Furthermore, according to \citet{diebold2006}, the first three elements of the $i$-th row of the loadings matrix $\Lambda$ should be given by
\beq\label{eq:load}
\lambda_{i,1}(c) = 1,\qquad \lambda_{i,2}(c) = \l(\frac {1 - e^{-c i}}{ci}\r), \qquad \lambda_{i,3}(c) = \l(\frac {1 - e^{-c i}}{ci}-e^{-ci}\r), 
\eeq
for some $c>0$. To confirm this finding, in Figure \ref{fig:load}, we plot the estimated loadings $(\widehat{\lambda}_{i,1},\widehat{\lambda}_{i,2},\widehat{\lambda}_{i,3})$ (left panel) together with the theoretical curves in \eqref{eq:load} computed for $c=0.2$. 
 
 \begin{figure}[t!]
 \centering
 \caption{Estimated and theoretical factor loadings.}\label{fig:load}
 \begin{tabular}{cc}
\includegraphics[width=.45\textwidth]{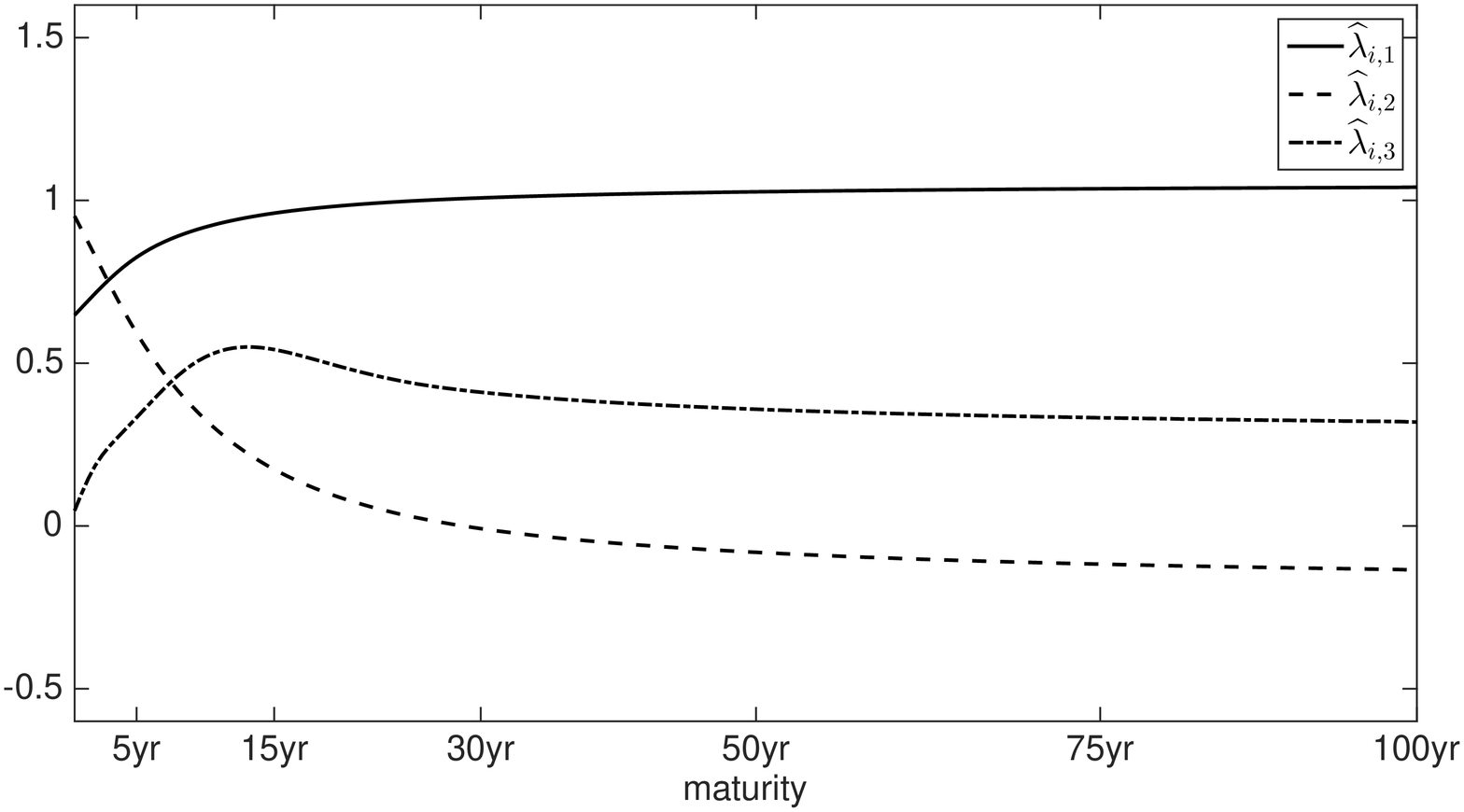}&\includegraphics[width=.45\textwidth]{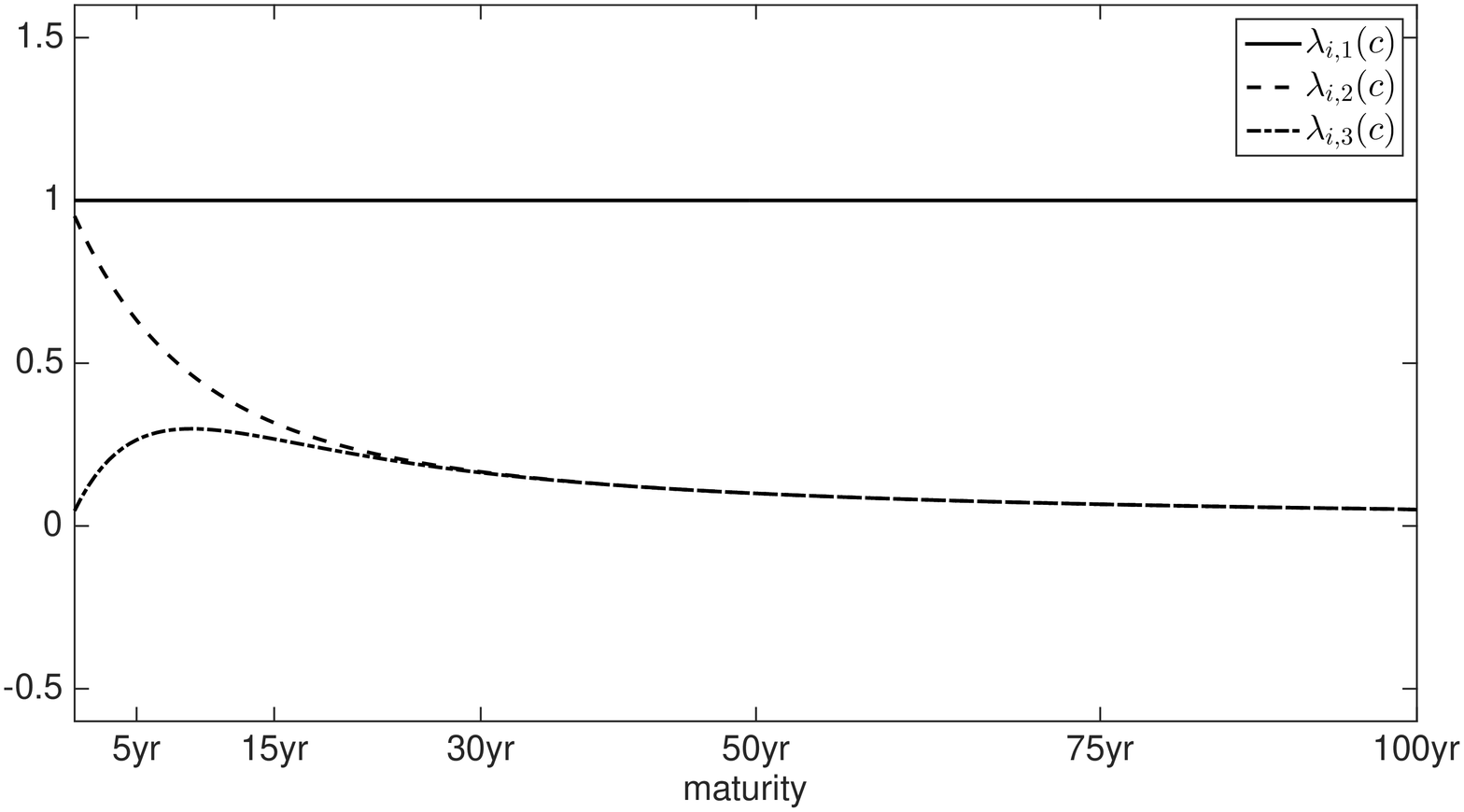}\\
\small{$\widehat{\lambda}_{i,j}$, $j=1,2,3$}&\small{$\lambda_{i,j}(c)$, $j=1,2,3$ and $c=0.2$}\\
 \end{tabular} 
\end{figure}

As far as the remaining two estimated common factors are concerned, we note that, in addition to level, slope and curvature, macroeconomic and financial factors have also been incorporated in the study of yield curves -- see for example  \citet{estrella1998}, \citet{ang2003}, \citet{diebold2006b}, \citet{duffie2007multi}, and \citet{coroneo2016unspanned}. We evaluate the correlation between $S_t$ -  the spread between the 10 years HQM bond rate and the Federal Funds rate - and the fourth factor finding that $\mbox{Corr}(S_t,\wh{\mathcal {F}}_{4,t})=.51$, whence we propose to interpret $\wh{\mathcal F}_{4,t}$ as the \textit{spread factor}. Also, letting $R_t$ be the yearly returns of the Standard \& Poor's index, we have $\mbox{Corr}(R_t,\wh{\mathcal {F}}_{5,t})=.30$; this seems to suggest that $\wh{\mathcal F}_{5,t}$ may be viewed as a \textit{financial factor}, or that, at a minimum, $\wh{\mathcal F}_{5,t}$ is intimately related to the financial market.\footnote{Data for $S_t$ are available at {\tt https://fred.stlouisfed.org}.\\ Data for $R_t$ are available at {\tt http://www.econ.yale.edu/\textasciitilde shiller/data.htm}.} These results are in line with the results by \citet{duffie2007multi}. In the last two panels of Figure \ref{fig:ident} we report the fourth an fifth estimated factors (solid red lines) and the corresponding proxies (dashed black lines).

\section{Conclusions\label{conclusions}}

In this paper, we propose a methodology to estimate the dimension of the
common factor space for a given dataset $X_{i,t}$. We do not assume that the
data are stationary or that they have (or not) linear trends: our procedure
estimates separately the number of common factors with a linear trend (which
can be only 0 or 1), the number of zero mean, $I\left( 1\right) $ common
factors, and the number of zero mean, $I\left( 0\right) $ common factors.

Since estimation of these dimensions is carried out via testing
(as opposed to using an information criterion or some other diagnostic), the
results provide several interesting interpretations. For example, having $%
r_{1}=0$ means that the data have been tested for the presence of common
linear trends, and none has been found; finding $r^*=0$ indicates that the data have been tested for (the null of)
non-stationarity, and have been found to be stationary; etc. 
Our methodology thus complements
the results recently derived by \citet{zhangpangao}.

Technically, our approach exploits the well-known eigenvalue separation
property that characterises the covariance matrix of data with a common
factor structure: essentially, the eigenvalues associated to common factors
diverge to positive infinity, whereas the other ones are bounded. On top of
this, we exploit the also well-known fact that linear trends, unit roots and
stationary processes all imply different rates of divergence of the
eigenvalues: these two facts allow us not merely to check whether there are
common factors (and how many these are) but also to discriminate between
those that have a trend, those that have a unit root, and the stationary
ones. In this respect, our procedure is akin to the one proposed by %
\citet{bai04} and \citet{zyr18}, although it is based on tests rather than an information
criterion, and it entertains the possibility that linear trends could be
present.

Several interesting issues, in the analysis of high dimensional, possibly
non-stationary, time series, remain outstanding. In model (\ref{model}), we
made no attempt to allow for the idiosyncratic components, $u_{i,t}$, to be
non-stationary, thus relegating all the possible non-stationarity to the
common component $\widetilde{\Lambda }\widetilde{F}_{t}$. Still, it would be
worth considering the case where $u_{i,t}\sim I\left( 1\right) $ for at
least some $i$, so as to be able to disentangle common and idiosyncratic
sources of non-stationarity. In such a case, our procedure would not be
immediately applicable, since, chiefly, (\ref{lambda-x-small}) would no
longer hold. Also, by proper rescaling of the covariance matrix, our approach can be readily generalized to $I(d)$ factors with $d\ge 1$. These extensions are under current investigation by the authors.


{{\ 
\bibliographystyle{chicago}
\bibliography{BT_biblio}
}}

\newpage
\appendix
\small
\section{Proofs}
\subsection{Preliminary lemmas}


Henceforth, $\nu ^{\left( p\right) }\left( A\right) $ represent the
eigenvalues, sorted in decreasing order, for a matrix $A$; we occasionally
employ the notation $\nu ^{\left( \min \right) }\left( A\right) $ to
denote the smallest eigenvalue of $A$. Also, \textquotedblleft $\overset{D}{=%
}$\textquotedblright\ denotes equality in distribution. We also use the
following matrix notation 
\begin{eqnarray*}
X_{t} &=&\Lambda ^{\left( 1\right) }f_{t}^{\left( 1\right) }+u_{t}^{\left(
1\right) } \\
&=&\Lambda ^{\left( 1\right) }f_{t}^{\left( 1\right) }+\Lambda ^{\left(
2\right) }f_{t}^{\left( 2\right) }+u_{t}^{\left( 2\right) } \\
&=&\Lambda ^{\left( 1\right) }f_{t}^{\left( 1\right) }+\Lambda ^{\left(
2\right) }f_{t}^{\left( 2\right) }+\Lambda ^{\left( 3\right) }f_{t}^{\left(
3\right) }+u_{t}.
\end{eqnarray*}%
As far as the notation is concerned, $\Lambda ^{\left( 1\right) }$ is $%
N\times r_{1}$; $\Lambda ^{\left( 2\right) }$ is $N\times r_{2}$; and,
finally, $\Lambda ^{\left( 3\right) }$ is $N\times r_{3}$.

We begin with the following lemma, which is useful to derive almost sure
rates.

\begin{lemma}
\label{borelcantelli} Consider a multi-index random variable $%
U_{i_{1},...,i_{h}}$, with $1\leq i_{1}\leq S_{1}$, $1\leq i_{2}\leq S_{2}$,
etc... Assume that%
\begin{equation}
\sum_{i_{1}}\cdot \cdot \sum_{i_{h}}\frac{1}{S_{1}\cdot ...\cdot S_{h}}%
P\left( \max_{1\leq i_{1}\leq S_{1},...,1\leq i_{h}\leq S_{h}}\left\vert
U_{i_{1},...,i_{h}}\right\vert >\epsilon L_{S_{1},...,S_{h}}\right) <\infty ,
\label{bc-1}
\end{equation}%
for some $\epsilon >0$ and a sequence $L_{S_{1},...,S_{h}}$ defined as%
\begin{equation*}
L_{S_{1},...,S_{h}}=S_{1}^{d_{1}}\cdot ...\cdot S_{h}^{d_{h}}l_{1}\left(
S_{1}\right) \cdot ...l_{h}\left( S_{h}\right) ,
\end{equation*}%
where $d_{1}$, $d_{2}$, etc. are non-negative numbers and $l_{1}\left( \cdot
\right) $, $l_{2}\left( \cdot \right) $, etc. are slowly varying functions
in the sense of Karamata. Then it holds that%
\begin{equation}
\lim \sup_{\left( S_{1},...,S_{h}\right) \rightarrow \infty }\frac{%
\left\vert U_{S_{1},...,S_{h}}\right\vert }{L_{S_{1},...,S_{h}}}=0\text{ 
\textit{a.s.}}  \label{bc-2}
\end{equation}
\end{lemma}

\begin{proof}
The proof follows similar arguments as the proof of Lemma 2 in%
\citet{trapani17} - see also \citet{cai2006}. We begin by noting that, for every $h$-tuple $\left(
S_{1},...,S_{h}\right) $, there is a $h$-tuple of integers $\left(
k_{1},...,k_{h}\right) $ such that $2^{k_{1}}\leq S_{1}<2^{k_{1}+1}$, $%
2^{k_{2}}\leq S_{2}<2^{k_{2}+1}$, etc. Similarly, there is a $h$-tuple of
real numbers defined over $\left[ 0,1\right) $, say $\left( \rho
_{1},...,\rho _{h}\right) $, such that $2^{k_{1}+\rho _{1}}=S_{1}$, $%
2^{k_{2}+\rho _{2}}=S_{2}$, etc. Consider now the short-hand notation%
\begin{equation*}
L_{k_{1},...,k_{h}}=\left( 2^{k_{1}+1}\right) ^{d_{1}}\cdot ...\cdot \left(
2^{k_{h}+1}\right) ^{d_{h}}l_{1}\left( S_{1}\right) \cdot ...l_{h}\left(
S_{h}\right) ,
\end{equation*}%
\begin{equation*}
P\left( \max_{1\leq i_{1}\leq S_{1},...,1\leq i_{h}\leq S_{h}}\left\vert
U_{i_{1},...,i_{h}}\right\vert >\epsilon L_{k_{1},...,k_{h}}\right)
=P_{k_{1},...,k_{h}};
\end{equation*}%
by (\ref{bc-1}), we have 
\begin{equation*}
\sum_{k_{1}=0}^{\infty }\cdot ...\cdot \sum_{k_{h}=0}^{\infty }\frac{%
2^{k_{1}+1}\cdot ...\cdot 2^{k_{h}+1}}{\left( 2^{k_{1}+1}-1\right) \cdot
...\cdot \left( 2^{k_{h}+1}-1\right) }P_{k_{1},...,k_{h}}<\infty .
\end{equation*}%
This, in turn, entails that%
\begin{equation*}
\sum_{k_{1}=0}^{\infty }\cdot ...\cdot \sum_{k_{h}=0}^{\infty
}P_{k_{1},...,k_{h}}\leq \sum_{k_{1}=0}^{\infty }\cdot ...\cdot
\sum_{k_{h}=0}^{\infty }\frac{2^{k_{1}+1}\cdot ...\cdot 2^{k_{h}+1}}{\left(
2^{k_{1}+1}-1\right) \cdot ...\cdot \left( 2^{k_{h}+1}-1\right) }%
P_{k_{1},...,k_{h}}<\infty ;
\end{equation*}%
thus, by the Borel-Cantelli Lemma%
\begin{equation*}
\frac{\max_{1\leq i_{1}\leq S_{1},...,1\leq i_{h}\leq S_{h}}\left\vert
U_{i_{1},...,i_{h}}\right\vert }{L_{S_{1},...,S_{h}}}\rightarrow 0\text{ a.s.%
}
\end{equation*}%
Therefore we have%
\begin{eqnarray*}
\frac{\left\vert U_{S_{1},...,S_{h}}\right\vert }{L_{S_{1},...,S_{h}}} &\leq
&\frac{\max_{1\leq i_{1}\leq S_{1},...,1\leq i_{h}\leq S_{h}}\left\vert
U_{i_{1},...,i_{h}}\right\vert }{L_{k_{1},...,k_{h}}}\frac{%
L_{k_{1},...,k_{h}}}{L_{S_{1},...,S_{h}}} \\
&\leq &C\frac{\max_{1\leq i_{1}\leq S_{1},...,1\leq i_{h}\leq
S_{h}}\left\vert U_{i_{1},...,i_{h}}\right\vert }{L_{k_{1},...,k_{h}}}%
\rightarrow 0\text{ a.s.,}
\end{eqnarray*}%
which, finally, implies (\ref{bc-2}). 
\end{proof}

Let now $\gamma ^{\left( p\right) }$ and $\omega ^{\left( p\right) }$ denote
the $p$-th largest eigenvalues of $\Lambda T^{-1}\sum_{t=1}^{T}E\left(
\Delta f_{t}\Delta f_{t}^{\prime }\right) \Lambda ^{\prime }$ and $%
T^{-1}\sum_{t=1}^{T}E\left( \Delta u_{t}\Delta u_{t}^{\prime }\right) $
respectively. By Assumption \ref{as-4}, it can be easily verified using the
arguments in the proof of Lemma 1 in \citet{trapani17} that $\gamma ^{\left( p\right) }=C_{p}N$ for $%
1\leq p\leq r$; $\omega ^{\left( 1\right) }\leq C_{1}$; and $\lim
\inf_{N\rightarrow \infty }\omega ^{\left( N\right) }>0$.

We will often need the following lemma, shown in \citet{trapani17} (see
Lemma A1), which we report here for convenience.

\begin{lemma}
\label{trapani2017}Under Assumption \ref{as-4}, it holds that, as $\min
\left( N,T\right) \rightarrow \infty $%
\begin{eqnarray*}
\lim \sup_{N\rightarrow \infty }{\overline{\nu }}_{3,p}(k) &=&\overline{%
\nu }_{3,p}^{U}(k)<\infty , \\
\lim \inf_{N\rightarrow \infty }{\overline{\nu }}_{3,p}(k) &=&\overline{%
\nu }_{3,p}^{L}(k)>0,
\end{eqnarray*}%
for every $p$ and $k$, where ${\overline{\nu }}_{3,p}(k)$ is defined in equation (\ref{v3-rescale}).
\end{lemma}

\begin{proof}
We begin by showing that
\begin{eqnarray}
\lim \sup_{N\rightarrow \infty }\frac{1}{N-k+1 }%
\sum_{h=k}^{N}\nu _{3}^{\left( h\right) } &=&\overline{\nu }%
_{3,p}^{U}(k)<\infty ,  \label{limsup} \\
\lim \inf_{N\rightarrow \infty }\frac{1}{N-k+1 }%
\sum_{h=k}^{N}\nu _{3}^{\left( h\right) } &=&\overline{\nu }_{3,p}^{L}(k)>0.
\label{liminf}
\end{eqnarray}%
Letting 
\begin{equation*}
\overline{\nu }_{3,p}(k)=\frac{1}{N-k+1 }\sum_{h=k}^{N}\nu
_{3}^{\left( h\right) },
\end{equation*}%
note that, by Weyl's inequalities, we have $\gamma ^{\left( h\right)
}+\omega ^{\left( N\right) }\leq \nu _{3}^{\left( h\right) }\leq \gamma
^{\left( h\right) }+\omega ^{\left( 1\right) }$. Thus%
\begin{equation}
\omega ^{\left( N\right) }+\frac{1}{N-k+1 }\sum_{h=k}^{N}%
\gamma ^{\left( h\right) }\leq \overline{\nu }_{3,p}(k)\leq \omega ^{\left(
1\right) }+\frac{1}{N-k+1 }\sum_{h=k}^{N}\gamma ^{\left(
h\right) }.  \label{weyl}
\end{equation}%
Assumption \ref{as-4} implies that 
\begin{equation*}
0\leq \frac{1}{N-k+1 }\sum_{h=k}^{N}\gamma ^{\left( h\right)
}\leq C_{k+1}<\infty ,
\end{equation*}%
so that (\ref{weyl}) becomes 
\begin{equation*}
\omega ^{\left( N\right) }\leq \overline{\nu }_{3,p}(k)\leq C_{0}+C_{k+1},
\end{equation*}%
whence (\ref{limsup}) and (\ref{liminf}) follow for each $k$. Hereafter, the
proof is exactly the same as that of Lemma A1 in \citet{trapani17} and thus
omitted.
\end{proof}

\begin{lemma}
\label{LIL}Under Assumption \ref{as-1}, it holds that%
\begin{equation*}
\lim \inf_{T\rightarrow \infty }\frac{\ln \ln T}{T^{2}}\sum_{t=1}^{T}f_{t}^{%
\ast }f_{t}^{\ast \prime }=D\text{ \ \ a.s.},
\end{equation*}%
where $D$ is a positive definite matrix of dimension $[r_2 + r_1(1 - d_1)d_2] \times [r_2 + r_1(1 - d_1)d_2]$.
\end{lemma}

\begin{proof} 
We have
\begin{eqnarray*}
f_{t}^{\ast }f_{t}^{\ast \prime } &=&\left( f_{t}^{\ast }\pm \Sigma _{\Delta f*}^{1/2}W\left( t\right) \right) \left( f_{t}^{\ast }\pm \Sigma _{\Delta f*}^{1/2}W\left( t\right) \right) ^{\prime } \\
&=&\Sigma _{\Delta f^*}^{1/2}W\left( t\right) W\left( t\right) ^{\prime
}\Sigma _{\Delta f^*}^{1/2}+\Sigma _{\Delta f^*}^{1/2}W\left( t\right) \left( f_{t}^{\ast }-\Sigma_{\Delta f^*}^{1/2}W\left( t\right) \right)^{\prime } +\left(
f_{t}^{\ast }-\Sigma _{\Delta f^*}^{1/2}W\left( t\right) \right) W\left( t\right) ^{\prime } \Sigma _{\Delta f^*}^{1/2}\\
&&+\left( f_{t}^{\ast }-\Sigma _{\Delta f^*}^{1/2}W\left( t\right) \right)
\left( f_{t}^{\ast }-\Sigma _{\Delta f^*}^{1/2}W\left( t\right) \right)
^{\prime }. \\
&=&I+II+III+IV.
\end{eqnarray*}%
Let $b$ be a nonzero vector of dimension $r_{1}+r_{2}$, such that $%
\left\Vert b\right\Vert <\infty $.\ We will prove that 
\begin{equation*}
\lim \inf_{T\rightarrow \infty }\frac{\ln \ln T}{T^{2}}\sum_{t=1}^{T}b^{%
\prime }f_{t}^{\ast }f_{t}^{\ast \prime }b>0\text{ \ \ \textit{a.s.}},
\end{equation*}%
for every $b$, thus proving the lemma. Clearly
\begin{equation*}
\frac{\ln \ln T}{T^{2}}\sum_{t=1}^{T}b^{\prime }\left( f_{t}^{\ast }-\Sigma_{\Delta f^*}^{1/2}W\left( t\right) \right) \left( f_{t}^{\ast }-\Sigma_{\Delta f^*}^{1/2}W\left( t\right) \right) ^{\prime }b\leq C_{0}\frac{\ln \ln
T}{T^{2}}\sum_{t=1}^{T}t^{1-2\epsilon }=o_{a.s.}\left( 1\right) ,
\end{equation*}%
by Assumption \ref{as-1}\textit{(iv)}. This entails that $IV$ is dominated.
Consider now $II$ and $III$. By the Law of the Iterated Logarithm (henceforth, LIL), we have that there exists a random $t_{0}$ such that, for all $t \geq t_{0}$, there exists a positive finite constant $C_{0}$ such that $\left\Vert W_{t}\right\Vert ^{2}\leq C_{0}t^{1/2}(\ln \ln t)^{1/2}$. 
Thus, using Assumption \ref{as-1}\textit{(iv)}
\begin{equation*}
\frac{\ln \ln T}{T^{2}}\sum_{t=1}^{T}b^{\prime }\Sigma _{\Delta f^*}^{1/2}W\left( t\right) \left(
f_{t}^{\ast }-\Sigma _{\Delta f^*}^{1/2}W\left( t\right) ^{\prime
}\right) b\leq C_{0}\frac{\ln \ln T}{T^{2}}\sum_{t=1}^{T}t^{1/2}\left( \ln\ln
t\right)^{1/2}  t^{1/2-\epsilon }\left( \ln t\right)
^{\left( 1+\epsilon \right) /\left( 2+\delta \right) }=o_{a.s.}\left(
1\right) .
\end{equation*}%
Finally it holds that%
\begin{equation*}
\lim \inf_{T\rightarrow \infty }\frac{\ln \ln T}{T^{2}}\sum_{t=1}^{T}b^{%
\prime }\Sigma _{\Delta f^*}^{1/2}W\left( t\right) W\left( t\right) ^{\prime
}\Sigma _{\Delta f^*}^{1/2}b=\frac{1}{4}\left( b^{\prime }\Sigma _{\Delta f^*}b\right) >0,
\end{equation*}%
by noting that $b^{\prime }\Sigma _{\Delta f^*}^{1/2}W\left( t\right) \overset{%
D}{=}\left( b^{\prime }\Sigma _{\Delta f^*}b\right) ^{1/2}B\left( t\right) $
with $B\left( t\right) $\ a scalar, standard Wiener process and by\ applying
equation (4.6) in \citet{donsker1977}, and by the positive definitness of $%
\Sigma _{\Delta f^*}$. Since this holds for all $b$, the Lemma follows. 
\end{proof}

We will now make extensive use of the notation $\widetilde{f}_{t}^{\left(
1\right) }=d_{1}t+d_{2}f_{t}^{\left( 1\right) \dag }$.

\begin{lemma}
\label{eigenval-large}Let $f_{t}^{\left( 1,2\right) }=\left[ \widetilde{f}%
_{t}^{\left( 1\right) },f_{t}^{\left( 2\right) \prime }\right] ^{\prime }$.
Under Assumptions \ref{as-1} and \ref{as-3}-\ref{as-4}, it holds that%
\begin{eqnarray}
\nu ^{\left( 1\right) }\left( \frac{1}{T^{2}}\sum_{t=1}^{T}f_{t}^{\left(
1,2\right) }f_{t}^{\left( 1,2\right) \prime }\right) &\geq &C_{0}T\text{ if }%
d_{1}=1,  \label{large-1} \\
\nu ^{\left( p\right) }\left( \frac{1}{T^{2}}\sum_{t=1}^{T}f_{t}^{\left(
1,2\right) }f_{t}^{\left( 1,2\right) \prime }\right) &\geq &\frac{C_{0}}{\ln
\ln T}\text{, } \text{for }d_{1}+1 \leq p\leq r_{2}+\max \left\{ d_{1},d_{2}\right\} , 
 \label{large-2}  \\
\nu ^{\left( r_{2}+1\right) }\left( \frac{1}{T^{2}}\sum_{t=1}^{T}f_{t}^{%
\left( 1,2\right) }f_{t}^{\left( 1,2\right) \prime }\right) &\leq &\frac{%
C_{0}}{T}\left( \ln T\right) ^{3/2+\epsilon }\text{ if }d_{1}=d_{2}=0,
\label{large-3}
\end{eqnarray}%
for $N$, $T$ large enough.
\end{lemma}

\begin{proof}
Let $\widetilde{d_{1}}=\left[ d_{1},0,...,0\right] ^{\prime }$ be an $\left(
r_{2}+1\right) $-dimensional vector. We have%
\begin{equation*}
\frac{1}{T^{2}}\sum_{t=1}^{T}f_{t}^{\left( 1,2\right) }f_{t}^{\left(
1,2\right) \prime }=\frac{1}{T^{2}}\sum_{t=1}^{T}t^{2}\widetilde{d_{1}}%
\widetilde{d_{1}}^{\prime }+\frac{1}{T^{2}}\widetilde{d_{1}}%
\sum_{t=1}^{T}tf_{t}^{\ast \prime }+\frac{1}{T^{2}}\sum_{t=1}^{T}f_{t}^{\ast
}t\widetilde{d_{1}}^{\prime }+\frac{1}{T^{2}}\sum_{t=1}^{T}f_{t}^{\ast
}f_{t}^{\ast \prime }.
\end{equation*}%
In the proof, we make repeated use of the lower bound entailed by Weyl's
inequality (see \citealp[p.181]{hornjohnson})
\begin{equation*}
\nu ^{\left( p\right) }\left( A+B\right) \geq \nu ^{\left( p\right)
}\left( A\right) +\nu ^{\left( \min \right) }\left( B\right) ,
\end{equation*}%
for two symmetric matrices $A$ and $B$. Clearly%
\begin{equation}
\nu ^{\left( 1\right) }\left( \frac{1}{T^{2}}\sum_{t=1}^{T}f_{t}^{\left(
1,2\right) }f_{t}^{\left( 1,2\right) \prime }\right) \geq \nu ^{\left(
1\right) }\left( \frac{1}{T^{2}}\sum_{t=1}^{T}t^{2}\widetilde{d_{1}}%
\widetilde{d_{1}}^{\prime }\right) +\nu ^{\left( \min \right) }\left(
B\right) ,  \label{theta-1}
\end{equation}%
with 
\begin{equation*}
B=\frac{1}{T^{2}}\widetilde{d_{1}}\sum_{t=1}^{T}tf_{t}^{\ast \prime }+\frac{1%
}{T^{2}}\sum_{t=1}^{T}f_{t}^{\ast }t\widetilde{d_{1}}^{\prime }+\frac{1}{%
T^{2}}\sum_{t=1}^{T}f_{t}^{\ast }f_{t}^{\ast \prime }.
\end{equation*}%
Simple algebra yields%
\begin{equation*}
\nu ^{\left( 1\right) }\left( \frac{1}{T^{2}}\sum_{t=1}^{T}t^{2}%
\widetilde{d_{1}}\widetilde{d_{1}}^{\prime }\right) =\frac{d_{1}^{2}}{3}T.
\end{equation*}%
Also, we have that $\left\vert \nu ^{\left( \min \right) }\left( B\right)
\right\vert =O_{a.s.}\left( \ln \ln T\right) $; indeed%
\begin{equation*}
\nu ^{\left( \min \right) }\left( B\right) \leq \nu ^{\left( \min
\right) }\left( \frac{1}{T^{2}}\sum_{t=1}^{T}f_{t}^{\ast }f_{t}^{\ast \prime
}\right) +\nu ^{\left( \min \right) }\left( \frac{1}{T^{2}}\widetilde{%
d_{1}}\sum_{t=1}^{T}tf_{t}^{\ast \prime }+\frac{1}{T^{2}}%
\sum_{t=1}^{T}f_{t}^{\ast }t\widetilde{d_{1}}^{\prime }\right) =\nu
^{\left( \min \right) }\left( \frac{1}{T^{2}}\sum_{t=1}^{T}f_{t}^{\ast
}f_{t}^{\ast \prime }\right) ,
\end{equation*}%
and by \citet[Example 2]{donsker1977} it holds that%
\begin{equation*}
\nu ^{\left( \min \right) }\left( \frac{1}{T^{2}}\sum_{t=1}^{T}f_{t}^{%
\ast }f_{t}^{\ast \prime }\right) \leq C_{0}\ln \ln T.
\end{equation*}%
Thus, by (\ref{theta-1})%
\begin{equation*}
\nu ^{\left( 1\right) }\left( \frac{1}{T^{2}}\sum_{t=1}^{T}f_{t}^{\left(
1,2\right) }f_{t}^{\left( 1,2\right) \prime }\right) \geq C_{0}T,
\end{equation*}%
which proves (\ref{large-1}). Turning to (\ref{large-2}), for each $p>1$%
\begin{eqnarray*}
\nu ^{\left( p\right) }\left( \frac{1}{T^{2}}\sum_{t=1}^{T}f_{t}^{\left(
1,2\right) }f_{t}^{\left( 1,2\right) \prime }\right)  &\geq &\nu ^{\left(
p\right) }\left( \frac{1}{T^{2}}\sum_{t=1}^{T}t^{2}\widetilde{d_{1}}%
\widetilde{d_{1}}^{\prime }+\frac{1}{T^{2}}\sum_{t=1}^{T}f_{t}^{\ast
}f_{t}^{\ast \prime }\right)  \\
&&+\nu ^{\left( \min \right) }\left( \frac{1}{T^{2}}\widetilde{d_{1}}%
\sum_{t=1}^{T}tf_{t}^{\ast \prime }+\frac{1}{T^{2}}\sum_{t=1}^{T}f_{t}^{\ast
}t\widetilde{d_{1}}^{\prime }\right)  \\
&\geq &\nu ^{\left( p\right) }\left( \frac{1}{T^{2}}\sum_{t=1}^{T}f_{t}^{%
\ast }f_{t}^{\ast \prime }\right) +\nu ^{\left( \min \right) }\left( 
\frac{1}{T^{2}}\sum_{t=1}^{T}t^{2}\widetilde{d_{1}}\widetilde{d_{1}}^{\prime
}\right) =\nu ^{\left( p\right) }\left( \frac{1}{T^{2}}%
\sum_{t=1}^{T}f_{t}^{\ast }f_{t}^{\ast \prime }\right) ,
\end{eqnarray*}%
so that the desired result follows immediately from Lemma \ref{LIL}.
Finally, consider (\ref{large-3}). Let $\widetilde{g}_{t}=\left[
g_{t},0,...,0\right] ^{\prime }$ and $\widetilde{f}_{t}=\left[
0,f_{t}^{\left( 2\right) \prime }\right] ^{\prime }$ be two $\left(
r_{2}+1\right) $-dimensional vectors; in this case we have%
\begin{equation*}
\frac{1}{T^{2}}\sum_{t=1}^{T}f_{t}^{\left( 1,2\right) }f_{t}^{\left(
1,2\right) \prime }=\frac{1}{T^{2}}\sum_{t=1}^{T}\widetilde{g}_{t}\widetilde{%
g}_{t}^{\prime }+\frac{1}{T^{2}}\sum_{t=1}^{T}\widetilde{f}_{t}\widetilde{f}%
_{t}^{\prime };
\end{equation*}%
thus%
\begin{equation*}
\nu ^{\left( \min \right) }\left( \frac{1}{T^{2}}\sum_{t=1}^{T}f_{t}^{%
\left( 1,2\right) }f_{t}^{\left( 1,2\right) \prime }\right) \leq \nu
^{\left( 1\right) }\left( \frac{1}{T^{2}}\sum_{t=1}^{T}\widetilde{g}_{t}%
\widetilde{g}_{t}^{\prime }\right) +\nu ^{\left( \min \right) }\left( 
\frac{1}{T^{2}}\sum_{t=1}^{T}\widetilde{f}_{t}\widetilde{f}_{t}^{\prime
}\right) \leq \frac{1}{T^{2}}\sum_{t=1}^{T}g_{t}^{2}.
\end{equation*}%
Assumption \ref{as-2}\textit{(i)} and equation (2.3) in %
\citet{serfling1970} imply that%
\begin{equation*}
E\max_{1\leq \widetilde{t}\leq T}\left\Vert \sum_{t=1}^{\widetilde{t}%
}g_{t}^{2}\right\Vert ^{2}\leq C_{0}\left( \ln T\right) ^{2}T,
\end{equation*}%
which, through Lemma \ref{borelcantelli}, yields the desired result.%
\end{proof}

\begin{lemma}
\label{remainder-1}Under Assumptions \ref{as-1}-\ref{as-3}%
\begin{equation}
\max_{1\leq p\leq N}\left\vert \nu ^{\left( p\right) }\left( \frac{1}{%
T^{3}}\sum_{t=1}^{T}u_{t}^{\left( 1\right) }u_{t}^{\left( 1\right) \prime }+%
\frac{1}{T^{3}}\sum_{t=1}^{T}\Lambda ^{\left( 1\right) }\widetilde{f}%
_{t}^{\left( 1\right) }u_{t}^{\left( 1\right) \prime }+\frac{1}{T^{3}}%
\sum_{t=1}^{T}u_{t}^{\left( 1\right) }\widetilde{f}_{t}^{\left( 1\right)
}\Lambda ^{\left( 1\right) \prime }\right) \right\vert =O_{a.s.}\left( \frac{%
N}{\sqrt{T}}l_{N,T}\right) . \nonumber
\end{equation}
\end{lemma}

\begin{proof}
We show the lemma for the case $d_{1}=d_{2}=1$; when either dummy is zero,
calculations become easier and the result can be readily shown. Let%
\begin{equation*}
\max_{1\leq p\leq N}\left\vert \nu ^{\left( p\right) }\left( \frac{1}{%
T^{3}}\sum_{t=1}^{T}u_{t}^{\left( 1\right) }u_{t}^{\left( 1\right) \prime }+%
\frac{1}{T^{3}}\sum_{t=1}^{T}\Lambda ^{\left( 1\right) }\widetilde{f}%
_{t}^{\left( 1\right) }u_{t}^{\left( 1\right) \prime }+\frac{1}{T^{3}}%
\sum_{t=1}^{T}u_{t}^{\left( 1\right) }\widetilde{f}_{t}^{\left( 1\right)
}\Lambda ^{\left( 1\right) \prime }\right) \right\vert =\nu ^{\left( \max
\right) },
\end{equation*}%
for short. It holds that%
\begin{eqnarray}
\frac{1}{3}\nu ^{\left( \max \right) } &\leq &\frac{1}{3}\left(
\sum_{i=1}^{N}\sum_{j=1}^{N}\left\vert \frac{1}{T^{3}}%
\sum_{t=1}^{T}u_{i,t}^{\left( 1\right) }u_{j,t}^{\left( 1\right) }+\frac{1}{%
T^{3}}\sum_{t=1}^{T}\Lambda _{i}^{\left( 1\right) }\widetilde{f}_{t}^{\left(
1\right) }u_{j,t}^{\left( 1\right) }+\frac{1}{T^{3}}\sum_{t=1}^{T}\Lambda
_{j}^{\left( 1\right) }\widetilde{f}_{t}^{\left( 1\right) }u_{i,t}^{\left(
1\right) }\right\vert ^{2}\right) ^{1/2}.  \label{theta-max-1} \\
&\leq &\left( \sum_{i=1}^{N}\sum_{j=1}^{N}\left\vert \frac{1}{T^{3}}%
\sum_{t=1}^{T}u_{i,t}^{\left( 1\right) }u_{j,t}^{\left( 1\right)
}\right\vert ^{2}\right) ^{1/2}+\left(
\sum_{i=1}^{N}\sum_{j=1}^{N}\left\vert \frac{1}{T^{3}}\sum_{t=1}^{T}%
\sum_{k=1}^{r}\Lambda _{i}^{\left( 1\right) }\widetilde{f}_{t}^{\left(
1\right) }u_{j,t}^{\left( 1\right) }\right\vert ^{2}\right) ^{1/2}  \notag \\
&&+\left( \sum_{i=1}^{N}\sum_{j=1}^{N}\left\vert \frac{1}{T^{3}}%
\sum_{t=1}^{T}\Lambda _{j}^{\left( 1\right) }\widetilde{f}_{t}^{\left(
1\right) }u_{i,t}^{\left( 1\right) }\right\vert ^{2}\right) ^{1/2},  \notag
\end{eqnarray}%
where the first passage is the usual spectral norm inequality, and the last
passage follows from applying (twice) the $C_{r}$-inequality (\citealp[p. 140]{davidson}). \\
Let now 
\begin{equation}
u_{i,t}^{\left( 2\right) }=\lambda ^{\left( 1\right) }g_{t}+\lambda ^{\left(
3\right) \prime }f_{t}^{\left( 3\right) }+u_{i,t}  \label{u-2}
\end{equation}%
and note that%
\begin{eqnarray*}
&&\frac{1}{3}\sum_{i=1}^{N}\sum_{j=1}^{N}\left\vert \frac{1}{T^{3}}%
\sum_{t=1}^{T}u_{i,t}^{\left( 1\right) }u_{j,t}^{\left( 1\right)
}\right\vert ^{2} \\
&\leq &\sum_{i=1}^{N}\sum_{j=1}^{N}\left\vert \frac{1}{T^{3}}%
\sum_{t=1}^{T}u_{i,t}^{\left( 2\right) }u_{j,t}^{\left( 2\right)
}\right\vert ^{2}+\sum_{i=1}^{N}\sum_{j=1}^{N}\left\vert \frac{1}{T^{3}}%
\sum_{t=1}^{T}\sum_{k=1}^{r_{2}}\Lambda _{i,k}^{\left( 2\right)
}f_{k,t}^{\left( 2\right) }f_{k,t}^{\left( 2\right) \prime }\Lambda
_{i,k}^{\left( 2\right) \prime }\right\vert ^{2} \\
&&+\sum_{i=1}^{N}\sum_{j=1}^{N}\left\vert \frac{1}{T^{3}}\sum_{t=1}^{T}%
\sum_{k=1}^{r_{2}}\Lambda _{i,k}^{\left( 2\right) }f_{k,t}^{\left( 2\right)
}u_{j,t}^{\left( 2\right) }\right\vert
^{2}+\sum_{i=1}^{N}\sum_{j=1}^{N}\left\vert \frac{1}{T^{3}}%
\sum_{t=1}^{T}\sum_{k=1}^{r_{2}}\Lambda _{j,k}^{\left( 2\right)
}f_{k,t}^{\left( 2\right) }u_{i,t}^{\left( 2\right) }\right\vert ^{2}.
\end{eqnarray*}%
We have 
\begin{eqnarray*}
&&E\max_{h_{1},h_{2},\widetilde{t}}\sum_{i=1}^{h_{1}}\sum_{j=1}^{h_{2}}\left\vert \frac{1}{T^{3}}\sum_{t=1}^{\widetilde{t}}u_{i,t}^{\left( 2\right)
}u_{j,t}^{\left( 2\right) }\right\vert ^{2}\leq
\sum_{i=1}^{N}\sum_{j=1}^{N}E\max_{\widetilde{t}}\left\vert \frac{1}{T^{3}}%
\sum_{t=1}^{\widetilde{t}}u_{i,t}^{\left( 2\right) }u_{j,t}^{\left( 2\right)
}\right\vert ^{2} \\
&\leq &C_{0}T\frac{1}{T^{6}}\sum_{i=1}^{N}\sum_{j=1}^{N}E\sum_{t=1}^{T}\left\vert u_{i,t}^{\left( 2\right) }\right\vert ^{2}\left\vert u_{j,t}^{\left(
2\right) }\right\vert ^{2}\leq C_{0}N^{2}T^{-4}\max_{1\leq i\leq
N}E\left\vert u_{i,t}^{\left( 2\right) }\right\vert ^{4} \\
&\leq &C_{0}N^{2}T^{-4}\left( \max_{1\leq i\leq N}E\left\vert
u_{i,t}\right\vert ^{4}+\max_{1\leq i\leq N}\left\Vert \lambda _{i}^{\left(
3\right) }\right\Vert ^{4}E\left\Vert f_{t}^{\left( 3\right) }\right\Vert
^{4}++\max_{1\leq i\leq N}\left\Vert \lambda _{i}^{\left( 1\right)
}\right\Vert ^{4}E\left\Vert g_{t}\right\Vert ^{4}\right) \\
&\leq &C_{0}N^{2}T^{-4},
\end{eqnarray*}%
so that, by Lemma \ref{borelcantelli}%
\begin{equation}
\sum_{i=1}^{N}\sum_{j=1}^{N}\left\vert \frac{1}{T^{3}}%
\sum_{t=1}^{T}u_{i,t}^{\left( 2\right) }u_{j,t}^{\left( 2\right)
}\right\vert ^{2}=O_{a.s.}\left( \frac{N^{2}}{T^{4}}\ln ^{2+\epsilon }N\ln
^{1+\epsilon }T\right) .  \label{trend-1}
\end{equation}%
Also%
\begin{eqnarray*}
&&\sum_{i=1}^{N}\sum_{j=1}^{N}\left\vert \frac{1}{T^{3}}\sum_{t=1}^{T}%
\sum_{k=1}^{r_{2}}\Lambda _{i,k}^{\left( 2\right) }f_{k,t}^{\left( 2\right)
}f_{k,t}^{\left( 2\right) \prime }\Lambda _{i,k}^{\left( 2\right)
}\right\vert ^{2} \\
&\leq &T^{-6}N^{2}\left( \max_{i}\left\Vert \Lambda _{i,k}^{\left( 2\right)
}\right\Vert \right) ^{2}\left\Vert \sum_{t=1}^{T}f_{t}^{\left( 2\right)
}f_{t}^{\left( 2\right) \prime }\right\Vert ^{2};
\end{eqnarray*}%
on account of Assumption \ref{as-1}\textit{(iv)}, it holds that%
\begin{equation*}
\left\Vert \frac{\sum_{t=1}^{T}f_{t}^{\left( 2\right) }f_{t}^{\left(
2\right) \prime }}{T^{2}\ln \ln T}\right\Vert ^{2}=\left\Vert \Sigma_{\Delta f^*}^{1/2}\frac{\sum_{t=1}^{T}W\left( t\right) W\left( t\right)
^{\prime }}{T^{2}\ln \ln T}\Sigma _{\Delta f^*}^{1/2}\right\Vert
^{2}+o_{a.s.}\left( 1\right) =O_{a.s.}\left( 1\right) ;
\end{equation*}%
the final result follows from \citet[Example 2]{donsker1977}). Thus%
\begin{equation}
\sum_{i=1}^{N}\sum_{j=1}^{N}\left\vert \frac{1}{T^{3}}\sum_{t=1}^{T}%
\sum_{k=1}^{r_{2}}\Lambda _{i,k}^{\left( 2\right) }f_{k,t}^{\left( 2\right)
}f_{k,t}^{\left( 2\right) \prime }\Lambda _{i,k}^{\left( 2\right)
}\right\vert ^{2}=O_{a.s.}\left( \frac{N^{2}}{T^{2}}\left( \ln \ln T\right)
^{2}\right) .  \label{trend-2}
\end{equation}%
Finally, consider 
\begin{eqnarray*}
&&E\max_{h_{1},h_{2},\widetilde{t}}\sum_{i=1}^{h_{1}}\sum_{j=1}^{h_{2}}\left\vert \sum_{t=1}^{\widetilde{t}}\Lambda _{i}^{\left( 2\right) \prime
}f_{t}^{\left( 2\right) }u_{j,t}^{\left( 2\right) }\right\vert ^{2}\leq
\sum_{i=1}^{N}\sum_{j=1}^{N}E\max_{\widetilde{t}}\left\vert \sum_{t=1}^{%
\widetilde{t}}\Lambda _{i}^{\left( 2\right) \prime }f_{t}^{\left( 2\right)
}u_{j,t}^{\left( 2\right) }\right\vert ^{2} \\
&\leq &C_{0}\left( \ln T\right) ^{2}\sum_{i=1}^{N}\sum_{j=1}^{N}E\left\vert
\sum_{t=1}^{T}\Lambda _{i}^{\left( 2\right) \prime }f_{t}^{\left( 2\right)
}u_{j,t}^{\left( 2\right) }\right\vert ^{2}\leq C_{0}\left( \ln T\right)
^{2}\sum_{i=1}^{N}\sum_{j=1}^{N}\Lambda _{i}^{\left( 2\right) \prime
}\sum_{t=1}^{T}\sum_{s=1}^{T}E\left( f_{t}^{\left( 2\right) }u_{j,t}^{\left(
2\right) }f_{s}^{\left( 2\right) \prime }u_{j,s}^{\left( 2\right) }\right)
\Lambda _{i}^{\left( 2\right) } \\
&\leq &C_{1}\left( \max_{i}\left\Vert \Lambda _{i}^{\left( 2\right) \prime
}\right\Vert \right) ^{2}\left( \ln T\right)
^{2}\sum_{i=1}^{N}\sum_{j=1}^{N}E\left\Vert \sum_{t=1}^{T}f_{t}^{\left(
2\right) }u_{j,t}^{\left( 2\right) }\right\Vert ^{2}\leq
C_{2}N^{2}T^{2}\left( \ln T\right) ^{2},
\end{eqnarray*}%
having used Assumption \ref{as-2}\textit{(ii)}, so that%
\begin{equation*}
\sum_{i=1}^{N}\sum_{j=1}^{N}\left\vert \frac{1}{T^{3}}\sum_{t=1}^{T}%
\sum_{k=1}^{r_{2}}\Lambda _{i,k}^{\left( 2\right) }f_{k,t}^{\left( 2\right)
}u_{j,t}^{\left( 2\right) }\right\vert ^{2}=O_{a.s.}\left( \frac{N^{2}}{T^{4}%
}\ln ^{2+\epsilon }N\ln ^{3+\epsilon }T\right) .
\end{equation*}%
Putting all together, we have%
\begin{equation*}
\left( \sum_{i=1}^{N}\sum_{j=1}^{N}\left\vert \frac{1}{T^{3}}%
\sum_{t=1}^{T}u_{i,t}^{\left( 1\right) }u_{j,t}^{\left( 1\right)
}\right\vert ^{2}\right) ^{1/2}=O_{a.s.}\left( \frac{N}{T}\ln \ln T\right) .
\end{equation*}%
Consider now%
\begin{eqnarray}
&&\frac{1}{2}\sum_{i=1}^{N}\sum_{j=1}^{N}\left\vert \frac{1}{T^{3}}%
\sum_{t=1}^{T}\Lambda _{i}^{\left( 1\right) }\widetilde{f}_{t}^{\left(
1\right) }u_{j,t}^{\left( 1\right) }\right\vert ^{2}  \notag \\
&\leq &\sum_{i=1}^{N}\sum_{j=1}^{N}\left\vert \frac{1}{T^{3}}%
\sum_{t=1}^{T}\Lambda _{i}^{\left( 1\right) }tu_{j,t}^{\left( 1\right)
}\right\vert ^{2}+\sum_{i=1}^{N}\sum_{j=1}^{N}\left\vert \frac{1}{T^{3}}%
\sum_{t=1}^{T}\Lambda _{i}^{\left( 1\right) }f_{t}^{\left( 1\right) \dag
}u_{j,t}^{\left( 1\right) }\right\vert ^{2}.  \label{lemma10-1}
\end{eqnarray}%
We have%
\begin{eqnarray*}
&&\frac{1}{2}\sum_{i=1}^{N}\sum_{j=1}^{N}\left\vert \frac{1}{T^{3}}%
\sum_{t=1}^{T}\Lambda _{i}^{\left( 1\right) }tu_{j,t}^{\left( 1\right)
}\right\vert ^{2} \\
&\leq &\sum_{i=1}^{N}\sum_{j=1}^{N}\left\vert \frac{1}{T^{3}}%
\sum_{t=1}^{T}\Lambda _{i}^{\left( 1\right) }tf_{t}^{\left( 2\right)
^{\prime }}\Lambda _{j}^{\left( 2\right) \prime }\right\vert
^{2}+\sum_{i=1}^{N}\sum_{j=1}^{N}\left\vert \frac{1}{T^{3}}%
\sum_{t=1}^{T}\Lambda _{i}^{\left( 1\right) }tu_{j,t}^{\left( 2\right)
}\right\vert ^{2}.
\end{eqnarray*}%
Note that%
\begin{eqnarray*}
&&E\max_{h_{1},h_{2},\widetilde{t}}\sum_{i=1}^{h_{1}}\sum_{j=1}^{h_{2}}\left\vert \frac{1}{T^{3}}\sum_{t=1}^{\widetilde{t}}\Lambda _{i}^{\left( 1\right)
}tf_{t}^{\left( 2\right) ^{\prime }}\Lambda _{j}^{\left( 2\right) \prime
}\right\vert ^{2}\leq T^{-6}\sum_{i=1}^{N}\sum_{j=1}^{N}E\max_{\widetilde{t}%
}\left\vert \sum_{t=1}^{\widetilde{t}}\Lambda _{i}^{\left( 1\right)
}tf_{t}^{\left( 2\right) ^{\prime }}\Lambda _{j}^{\left( 2\right) \prime
}\right\vert ^{2} \\
&\leq &C_{0}\left( \ln T\right)
^{2}T^{-6}\sum_{i=1}^{N}\sum_{j=1}^{N}E\left\vert \sum_{t=1}^{T}\Lambda
_{i}^{\left( 1\right) }tf_{t}^{\left( 2\right) ^{\prime }}\Lambda
_{j}^{\left( 2\right) \prime }\right\vert ^{2}\leq C_{0}\left( \ln T\right)
^{2}T^{-6}\sum_{i=1}^{N}\sum_{j=1}^{N}\Lambda _{i}^{\left( 1\right)
}\sum_{t=1}^{T}\sum_{s=1}^{T}E\left( tf_{s}^{\left( 2\right) \prime }\right)
\Lambda _{j}^{\left( 2\right) } \\
&\leq &C_{1}N^{2}T^{-6}\left( \max_{i}\left\Vert \Lambda _{i}^{\left(
1\right) }\right\Vert \right) ^{2}\left( \max_{i}\left\Vert \Lambda
_{i}^{\left( 2\right) }\right\Vert \right) ^{2}\left( \ln T\right)
^{2}E\left\Vert \sum_{t=1}^{T}tf_{t}^{\left( 2\right) }\right\Vert ^{2}\leq
C_{2}N^{2}T^{-1}\left( \ln T\right) ^{2},
\end{eqnarray*}%
having used Assumption \ref{as-2}\textit{(iii)}; Lemma \ref{borelcantelli}
entails that%
\begin{equation*}
\sum_{i=1}^{N}\sum_{j=1}^{N}\left\vert \frac{1}{T^{3}}\sum_{t=1}^{T}\Lambda
_{i}^{\left( 1\right) }tf_{t}^{\left( 2\right) ^{\prime }}\Lambda
_{j}^{\left( 2\right) \prime }\right\vert ^{2}=O_{a.s.}\left( \frac{N^{2}}{T}%
\ln ^{2+\epsilon }N\ln ^{3+\epsilon }T\right) .
\end{equation*}%
Similar passages yield%
\begin{equation*}
\sum_{i=1}^{N}\sum_{j=1}^{N}\left\vert \frac{1}{T^{3}}\sum_{t=1}^{T}\Lambda
_{i}^{\left( 1\right) }tu_{j,t}^{\left( 2\right) }\right\vert
^{2}=O_{a.s.}\left( \frac{N^{2}}{T^{3}}\ln ^{2+\epsilon }N\ln ^{3+\epsilon
}T\right) .
\end{equation*}%
Thus, finally%
\begin{equation*}
\sum_{i=1}^{N}\sum_{j=1}^{N}\left\vert \frac{1}{T^{3}}\sum_{t=1}^{T}\Lambda
_{i}^{\left( 1\right) }tu_{j,t}^{\left( 1\right) }\right\vert
^{2}=O_{a.s.}\left( \frac{N^{2}}{T}\ln ^{2+\epsilon }N\ln ^{3+\epsilon
}T\right) .
\end{equation*}%
We now consider the next term in equation (\ref{lemma10-1}). We have%
\begin{equation*}
\sum_{i=1}^{N}\sum_{j=1}^{N}\left\vert \frac{1}{T^{3}}\sum_{t=1}^{T}\Lambda
_{i}^{\left( 1\right) }f_{t}^{\left( 1\right) \dag }u_{j,t}^{\left( 1\right)
}\right\vert ^{2}=\sum_{i=1}^{N}\sum_{j=1}^{N}\left\vert \frac{1}{T^{3}}%
\sum_{t=1}^{T}\Lambda _{i}^{\left( 1\right) }f_{t}^{\left( 1\right) \dag
}f_{t}^{\left( 2\right) ^{\prime }}\Lambda _{j}^{\left( 2\right) \prime
}\right\vert ^{2}+\sum_{i=1}^{N}\sum_{j=1}^{N}\left\vert \frac{1}{T^{3}}%
\sum_{t=1}^{T}\Lambda _{i}^{\left( 1\right) }f_{t}^{\left( 1\right) \dag
}u_{j,t}^{\left( 2\right) }\right\vert ^{2}.
\end{equation*}%
Similar passages as above yield%
\begin{eqnarray*}
&&E\max_{h_{1},h_{2},\widetilde{t}}\sum_{i=1}^{h_{1}}\sum_{j=1}^{h_{2}}\left\vert \frac{1}{T^{3}}\sum_{t=1}^{\widetilde{t}}\Lambda _{i}^{\left( 1\right)
}f_{t}^{\left( 1\right) \dag }f_{t}^{\left( 2\right) ^{\prime }}\Lambda
_{j}^{\left( 2\right) \prime }\right\vert ^{2} \\
&\leq &C_{0}N^{2}T^{-6}\left( \max_{i}\left\Vert \Lambda _{i}^{\left(
1\right) }\right\Vert \right) ^{2}\left( \max_{i}\left\Vert \Lambda
_{i}^{\left( 2\right) }\right\Vert \right) ^{2}\left( \ln T\right)
^{2}E\left\Vert \sum_{t=1}^{T}f_{t}^{\left( 1\right) \dag }f_{t}^{\left(
2\right) ^{\prime }}\right\Vert ^{2}\leq C_{1}N^{2}T^{-6}\left( \ln T\right)
^{2}T^{4},
\end{eqnarray*}%
having used Assumption \ref{as-1}\textit{(vi)}. Similarly%
\begin{eqnarray*}
&&E\max_{h_{1},h_{2},\widetilde{t}}\sum_{i=1}^{h_{1}}\sum_{j=1}^{h_{2}}\left\vert \frac{1}{T^{3}}\sum_{t=1}^{\widetilde{t}}\Lambda _{i}^{\left( 1\right)
}f_{t}^{\left( 1\right) \dag }u_{j,t}^{\left( 2\right) }\right\vert ^{2} \\
&\leq &C_{0}N^{2}T^{-6}\left( \max_{i}\left\Vert \Lambda _{i}^{\left(
1\right) }\right\Vert \right) ^{2}\left( \ln T\right) ^{2}E\left\Vert
\sum_{t=1}^{T}f_{t}^{\left( 1\right) \dag }u_{j,t}^{\left( 2\right)
}\right\Vert ^{2}\leq C_{1}N^{2}T^{-6}\left( \ln T\right) ^{2}T^{2},
\end{eqnarray*}%
having used Assumption \ref{as-2}\textit{(ii)}. Thus, using Lemma \ref%
{borelcantelli}%
\begin{equation*}
\sum_{i=1}^{N}\sum_{j=1}^{N}\left\vert \frac{1}{T^{3}}\sum_{t=1}^{T}\Lambda
_{i}^{\left( 1\right) }f_{t}^{\left( 1\right) \dag }u_{j,t}^{\left( 1\right)
}\right\vert ^{2}=O_{a.s.}\left( \frac{N}{T}\ln ^{1+\epsilon }N\ln ^{\frac{3%
}{2}+\epsilon }T\right) .
\end{equation*}%
Using (\ref{theta-max-1}) and putting all together, the desired result
obtains. 
\end{proof}

\begin{lemma}
\label{remainder}Under Assumptions \ref{as-1}-\ref{as-3}%
\begin{equation}
\max_{1\leq p\leq N}\left\vert \nu ^{\left( p\right) }\left( \frac{1}{%
T^{2}}\sum_{t=1}^{T}u_{t}^{\left( 2\right) }u_{t}^{\left( 2\right) \prime }+%
\frac{1}{T^{2}}\sum_{t=1}^{T}\Lambda ^{\left( 1,2\right) }f_{t}^{\left(
1,2\right) }u_{t}^{\left( 2\right) \prime }+\frac{1}{T^{2}}%
\sum_{t=1}^{T}u_{t}^{\left( 2\right) }f_{t}^{\left( 1,2\right) \prime
}\Lambda ^{\left( 1,2\right) \prime }\right) \right\vert =O_{a.s.}\left( 
\frac{N}{\sqrt{T}}l_{N,T}\right) ,  \label{lambda-max-remainder}
\end{equation}%
where $u_{t}^{\left( 2\right) }$\ is defined in (\ref{u-2}).
\end{lemma}

\begin{proof}
Let%
\begin{equation*}
\max_{1\leq p\leq N}\left\vert \nu ^{\left( p\right) }\left( \frac{1}{%
T^{2}}\sum_{t=1}^{T}u_{t}^{\left( 2\right) }u_{t}^{\left( 2\right) \prime }+%
\frac{1}{T^{2}}\sum_{t=1}^{T}\Lambda ^{\left( 1,2\right) }f_{t}^{\left(
1,2\right) }u_{t}^{\left( 2\right) \prime }+\frac{1}{T^{2}}%
\sum_{t=1}^{T}u_{t}^{\left( 2\right) }f_{t}^{\left( 1,2\right) \prime
}\Lambda ^{\left( 1,2\right) \prime }\right) \right\vert =\nu ^{\left(
\max \right) },
\end{equation*}%
for short. As before%
\begin{eqnarray}
&&\frac{1}{3}\nu ^{\left( \max \right) }  \label{theta-max} \\
&\leq &\frac{1}{3}\left( \sum_{i=1}^{N}\sum_{j=1}^{N}\left\vert \frac{1}{%
T^{2}}\sum_{t=1}^{T}u_{i,t}^{\left( 2\right) }u_{j,t}^{\left( 2\right) }+%
\frac{1}{T^{2}}\sum_{t=1}^{T}\sum_{k=1}^{r_{2}+r_{1}}\Lambda _{i,k}^{\left(
1,2\right) }f_{k,t}^{\left( 1,2\right) }u_{j,t}^{\left( 2\right) }+\frac{1}{%
T^{2}}\sum_{t=1}^{T}\sum_{k=1}^{r_{2}+r_{1}}\Lambda _{j,k}^{\left(
1,2\right) }f_{k,t}^{\left( 1,2\right) }u_{i,t}^{\left( 2\right)
}\right\vert ^{2}\right) ^{1/2}.  \notag \\
&\leq &\left( \sum_{i=1}^{N}\sum_{j=1}^{N}\left\vert \frac{1}{T^{2}}%
\sum_{t=1}^{T}u_{i,t}^{\left( 2\right) }u_{j,t}^{\left( 2\right)
}\right\vert ^{2}\right) ^{1/2}+\left(
\sum_{i=1}^{N}\sum_{j=1}^{N}\left\vert \frac{1}{T^{2}}\sum_{t=1}^{T}%
\sum_{k=1}^{r_{2}+r_{1}}\Lambda _{i,k}^{\left( 1,2\right) }f_{k,t}^{\left(
1,2\right) }u_{j,t}^{\left( 2\right) }\right\vert ^{2}\right) ^{1/2}  \notag
\\
&&+\left( \sum_{i=1}^{N}\sum_{j=1}^{N}\left\vert \frac{1}{T^{2}}%
\sum_{t=1}^{T}\sum_{k=1}^{r_{2}+r_{1}}\Lambda _{j,k}^{\left( 1,2\right)
}f_{k,t}^{\left( 1,2\right) }u_{i,t}^{\left( 2\right) }\right\vert
^{2}\right) ^{1/2}.  \notag
\end{eqnarray}%
Consider the first term; by (\ref{trend-1}), 
\begin{equation*}
\left( \sum_{i=1}^{N}\sum_{j=1}^{N}\left\vert \frac{1}{T^{2}}%
\sum_{t=1}^{T}u_{i,t}^{\left( 2\right) }u_{j,t}^{\left( 2\right)
}\right\vert ^{2}\right) ^{1/2}=O_{a.s.}\left( \frac{N}{T}\left( \ln
N\right) ^{1+\epsilon }\left( \ln T\right) ^{\left( 1+\epsilon \right)
/2}\right) .
\end{equation*}%
Similarly, considering the second term in (\ref{theta-max}) we have 
\begin{eqnarray*}
&&E\max_{h_{1},h_{2},\widetilde{t}}\sum_{i=1}^{h_{1}}\sum_{j=1}^{h_{2}}\left\vert \frac{1}{T^{2}}\sum_{t=1}^{\widetilde{t}}\sum_{k=1}^{r_{2}+r_{1}}%
\Lambda _{i,k}^{\left( 1,2\right) }f_{k,t}^{\left( 1,2\right)
}u_{j,t}^{\left( 2\right) }\right\vert ^{2}\leq
T^{-4}\sum_{i=1}^{N}\sum_{j=1}^{N}E\max_{\widetilde{t}}\left\vert
\sum_{t=1}^{\widetilde{t}}\sum_{k=1}^{r_{2}+r_{1}}\Lambda _{i,k}^{\left(
1,2\right) }f_{k,t}^{\left( 1,2\right) }u_{j,t}^{\left( 2\right)
}\right\vert ^{2} \\
&\leq &C_{0}T^{-4}\left( \ln T\right)
^{2}\sum_{i=1}^{N}\sum_{j=1}^{N}E\left\vert
\sum_{t=1}^{T}\sum_{k=1}^{r_{2}+r_{1}}\Lambda _{i,k}^{\left( 1,2\right)
}f_{k,t}^{\left( 1,2\right) }u_{j,t}^{\left( 2\right) }\right\vert ^{2} \\
&\leq &C_{0}T^{-4}\left( \ln T\right) ^{2}\left( \max_{1\leq i\leq
N}\left\Vert \Lambda _{i}^{\left( 1,2\right) }\right\Vert \right)
^{2}\sum_{i=1}^{N}\sum_{j=1}^{N}\sum_{t=1}^{T}\sum_{s=1}^{T}E\left(
f_{t}^{\left( 1,2\right) }f_{s}^{\left( 1,2\right) \prime }u_{j,t}^{\left(
2\right) }u_{j,s}^{\left( 2\right) }\right)  \\
&\leq &C_{0}T^{-4}\left( \ln T\right)
^{2}\sum_{i=1}^{N}\sum_{j=1}^{N}E\left\Vert \sum_{t=1}^{T}f_{t}^{\left(
1,2\right) }u_{j,t}^{\left( 2\right) }\right\Vert ^{2}\leq
C_{0}N^{2}T^{-1}\left( \ln T\right) ^{2},
\end{eqnarray*}%
having used equation (2.3) in \citet{serfling1970}, Assumption \ref{as-3}%
\textit{(i)} and Assumption \ref{as-2}\textit{(ii)}.\ From here henceforth,
the proof is the same as for the first tem in (\ref{theta-max}); also, the
proof for the third term in (\ref{theta-max}) is exactly the same, and it is
therefore omitted. Putting everything together, the lemma follows. 
\end{proof}

\subsection{Proofs of main results}

\begin{proof}[Proof of Lemma \protect\ref{maciej}]
When $d_{1}=0$, the lemma follows immediately from $B$ having full rank.
When $d_{1}=1$, the proof follows the arguments in \citet{maciejowska2010}.
Let%
\begin{equation*}
\mathcal{F}_{t}=\left( A|B\right) \left( 
\begin{array}{c}
t \\ 
\psi _{t}%
\end{array}%
\right) =C\left( 
\begin{array}{c}
t \\ 
\psi _{t}%
\end{array}%
\right) ;
\end{equation*}%
by Assumption \ref{maciejowska}\textit{(ii)}, $C$ has full rank. It is
therefore possible to re-write the expression above as%
\begin{equation*}
\mathcal{F}_{t}=P\left( D_{1}|D_{2}\right) \left( 
\begin{array}{c}
t \\ 
\psi _{t}%
\end{array}%
\right) ,
\end{equation*}%
where $D_{1}=\left[ 1,0,...,0\right] ^{\prime }$ is $r\times 1$, and $P$ and 
$D_{2}$ are $r\times r$ and have full rank. Among the possible matrices that
satisfy this representation one can consider $\left( D_{1}|D_{2}\right)
=\left( I_{r}|E\right) $, where $E=\left[ E_{1},...,E_{r}\right] $ is a
nonzero vector. The desired result follows immediately after computing 
\begin{equation*}
P^{-1}\mathcal{F}_{t}=\left( 
\begin{array}{c}
t+E_{1}\psi _{r,t} \\ 
\psi _{1,t}+E_{2}\psi _{r,t} \\ 
\psi _{2,t}+E_{3}\psi _{r,t} \\ 
. \\ 
\psi _{r-1,t}+E_{r}\psi _{r,t}%
\end{array}%
\right) .
\end{equation*}%
\end{proof}

\begin{proof}[Proof of Theorem \protect\ref{eigenvalues}]
We start with (\ref{lambda-x-trend})-(\ref{lambda-x-notrend}). Weyl's
inequality entails that, for $0\leq p\leq r_{1}$%
\begin{eqnarray*}
\nu _{1}^{\left( p\right) } &\geq &\nu ^{\left( p\right) }\left( \frac{1}{%
T^{3}}\sum_{t=1}^{T}\Lambda ^{\left( 1\right) }\widetilde{f}_{t}^{\left(
1\right) }\widetilde{f}_{t}^{\left( 1\right) \prime }\Lambda ^{\left(
1\right) \prime }\right)  \\
&&+\nu ^{\left( N\right) }\left( \frac{1}{T^{3}}\sum_{t=1}^{T}u_{t}^{%
\left( 1\right) }u_{t}^{\left( 1\right) \prime }+\frac{1}{T^{3}}%
\sum_{t=1}^{T}\Lambda ^{\left( 1\right) }\widetilde{f}_{t}^{\left( 1\right)
}u_{t}^{\left( 1\right) \prime }+\frac{1}{T^{3}}\sum_{t=1}^{T}u_{t}^{\left(
1\right) }\widetilde{f}_{t}^{\left( 1\right) \prime }\Lambda ^{\left(
1\right) \prime }\right) .
\end{eqnarray*}%
We already know that, by Lemma \ref{remainder-1}%
\begin{equation*}
\nu ^{\left( N\right) }\left( \frac{1}{T^{3}}\sum_{t=1}^{T}u_{t}^{\left(
1\right) }u_{t}^{\left( 1\right) \prime }+\frac{1}{T^{3}}\sum_{t=1}^{T}%
\Lambda ^{\left( 1\right) }\widetilde{f}_{t}^{\left( 1\right) }u_{t}^{\left(
1\right) \prime }+\frac{1}{T^{3}}\sum_{t=1}^{T}u_{t}^{\left( 1\right) }%
\widetilde{f}_{t}^{\left( 1\right) \prime }\Lambda ^{\left( 1\right) \prime
}\right) =O_{a.s.}\left( \frac{N}{\sqrt{T}}l_{N,T}\right) .
\end{equation*}%
Consider now%
\begin{eqnarray*}
&&\nu ^{\left( p\right) }\left( \frac{1}{T^{3}}\sum_{t=1}^{T}\Lambda
^{\left( 1\right) }\widetilde{f}_{t}^{\left( 1\right) }\widetilde{f}%
_{t}^{\left( 1\right) \prime }\Lambda ^{\left( 1\right) \prime }\right)  \\
&\geq &d_{1}^{2}\nu ^{\left( p\right) }\left( \frac{1}{T^{3}}%
\sum_{t=1}^{T}t^{2}\Lambda ^{\left( 1\right) \prime }\Lambda ^{\left(
1\right) }\right)  \\
&&+\nu ^{\left( N\right) }\left( 2\frac{d_{1}}{T^{3}}%
\sum_{t=1}^{T}tf_{t}^{\left( 1\right) \dag }\Lambda ^{\left( 1\right)
}\Lambda ^{\left( 1\right) \prime }+\frac{1}{T^{3}}\Lambda ^{\left( 1\right)
}\sum_{t=1}^{T}f_{t}^{\left( 1\right) \dag }f_{t}^{\left( 1\right) \dag
\prime }\Lambda ^{\left( 1\right) \prime }\right) 
\end{eqnarray*}%
We have%
\begin{equation*}
\nu ^{\left( p\right) }\left( \frac{1}{T^{3}}\sum_{t=1}^{T}t^{2}\Lambda
^{\left( 1\right) \prime }\Lambda ^{\left( 1\right) }\right) =\left( \frac{1%
}{T^{3}}\sum_{t=1}^{T}t^{2}\right) \nu ^{\left( p\right) }\left( \Lambda
^{\left( 1\right) \prime }\Lambda ^{\left( 1\right) }\right) \geq C_{0}N,
\end{equation*}%
in view of Assumption \ref{as-3}\textit{(ii)}. Consider now%
\begin{equation*}
\nu ^{\left( N\right) }\left( 2\frac{d_{1}}{T^{3}}\sum_{t=1}^{T}tf_{t}^{%
\left( 1\right) \dag }\Lambda ^{\left( 1\right) }\Lambda ^{\left( 1\right)
\prime }+\frac{1}{T^{3}}\sum_{t=1}^{T}\left( f_{t}^{\left( 1\right) \dag
}\right) ^{2}\Lambda ^{\left( 1\right) }\Lambda ^{\left( 1\right) \prime
}\right) ;
\end{equation*}%
by \citet[Example 2]{donsker1977} we have 
\begin{equation*}
\frac{1}{T^{3}}\sum_{t=1}^{T}\left( f_{t}^{\left( 1\right) \dag }\right)
^{2}=O_{a.s.}\left( \frac{\ln \ln T}{T}\right) .
\end{equation*}%
Also, by Assumption \ref{as-2}\textit{(iii)} and equation (2.3) in %
\citet{serfling1970} we have%
\begin{equation*}
E\max_{1\leq t\leq T}\left\vert \sum_{j=1}^{t}jf_{j}^{\left( 1\right) \dag
}\right\vert ^{2}=C_{0}T^{5}\left( \ln T\right) ^{2},
\end{equation*}%
so that by\ Lemma \ref{borelcantelli} we have%
\begin{equation*}
\frac{1}{T^{3}}\left\vert \sum_{t=1}^{T}tf_{t}^{\left( 1\right) \dag
}\right\vert =O_{a.s.}\left( T^{-1/2}\left( \ln T\right) ^{3/2+\epsilon
}\right) .
\end{equation*}%
The same steps as in the proofs of Lemmas \ref{remainder-1} and \ref%
{remainder} entail%
\begin{equation*}
\nu ^{\left( N\right) }\left( 2\frac{d_{1}}{T^{3}}\sum_{t=1}^{T}tf_{t}^{%
\left( 1\right) \dag }\Lambda ^{\left( 1\right) }\Lambda ^{\left( 1\right)
\prime }+\frac{1}{T^{3}}\sum_{t=1}^{T}\left( f_{t}^{\left( 1\right) \dag
}\right) ^{2}\Lambda ^{\left( 1\right) }\Lambda ^{\left( 1\right) \prime
}\right) =O_{a.s.}\left( \frac{N}{\sqrt{T}}l_{N,T}\right) .
\end{equation*}%
Putting everything together, the desired result follows. When $p>r_{1}$%
\begin{eqnarray*}
\nu _{1}^{\left( p\right) } &\leq &\nu ^{\left( p\right) }\left( \frac{1}{%
T^{3}}\sum_{t=1}^{T}\Lambda ^{\left( 1\right) }\widetilde{f}_{t}^{\left(
1\right) }\widetilde{f}_{t}^{\left( 1\right) \prime }\Lambda ^{\left(
1\right) \prime }\right)  \\
&&+\nu ^{\left( 1\right) }\left( \frac{1}{T^{3}}\sum_{t=1}^{T}u_{t}^{%
\left( 1\right) }u_{t}^{\left( 1\right) \prime }+\frac{1}{T^{3}}%
\sum_{t=1}^{T}\Lambda ^{\left( 1\right) }\widetilde{f}_{t}^{\left( 1\right)
}u_{t}^{\left( 1\right) \prime }+\frac{1}{T^{3}}\sum_{t=1}^{T}u_{t}^{\left(
1\right) }\widetilde{f}_{t}^{\left( 1\right) \prime }\Lambda ^{\left(
1\right) \prime }\right)  \\
&\leq &\nu ^{\left( 1\right) }\left( \frac{1}{T^{3}}\sum_{t=1}^{T}u_{t}^{%
\left( 1\right) }u_{t}^{\left( 1\right) \prime }+\frac{1}{T^{3}}%
\sum_{t=1}^{T}\Lambda ^{\left( 1\right) }\widetilde{f}_{t}^{\left( 1\right)
}u_{t}^{\left( 1\right) \prime }+\frac{1}{T^{3}}\sum_{t=1}^{T}u_{t}^{\left(
1\right) }\widetilde{f}_{t}^{\left( 1\right) \prime }\Lambda ^{\left(
1\right) \prime }\right) ;
\end{eqnarray*}%
Lemma \ref{remainder-1} immediately yields the desired result.

The proof of (\ref{lambda-x-large})-(\ref{lambda-x-small}) is very similar.
Whenever $1\leq p\leq r_{1}+r_{2}+\left( 1-r_{1}\right) d_{2}$, we have%
\begin{eqnarray*}
\nu _{2}^{\left( p\right) } &\geq &\nu ^{\left( p\right) }\left( \frac{1}{%
T^{2}}\sum_{t=1}^{T}\Lambda ^{\left( 1,2\right) }f_{t}^{\left( 1,2\right)
}f_{t}^{\left( 1,2\right) \prime }\Lambda ^{\left( 1,2\right) \prime
}\right)  \\
&&+\nu ^{\left( N\right) }\left( \frac{1}{T^{2}}\sum_{t=1}^{T}u_{t}^{%
\left( 2\right) }u_{t}^{\left( 2\right) \prime }+\frac{1}{T^{2}}%
\sum_{t=1}^{T}\Lambda ^{\left( 1,2\right) }f_{t}^{\left( 1,2\right)
}u_{t}^{\left( 2\right) \prime }+\frac{1}{T^{2}}\sum_{t=1}^{T}u_{t}^{\left(
2\right) }f_{t}^{\left( 1,2\right) \prime }\Lambda ^{\left( 1,2\right)
\prime }\right) .
\end{eqnarray*}%
By Lemma \ref{remainder} we have%
\begin{equation*}
\nu ^{\left( N\right) }\left( \frac{1}{T^{2}}\sum_{t=1}^{T}u_{t}^{\left(
2\right) }u_{t}^{\left( 2\right) \prime }+\frac{1}{T^{2}}\sum_{t=1}^{T}%
\Lambda ^{\left( 1,2\right) }f_{t}^{\left( 1,2\right) }u_{t}^{\left(
2\right) \prime }+\frac{1}{T^{2}}\sum_{t=1}^{T}u_{t}^{\left( 2\right)
}f_{t}^{\left( 1,2\right) \prime }\Lambda ^{\left( 1,2\right) \prime
}\right) =O_{a.s.}\left( \frac{N}{\sqrt{T}}l_{N,T}\right) .
\end{equation*}%
Also, using Theorem 7 in \citet{merikoski2004} and by Assumption \ref{as-3}%
\textit{(ii)}%
\begin{eqnarray*}
\nu ^{\left( p\right) }\left( \frac{1}{T^{2}}\sum_{t=1}^{T}\Lambda
^{\left( 1,2\right) }f_{t}^{\left( 1,2\right) }f_{t}^{\left( 1,2\right)
\prime }\Lambda ^{\left( 1,2\right) \prime }\right)  &\geq &\nu ^{\left(
p\right) }\left( \frac{1}{T^{2}}\sum_{t=1}^{T}f_{t}^{\left( 1,2\right)
}f_{t}^{\left( 1,2\right) \prime }\right) \nu ^{\left( \min \right)
}\left( \Lambda ^{\left( 1,2\right) \prime }\Lambda ^{\left( 1,2\right)
}\right)  \\
&\geq &C_{0}\frac{N}{\ln \ln T},
\end{eqnarray*}%
where the last passage follows from equation (\ref{large-2}) in Lemma \ref%
{eigenval-large}. Equation (\ref{lambda-x-large}) now follows readily.
Turning to (\ref{lambda-x-small}), whenever $p>r_{1}+r_{2}+\left(
1-r_{1}\right) d_{2}$,%
\begin{eqnarray*}
\nu _{2}^{\left( p\right) } &\leq &\nu ^{\left( p\right) }\left( \frac{1}{%
T^{2}}\sum_{t=1}^{T}\Lambda ^{\left( 1,2\right) }f_{t}^{\left( 1,2\right)
}f_{t}^{\left( 1,2\right) \prime }\Lambda ^{\left( 1,2\right) \prime
}\right)  \\
&&+\nu ^{\left( 1\right) }\left( \frac{1}{T^{2}}\sum_{t=1}^{T}u_{t}^{%
\left( 1\right) }u_{t}^{\left( 1\right) \prime }+\frac{1}{T^{2}}%
\sum_{t=1}^{T}\Lambda ^{\left( 1\right) }f_{t}^{\left( 1\right)
}u_{t}^{\left( 1\right) \prime }+\frac{1}{T^{2}}\sum_{t=1}^{T}u_{t}^{\left(
1\right) }f_{t}^{\left( 1\right) \prime }\Lambda ^{\left( 1\right) \prime
}\right)  \\
&\leq &\nu^{\left( 1\right) } \left( \frac{1}{T^{2}}\sum_{t=1}^{T}u_{t}^{\left( 1\right)
}u_{t}^{\left( 1\right) \prime }+\frac{1}{T^{2}}\sum_{t=1}^{T}\Lambda
^{\left( 1\right) }f_{t}^{\left( 1\right) }u_{t}^{\left( 1\right) \prime }+%
\frac{1}{T^{2}}\sum_{t=1}^{T}u_{t}^{\left( 1\right) }f_{t}^{\left( 1\right)
\prime }\Lambda ^{\left( 1\right) \prime }\right) ,
\end{eqnarray*}%
and Lemma \ref{remainder} immediately yields the desired result.
\end{proof}

\begin{proof}[Proof of Theorem \protect\ref{test-1}]

The proof is similar to that of related results in other papers - see e.g. %
\citet{trapani17}. We begin with (\ref{null-1}). Note that, under $%
H_{0,1}^{\left( p\right) }$, (\ref{lambda-x-trend}) and Lemma \ref{trapani2017}
entail that 
\begin{equation*}
P\left\{ \omega :\lim_{\min \left( N,T\right) \rightarrow \infty }\phi
_{1}^{\left( p\right) }\exp \left\{ -N^{1-\delta -\varepsilon }\right\}
=\infty \right\} =1,
\end{equation*}%
for every $\varepsilon >0$, and therefore we can henceforth assume that $%
\lim_{\min \left( N,T\right) \rightarrow \infty }\phi _{1}^{\left( p\right)
}=\infty $ and 
\begin{equation} \label{phidiv}
\left( \phi _{1}^{\left( p\right) }\right) ^{-1}=O\left(
\exp \left\{ -N^{1-\delta }\right\} \right).
\end{equation}
Let $E^{\ast }$ and $V^{\ast
} $ denote, respectively, expectation and variance conditional on $P^{\ast }$%
; we have, for $1 \leq j \leq R_{1}$%
\begin{equation*}
E^{\ast }\left( \zeta _{1,j}^{\left( p\right) }\left( u\right) \right)
=G_{1}\left( 0\right) \text{ and }V^{\ast }\left( \zeta _{1,j}^{\left(
p\right) }\left( u\right) \right) =G_{1}\left( 0\right) \left( 1-G_{1}\left(
0\right) \right) .
\end{equation*}%
Also%
\begin{eqnarray*}
&&\frac{1}{\sqrt{R_{1}}}\sum_{j=1}^{R_{1}}\left( \zeta _{1,j}^{\left(
p\right) }\left( u\right) -G_{1}\left( 0\right) \right) \\
&=&\frac{1}{\sqrt{R_{1}}}\sum_{j=1}^{R_{1}}\left( I\left( \xi _{1,j}^{\left(
p\right) }\leq 0\right) -G_{1}\left( 0\right) \right) +\frac{1}{\sqrt{R_{1}}}%
d_{u}\sum_{j=1}^{R_{1}}\left( G_{1}\left( u/\phi _{1}^{\left( p\right)
}\right) -G_{1}\left( 0\right) \right) \\
&&+\frac{1}{\sqrt{R_{1}}}\sum_{j=1}^{R_{1}}\left[ I\left( 0\leq \left\vert
\xi _{1,j}^{\left( p\right) }\right\vert \leq u/\phi _{1}^{\left( p\right)
}\right) -\left( G_{1}\left( u/\phi _{1}^{\left( p\right) }\right)
-G_{1}\left( 0\right) \right) d_{u}\right] ,
\end{eqnarray*}%
with $d_{u}=1$ for $u\geq 0$ and $-1$ otherwise. Letting $m_{G_{1}}$ denote
the upper bound for the density of $G_{1}$, we have%
\begin{equation*}
R_{1}^{-1}\int_{-\infty }^{\infty }\left( \sum_{j=1}^{R_{1}}\left(
G_{1}\left( u/\phi _{1}^{\left( p\right) }\right) -G_{1}\left( 0\right)
\right) \right) ^{2}dF_{1}\left( u\right) \leq m_{G_{1}}^{2}\frac{R_{1}}{%
\left( \phi _{1}^{\left( p\right) }\right) ^{2}}\int_{-\infty }^{\infty
}u^{2}dF_{1}\left( u\right) ,
\end{equation*}%
which drifts to zero under (\ref{r1}) by (\ref{phidiv}) and Assumption \ref{fg-1}. Also,
consider%
\begin{eqnarray*}
&&E^{\ast }\int_{-\infty }^{\infty }\left( \frac{1}{\sqrt{R_{1}}}%
\sum_{j=1}^{R_{1}}I\left( 0\leq \left\vert \xi _{1,j}^{\left( p\right)
}\right\vert \leq u/\phi _{1}^{\left( p\right) }\right) -\left( G_{1}\left(
u/\phi _{1}^{\left( p\right) }\right) -G_{1}\left( 0\right) \right)
d_{u}\right) ^{2}dF_{1}\left( u\right) \\
&=&E^{\ast }\int_{-\infty }^{\infty }\left( I\left( 0\leq \left\vert \xi
_{1,1}^{\left( p\right) }\right\vert \leq u/\phi _{1}^{\left( p\right)
}\right) -\left( G_{1}\left( u/\phi _{1}^{\left( p\right) }\right)
-G_{1}\left( 0\right) \right) d_{u}\right) ^{2}dF_{1}\left( u\right) \\
&=&\int_{-\infty }^{\infty }  V^{\ast }\{I\left( 0\leq \left\vert \xi
_{1,1}^{\left( p\right) }\right\vert \leq u/\phi _{1}^{\left( p\right)
}\right) \} dF_{1}\left( u\right)
\end{eqnarray*}%
by the independence of the $\xi _{1,j}^{\left( p\right) }$. Elementary
arguments yield%
\begin{eqnarray*}
V^{\ast }\{ I\left( 0\leq \xi _{1,1}^{\left( p\right) }\leq u/\phi _{1}^{\left(
p\right) }\right) \} &=&\left\vert G_{1}\left( u/\phi _{1}^{\left( p\right)
}\right) -G_{1}\left( 0\right) \right\vert \left( 1-\left\vert G_{1}\left(
u/\phi _{1}^{\left( p\right) }\right) -G_{1}\left( 0\right) \right\vert
\right) \\
&\leq &\left\vert G_{1}\left( u/\phi _{1}^{\left( p\right) }\right)
-G_{1}\left( 0\right) \right\vert \leq m_{G_{1}}\frac{\left\vert
u\right\vert }{\phi _{1}^{\left( p\right) }},
\end{eqnarray*}%
so that%
\begin{equation*}
\int_{-\infty }^{\infty } V^{\ast }\{ I\left( 0\leq \left\vert \xi
_{1,1}^{\left( p\right) }\right\vert \leq u/\phi _{1}^{\left( p\right)
}\right) \} dF_{1}\left( u\right) \rightarrow 0,
\end{equation*}%
as $\phi _{1}^{\left( p\right) }\rightarrow \infty $. Thus, by Markov
inequality, under (\ref{r1})%
\begin{eqnarray*}
&&\Theta_{1}^{\left( p \right)} = \int_{-\infty }^{\infty }\left( \frac{\sum_{j=1}^{R_{1}}\left( \zeta
_{1,j}^{\left( p\right) }\left( u\right) -G_{1}\left( 0\right) \right) }{%
\sqrt{R_{1}}\sqrt{G_{1}\left( 0\right) \left( 1-G_{1}\left( 0\right) \right) 
}}\right) ^{2}dF_{1}\left( u\right) \\
&=&\int_{-\infty }^{\infty }\left( \frac{\sum_{j=1}^{R_{1}}\left( I\left(
\xi _{1,j}^{\left( p\right) }\leq 0\right) -G_{1}\left( 0\right) \right) }{%
\sqrt{R_{1}}\sqrt{G_{1}\left( 0\right) \left( 1-G_{1}\left( 0\right) \right) 
}}\right) ^{2}dF_{1}\left( u\right) +o_{P^{\ast }}\left( 1\right) \overset{%
D^{\ast }}{\rightarrow }\chi^{2}_{1} ,
\end{eqnarray*}%
with the last passage following from the CLT for Bernoulli random variables and continuity.
This proves (\ref{null-1}).

We now turn to (\ref{alt-1}). By (\ref{lambda-x-notrend}) and Lemma \ref%
{trapani2017}, we have that, under $H_{A,1}^{\left( p\right) }$ 
\begin{equation*}
P\left\{ \omega :\lim_{\min \left( N,T\right) \rightarrow \infty }\phi
_{1}^{\left( p\right) }=1\right\} =1,
\end{equation*}%
and therefore we can henceforth assume that 
\begin{equation}
\lim_{\min \left( N,T\right) \rightarrow \infty }\phi _{1}^{\left( p\right)
}=1.  \label{lim-phi-alt}
\end{equation}%
We can write%
\begin{equation*}
\zeta _{1,j}^{\left( p\right) }\left( u\right) -G_{1}\left( 0\right) =\zeta
_{1,j}^{\left( p\right) }\left( u\right) -G_{1}\left( 0\right) \pm
G_{1}\left( u/\phi _{1}^{\left( p\right) }\right) ,
\end{equation*}%
so that%
\begin{eqnarray*}
&&E^{\ast }\int_{-\infty }^{\infty }\left( \frac{1}{\sqrt{R_{1}}}%
\sum_{j=1}^{R_{1}}\zeta _{1,j}^{\left( p\right) }\left( u\right)
-G_{1}\left( 0\right) \right) ^{2}dF_{1}\left( u\right)  \\
&=&E^{\ast }\int_{-\infty }^{\infty }\left( \frac{1}{\sqrt{R_{1}}}%
\sum_{j=1}^{R_{1}}\zeta _{1,j}^{\left( p\right) }\left( u\right)
-G_{1}\left( u/\phi _{1}^{\left( p\right) }\right) \right) ^{2}dF_{1}\left(
u\right) +R_{1}\int_{-\infty }^{\infty }\left( G_{1}\left( u/\phi
_{1}^{\left( p\right) }\right) -G_{1}\left( 0\right) \right)
^{2}dF_{1}\left( u\right)  \\
&=&\int_{-\infty }^{\infty }V^{\ast }\left( \zeta _{1,j}^{\left( p\right)
}\left( u\right) \right) dF_{1}\left( u\right) +R_{1}\int_{-\infty }^{\infty
}\left( G_{1}\left( u/\phi _{1}^{\left( p\right) }\right) -G_{1}\left(
0\right) \right) ^{2}dF_{1}\left( u\right) ,
\end{eqnarray*}%
having used again the independence of the $\zeta _{1,j}^{\left( p\right)
}\left( u\right) $. Clearly, $V^{\ast }\left( \zeta _{1,j}^{\left( p\right)
}\left( u\right) \right) <\infty $; also, as $\min \left( N,T\right)
\rightarrow \infty $, (\ref{lim-phi-alt}) yields 
\begin{equation*}
\int_{-\infty }^{\infty }\left( G_{1}\left( u/\phi _{1}^{\left( p\right)
}\right) -G_{1}\left( 0\right) \right) ^{2}dF_{1}\left( u\right)
=\int_{-\infty }^{\infty }\left( G_{1}\left( u\right) -G_{1}\left( 0\right)
\right) ^{2}dF_{1}\left( u\right) ,
\end{equation*}%
so that, finally%
\begin{equation*}
\frac{1}{R_{1}}\Theta^{\left( p \right)}_{1} = \frac{1}{R_{1}}\int_{-\infty }^{\infty }\left( \frac{\sum_{j=1}^{R_{1}}%
\left( \zeta _{1,j}^{\left( p\right) }\left( u\right) -G_{1}\left( 0\right)
\right) }{\sqrt{R_{1}}\sqrt{G_{1}\left( 0\right) \left( 1-G_{1}\left(
0\right) \right) }}\right) ^{2}dF_{1}\left( u\right) =\frac{\int_{-\infty
}^{\infty }\left( G_{1}\left( u\right) -G_{1}\left( 0\right) \right)
^{2}dF_{1}\left( u\right) }{G_{1}\left( 0\right) \left( 1-G_{1}\left(
0\right) \right) }+o\left( 1\right) .
\end{equation*}%
\end{proof}

\begin{proof}[Proof of Lemma \protect\ref{test-1-algo}]
Let $Z$\ be a $N\left( 0,1\right) $ random variable. By (\ref{null-1}),
using Bernstein concentration inequality we have that 
\begin{equation}
P^{\ast }\left( \Theta _{1}^{\left( p\right) }>c_{\alpha ,1}\right) =P^{\ast
}\left( Z^{2}>c_{\alpha ,1}\right) +o_{P^{\ast }}\left( 1\right) \leq 2\exp
\left( -\frac{1}{2}c_{\alpha ,1}\right) +o_{P^{\ast }}\left( 1\right) ,
\label{concentration}
\end{equation}%
which implies that $P^{\ast }\left( \Theta _{1}^{\left( p\right) }>c_{\alpha
,1}\right) $ drifts to zero as long as $c_{\alpha ,1}\rightarrow \infty $.
Therefore, under $H_{0,1}^{\left( 1\right) }$, there is zero probability of a
Type I error. Under $H_{A,1}^{\left( 1\right) }$, by (\ref{alt-1}) we have%
\begin{eqnarray*}
P^{\ast }\left( \Theta _{1}^{\left( p\right) }\leq c_{\alpha ,1}\right) 
&=&P^{\ast }\left[ \left( Z+C_{0}\sqrt{R_{1}}\right) ^{2}\leq c_{\alpha ,1}%
\right] +o_{P^{\ast }}\left( 1\right)  \\
&\leq&P^{\ast }\left( \left\vert Z \right\vert \leq \left\vert c_{\alpha ,1}\right\vert ^{1/2}-C_{0}%
\sqrt{R_{1}}\right) +o_{P^{\ast }}\left( 1\right)  \\
&\rightarrow&P^{\ast }\left( \left\vert Z \right\vert \leq -\infty \right)=0,
\end{eqnarray*}%
since $c_{\alpha ,1}=o\left( R_{1}\right) $. Thus, under the alternative
there is zero probability of a Type II\ error. This proves the desired
result. 
\end{proof}

\begin{proof}[Proof of Theorem \protect \ref{test-2}]
The proof is exactly the same as the proof of Theorem  \ref{test-1}.
\end{proof}

\begin{proof}[Proof of Lemma \protect\ref{test-2-algo}]
The proof is exactly the same as the proof of Theorem 3 in \citet{trapani17}.
\end{proof}
\clearpage
\section{Additional numerical results}
\subsection{The case $\bar \rho =0$}
\begin{table}[h!]
\centering
\caption{Average estimated number of factors with linear trend, $\wh{r}_1$.}\label{tab:case1b}
\vskip .2cm
\footnotesize{
\begin{tabular}{ll | ccc | ccc | ccc}
\hline
\hline
&& \multicolumn{3}{|c|}{$N=50$, $T=100$}&\multicolumn{3}{c|}{$N=100$, $T=100$}&\multicolumn{3}{c}{$N=200$, $T=100$}\\
$r_1\,$ & $r_2$ & $BT1$ & $BT2$ & $BT3$& $BT1$ & $BT2$ & $BT3$& $BT1$ & $BT2$ & $BT3$\\
\hline
0	&	0	&	0.00	&	0.00	&	0.00	&	0.00	&	0.00	&	0.00	&	0.00	&	0.00	&	0.00	\\
0	&	1	&	0.07	&	0.07	&	0.08	&	0.03	&	0.03	&	0.27	&	0.02	&	0.02	&	0.16	\\
0	&	2	&	0.00	&	0.01	&	0.00	&	0.00	&	0.00	&	0.04	&	0.00	&	0.00	&	0.02	\\
1	&	0	&	1.00	&	1.00	&	1.00	&	1.00	&	1.00	&	1.00	&	1.00	&	1.00	&	1.00	\\
1	&	1	&	1.00	&	1.00	&	1.00	&	1.00	&	1.00	&	1.00	&	1.00	&	1.00	&	1.00	\\
1	&	2	&	1.00	&	1.00	&	1.00	&	1.00	&	1.00	&	1.00	&	1.00	&	1.00	&	1.00	\\
\hline
\hline
&& \multicolumn{3}{|c|}{$N=100$, $T=200$}&\multicolumn{3}{c|}{$N=200$, $T=200$}&\multicolumn{3}{c}{$N=200$, $T=500$}\\
$r_1\,$ & $r_2$ & $BT1$ & $BT2$ & $BT3$& $BT1$ & $BT2$ & $BT3$& $BT1$ & $BT2$ & $BT3$\\
\hline
0	&	0	&	0.00	&	0.00	&	0.00	&	0.00	&	0.00	&	0.00	&	0.00	&	0.00	&	0.00	\\
0	&	1	&	0.01	&	0.01	&	0.18	&	0.00	&	0.00	&	0.14	&	0.00	&	0.00	&	0.05	\\
0	&	2	&	0.00	&	0.00	&	0.01	&	0.00	&	0.00	&	0.01	&	0.00	&	0.00	&	0.00	\\
1	&	0	&	1.00	&	1.00	&	1.00	&	1.00	&	1.00	&	1.00	&	1.00	&	1.00	&	1.00	\\
1	&	1	&	1.00	&	1.00	&	1.00	&	1.00	&	1.00	&	1.00	&	1.00	&	1.00	&	1.00	\\
1	&	2	&	1.00	&	1.00	&	1.00	&	1.00	&	1.00	&	1.00	&	1.00	&	1.00	&	1.00	\\
\hline
\hline
\end{tabular}
}
\end{table}

\begin{table}[h!]
\caption{Average estimated number of zero-mean $I(1)$ factors, $\wh{r}_2$, when $r_1=0$.}\label{tab:case2b}
\centering
\vskip .2cm
\footnotesize{
\begin{tabular}{ll |cccc | cccc | cccc  }
\hline
\hline
&& \multicolumn{4}{|c}{$N=50$, $T=100$}& \multicolumn{4}{|c}{$N=100$, $T=100$}& \multicolumn{4}{|c}{$N=200$, $T=100$}\\
$r_2\,$ & $r_3$ & $BT1$ & $BT2$ & $BT3$& $IC$& $BT1$ & $BT2$ & $BT3$& $IC$& $BT1$ & $BT2$ & $BT3$& $IC$\\
\hline
0	&	0	&	0.00	&	0.00	&	0.00	&	1.00	&	0.00	&	0.00	&	0.00	&	1.00	&	0.00	&	0.00	&	0.00	&	1.00	\\
0	&	1	&	0.00	&	0.00	&	0.00	&	1.00	&	0.00	&	0.00	&	0.00	&	1.00	&	0.00	&	0.00	&	0.00	&	1.00	\\
0	&	2	&	0.00	&	0.00	&	0.00	&	1.00	&	0.00	&	0.00	&	0.00	&	1.00	&	0.00	&	0.00	&	0.00	&	1.00	\\
1	&	0	&	1.00	&	1.00	&	1.00	&	1.00	&	1.00	&	1.00	&	1.00	&	1.00	&	1.00	&	1.00	&	1.00	&	1.00	\\
1	&	1	&	1.00	&	1.00	&	1.00	&	1.00	&	1.00	&	1.00	&	1.00	&	1.00	&	1.00	&	0.99	&	1.00	&	1.00	\\
1	&	2	&	1.00	&	1.00	&	1.00	&	1.00	&	0.99	&	0.99	&	1.00	&	1.00	&	0.99	&	0.99	&	1.00	&	1.00	\\
2	&	0	&	1.97	&	1.99	&	1.98	&	2.00	&	1.96	&	1.99	&	1.99	&	2.00	&	1.94	&	1.99	&	1.99	&	2.00	\\
2	&	1	&	1.86	&	1.87	&	1.85	&	2.00	&	1.91	&	1.98	&	1.99	&	1.99	&	1.84	&	1.93	&	1.98	&	2.00	\\
2	&	2	&	1.91	&	1.91	&	1.91	&	1.99	&	1.90	&	1.97	&	1.99	&	2.00	&	1.86	&	1.94	&	1.99	&	2.00	\\
\hline
\hline
&& \multicolumn{4}{|c}{$N=100$, $T=200$}& \multicolumn{4}{|c}{$N=200$, $T=200$}& \multicolumn{4}{|c}{$N=200$, $T=500$}\\
$r_2\,$ & $r_3$ & $BT1$ & $BT2$ & $BT3$& $IC$& $BT1$ & $BT2$ & $BT3$& $IC$& $BT1$ & $BT2$ & $BT3$& $IC$\\
\hline
0	&	0	&	0.00	&	0.00	&	0.00	&	1.00	&	0.00	&	0.00	&	0.00	&	1.00	&	0.00	&	0.00	&	0.00	&	1.00	\\
0	&	1	&	0.00	&	0.00	&	0.00	&	1.00	&	0.00	&	0.00	&	0.00	&	1.00	&	0.00	&	0.00	&	0.00	&	1.00	\\
0	&	2	&	0.00	&	0.00	&	0.00	&	1.00	&	0.00	&	0.00	&	0.00	&	1.00	&	0.00	&	0.00	&	0.00	&	1.00	\\
1	&	0	&	1.00	&	1.00	&	1.00	&	1.00	&	1.00	&	1.00	&	1.00	&	1.00	&	1.00	&	1.00	&	1.00	&	1.00	\\
1	&	1	&	1.00	&	1.00	&	1.00	&	1.00	&	1.00	&	1.00	&	1.00	&	1.00	&	1.00	&	1.00	&	1.00	&	1.00	\\
1	&	2	&	1.00	&	1.00	&	1.00	&	1.00	&	1.00	&	1.00	&	1.00	&	1.00	&	1.00	&	1.00	&	1.00	&	1.00	\\
2	&	0	&	2.00	&	2.00	&	2.00	&	2.00	&	1.98	&	2.00	&	2.00	&	2.00	&	1.99	&	2.00	&	2.00	&	2.00	\\
2	&	1	&	1.97	&	1.99	&	2.00	&	2.00	&	1.97	&	1.99	&	2.00	&	2.00	&	2.00	&	2.00	&	2.00	&	2.00	\\
2	&	2	&	1.98	&	1.99	&	2.00	&	2.00	&	1.97	&	1.98	&	2.00	&	2.00	&	1.99	&	1.99	&	2.00	&	2.00	\\
\hline
\hline
\end{tabular}
}
\end{table}

\begin{table}[h!]
\centering
\caption{Average estimated number of zero-mean $I(1)$ factors, $\wh{r}_2$, when $r_1=1$.}\label{tab:case3b}
\vskip .2cm
\footnotesize{
\begin{tabular}{ll |cccc | cccc | cccc  }
\hline
\hline
&& \multicolumn{4}{|c}{$N=50$, $T=100$}& \multicolumn{4}{|c}{$N=100$, $T=100$}& \multicolumn{4}{|c}{$N=200$, $T=100$}\\
$r_2\,$ & $r_3$ & $BT1$ & $BT2$ & $BT3$& $IC$& $BT1$ & $BT2$ & $BT3$& $IC$& $BT1$ & $BT2$ & $BT3$& $IC$\\
\hline
0	&	0	&	0.00	&	0.00	&	0.00	&		0.00	&	0.00	&	0.00	&	0.00	&		0.00	&	0.00	&	0.00	&	0.00	&		0.00	\\
0	&	1	&	0.00	&	0.00	&	0.00	&		0.00	&	0.00	&	0.00	&	0.00	&		0.00	&	0.00	&	0.00	&	0.00	&		0.00	\\
0	&	2	&	0.00	&	0.00	&	0.00	&		0.00	&	0.00	&	0.00	&	0.00	&		0.00	&	0.00	&	0.00	&	0.00	&		0.00	\\
1	&	0	&	0.99	&	1.00	&	1.00	&		1.00	&	0.96	&	0.95	&	0.94	&		1.00	&	1.00	&	0.99	&	1.00	&		1.00	\\
1	&	1	&	0.98	&	0.98	&	1.00	&		1.00	&	0.88	&	0.90	&	0.88	&		1.00	&	0.95	&	0.96	&	0.99	&		1.00	\\
1	&	2	&	0.97	&	0.95	&	0.99	&		0.99	&	0.90	&	0.90	&	0.88	&		1.00	&	0.97	&	0.96	&	0.99	&		1.00	\\
2	&	0	&	1.72	&	1.92	&	1.96	&		1.92	&	1.20	&	1.20	&	1.16	&		1.96	&	1.66	&	1.88	&	1.93	&		1.97	\\
2	&	1	&	1.29	&	1.52	&	1.77	&		1.60	&	1.01	&	0.99	&	0.99	&		1.82	&	1.20	&	1.43	&	1.72	&		1.90	\\
2	&	2	&	1.30	&	1.43	&	1.65	&		1.56	&	0.92	&	0.96	&	0.93	&		1.73	&	1.21	&	1.35	&	1.55	&		1.84	\\
\hline
\hline
&& \multicolumn{4}{|c}{$N=100$, $T=200$}& \multicolumn{4}{|c}{$N=200$, $T=200$}& \multicolumn{4}{|c}{$N=200$, $T=500$}\\
$r_2\,$ & $r_3$ & $BT1$ & $BT2$ & $BT3$& $IC$& $BT1$ & $BT2$ & $BT3$& $IC$& $BT1$ & $BT2$ & $BT3$& $IC$\\
\hline
0	&	0	&	0.00	&	0.00	&	0.00	&		0.00	&	0.00	&	0.00	&	0.00	&		0.00	&	0.00	&	0.00	&	0.00	&		0.00	\\
0	&	1	&	0.00	&	0.00	&	0.00	&		0.00	&	0.00	&	0.00	&	0.00	&		0.00	&	0.00	&	0.00	&	0.00	&		0.00	\\
0	&	2	&	0.00	&	0.00	&	0.01	&		0.00	&	0.00	&	0.00	&	0.00	&		0.00	&	0.00	&	0.00	&	0.00	&		0.00	\\
1	&	0	&	0.99	&	1.00	&	1.00	&		1.00	&	1.00	&	1.00	&	1.00	&		1.00	&	1.00	&	1.00	&	1.00	&		1.00	\\
1	&	1	&	0.99	&	1.00	&	1.00	&		1.00	&	1.00	&	1.00	&	1.00	&		1.00	&	1.00	&	1.00	&	1.00	&		1.00	\\
1	&	2	&	1.00	&	1.00	&	1.00	&		1.00	&	1.00	&	1.00	&	1.00	&		1.00	&	1.00	&	1.00	&	1.00	&		1.00	\\
2	&	0	&	1.91	&	1.99	&	1.99	&		1.98	&	1.89	&	1.97	&	1.99	&		1.99	&	1.98	&	1.99	&	1.99	&		1.99	\\
2	&	1	&	1.70	&	1.86	&	1.94	&		1.89	&	1.62	&	1.81	&	1.92	&		1.95	&	1.90	&	1.96	&	2.00	&		2.00	\\
2	&	2	&	1.67	&	1.81	&	1.91	&		1.84	&	1.61	&	1.75	&	1.86	&		1.93	&	1.88	&	1.95	&	1.98	&		1.98	\\
\hline
\hline
\end{tabular}
}
\end{table}
\clearpage
\subsection{The case $\bar \rho =0.8$}
\begin{table}[h!]
\centering
\caption{Average estimated number of factors with linear trend, $\wh{r}_1$.}\label{tab:case1c}
\vskip .2cm
\footnotesize{
\begin{tabular}{ll | ccc | ccc | ccc}
\hline
\hline
&& \multicolumn{3}{|c|}{$N=50$, $T=100$}&\multicolumn{3}{c|}{$N=100$, $T=100$}&\multicolumn{3}{c}{$N=200$, $T=100$}\\
$r_1\,$ & $r_2$ & $BT1$ & $BT2$ & $BT3$& $BT1$ & $BT2$ & $BT3$& $BT1$ & $BT2$ & $BT3$\\
\hline
0	&	0	&	0.00	&	0.00	&	0.00	&	0.00	&	0.00	&	0.00	&	0.00	&	0.00	&	0.00	\\
0	&	1	&	0.40	&	0.40	&	0.40	&	0.31	&	0.30	&	0.67	&	0.20	&	0.21	&	0.57	\\
0	&	2	&	0.30	&	0.28	&	0.31	&	0.16	&	0.15	&	0.46	&	0.07	&	0.07	&	0.34	\\
1	&	0	&	1.00	&	1.00	&	1.00	&	1.00	&	1.00	&	1.00	&	1.00	&	1.00	&	1.00	\\
1	&	1	&	1.00	&	1.00	&	1.00	&	1.00	&	1.00	&	1.00	&	0.99	&	0.99	&	1.00	\\
1	&	2	&	1.00	&	1.00	&	1.00	&	1.00	&	1.00	&	1.00	&	1.00	&	1.00	&	1.00	\\
\hline
\hline
&& \multicolumn{3}{|c|}{$N=100$, $T=200$}&\multicolumn{3}{c|}{$N=200$, $T=200$}&\multicolumn{3}{c}{$N=200$, $T=500$}\\
$r_1\,$ & $r_2$ & $BT1$ & $BT2$ & $BT3$& $BT1$ & $BT2$ & $BT3$& $BT1$ & $BT2$ & $BT3$\\
\hline
0	&	0	&	0.00	&	0.00	&	0.00	&	0.00	&	0.00	&	0.00	&	0.00	&	0.00	&	0.00	\\
0	&	1	&	0.22	&	0.21	&	0.56	&	0.17	&	0.18	&	0.44	&	0.09	&	0.09	&	0.32	\\
0	&	2	&	0.07	&	0.08	&	0.30	&	0.05	&	0.05	&	0.22	&	0.00	&	0.00	&	0.07	\\
1	&	0	&	1.00	&	1.00	&	1.00	&	1.00	&	1.00	&	1.00	&	1.00	&	1.00	&	1.00	\\
1	&	1	&	1.00	&	1.00	&	1.00	&	1.00	&	1.00	&	1.00	&	1.00	&	1.00	&	1.00	\\
1	&	2	&	1.00	&	1.00	&	1.00	&	1.00	&	1.00	&	1.00	&	1.00	&	1.00	&	1.00	\\
\hline
\hline
\end{tabular}
}
\end{table}

\begin{table}[h!]
\caption{Average estimated number of zero-mean $I(1)$ factors, $\wh{r}_2$, when $r_1=0$.}\label{tab:case2c}
\centering
\vskip .2cm
\footnotesize{
\begin{tabular}{ll |cccc | cccc | cccc  }
\hline
\hline
&& \multicolumn{4}{|c}{$N=50$, $T=100$}& \multicolumn{4}{|c}{$N=100$, $T=100$}& \multicolumn{4}{|c}{$N=200$, $T=100$}\\
$r_2\,$ & $r_3$ & $BT1$ & $BT2$ & $BT3$& $IC$& $BT1$ & $BT2$ & $BT3$& $IC$& $BT1$ & $BT2$ & $BT3$& $IC$\\
\hline
0	&	0	&	0.00	&	0.00	&	0.00	&		1.00	&	0.00	&	0.00	&	0.00	&		1.00	&	0.00	&	0.00	&	0.00	&		1.00	\\
0	&	1	&	0.00	&	0.00	&	0.00	&		1.00	&	0.00	&	0.00	&	0.00	&		1.00	&	0.00	&	0.00	&	0.00	&		1.00	\\
0	&	2	&	0.00	&	0.00	&	0.00	&		1.00	&	0.00	&	0.00	&	0.00	&		1.00	&	0.00	&	0.00	&	0.00	&		1.00	\\
1	&	0	&	1.00	&	1.00	&	1.00	&		1.00	&	1.00	&	1.00	&	1.00	&		1.00	&	1.00	&	1.00	&	1.00	&		1.00	\\
1	&	1	&	1.00	&	1.00	&	1.00	&		1.00	&	1.00	&	1.00	&	1.00	&		1.00	&	1.00	&	1.00	&	1.00	&		1.00	\\
1	&	2	&	1.00	&	1.00	&	1.00	&		1.00	&	1.00	&	1.00	&	1.00	&		1.00	&	1.00	&	1.00	&	1.00	&		1.00	\\
2	&	0	&	2.00	&	2.00	&	2.00	&		2.00	&	1.99	&	2.00	&	2.00	&		2.00	&	1.99	&	2.00	&	1.99	&		2.00	\\
2	&	1	&	2.00	&	1.99	&	1.99	&		2.00	&	1.99	&	2.00	&	2.00	&		2.00	&	1.99	&	2.00	&	2.00	&		2.00	\\
2	&	2	&	1.99	&	2.00	&	2.00	&		2.00	&	2.00	&	2.00	&	2.00	&		2.00	&	2.00	&	2.00	&	2.00	&		2.00	\\
\hline
\hline
&& \multicolumn{4}{|c}{$N=100$, $T=200$}& \multicolumn{4}{|c}{$N=200$, $T=200$}& \multicolumn{4}{|c}{$N=200$, $T=500$}\\
$r_2\,$ & $r_3$ & $BT1$ & $BT2$ & $BT3$& $IC$& $BT1$ & $BT2$ & $BT3$& $IC$& $BT1$ & $BT2$ & $BT3$& $IC$\\
\hline
0	&	0	&	0.00	&	0.00	&	0.00	&		1.00	&	0.00	&	0.00	&	0.00	&		1.00	&	0.00	&	0.00	&	0.00	&		1.00	\\
0	&	1	&	0.00	&	0.00	&	0.00	&		1.00	&	0.00	&	0.00	&	0.00	&		1.00	&	0.00	&	0.00	&	0.00	&		1.00	\\
0	&	2	&	0.00	&	0.00	&	0.00	&		1.00	&	0.00	&	0.00	&	0.00	&		1.00	&	0.00	&	0.00	&	0.00	&		1.00	\\
1	&	0	&	1.00	&	1.00	&	1.00	&		1.00	&	1.00	&	1.00	&	1.00	&		1.00	&	1.00	&	1.00	&	1.00	&		1.00	\\
1	&	1	&	1.00	&	1.00	&	1.00	&		1.00	&	1.00	&	1.00	&	1.00	&		1.00	&	1.00	&	1.00	&	1.00	&		1.00	\\
1	&	2	&	1.00	&	1.00	&	1.00	&		1.00	&	1.00	&	1.00	&	1.00	&		1.00	&	1.00	&	1.00	&	1.00	&		1.00	\\
2	&	0	&	2.00	&	2.00	&	2.00	&		2.00	&	2.00	&	2.00	&	2.00	&		2.00	&	2.00	&	2.00	&	2.00	&		2.00	\\
2	&	1	&	2.00	&	2.00	&	2.00	&		2.00	&	2.00	&	2.00	&	2.00	&		2.00	&	2.00	&	2.00	&	2.00	&		2.00	\\
2	&	2	&	2.00	&	2.00	&	2.00	&		2.00	&	2.00	&	2.00	&	2.00	&		2.00	&	2.00	&	2.00	&	2.00	&		2.00	\\
\hline
\hline
\end{tabular}
}
\end{table}

\begin{table}[h!]
\centering
\caption{Average estimated number of zero-mean $I(1)$ factors, $\wh{r}_2$, when $r_1=1$.}\label{tab:case3c}
\vskip .2cm
\footnotesize{
\begin{tabular}{ll |cccc | cccc | cccc  }
\hline
\hline
&& \multicolumn{4}{|c}{$N=50$, $T=100$}& \multicolumn{4}{|c}{$N=100$, $T=100$}& \multicolumn{4}{|c}{$N=200$, $T=100$}\\
$r_2\,$ & $r_3$ & $BT1$ & $BT2$ & $BT3$& $IC$& $BT1$ & $BT2$ & $BT3$& $IC$& $BT1$ & $BT2$ & $BT3$& $IC$\\
\hline
0	&	0	&	0.00	&	0.00	&	0.00	&		0.00	&	0.00	&	0.00	&	0.00	&		0.00	&	0.00	&	0.00	&	0.00	&		0.00	\\
0	&	1	&	0.00	&	0.00	&	0.00	&		0.00	&	0.00	&	0.00	&	0.00	&		0.00	&	0.00	&	0.00	&	0.00	&		0.00	\\
0	&	2	&	0.00	&	0.00	&	0.00	&		0.00	&	0.00	&	0.00	&	0.00	&		0.00	&	0.00	&	0.00	&	0.00	&		0.00	\\
1	&	0	&	1.00	&	1.00	&	1.00	&		1.00	&	1.00	&	1.00	&	1.00	&		1.00	&	1.00	&	1.00	&	1.00	&		1.00	\\
1	&	1	&	1.00	&	1.00	&	1.00	&		1.00	&	1.00	&	1.00	&	1.00	&		1.00	&	1.00	&	1.00	&	1.00	&		1.00	\\
1	&	2	&	1.00	&	0.99	&	1.00	&		1.00	&	1.00	&	1.00	&	1.00	&		1.00	&	1.00	&	1.00	&	0.99	&		1.00	\\
2	&	0	&	1.96	&	1.99	&	2.00	&		1.99	&	1.96	&	1.98	&	1.99	&		2.00	&	1.93	&	1.97	&	1.99	&		1.99	\\
2	&	1	&	1.85	&	1.91	&	1.98	&		1.94	&	1.84	&	1.94	&	1.98	&		1.97	&	1.82	&	1.93	&	1.97	&		1.99	\\
2	&	2	&	1.84	&	1.91	&	1.94	&		1.93	&	1.89	&	1.93	&	1.96	&		1.96	&	1.81	&	1.89	&	1.95	&		1.98	\\
\hline
\hline
&& \multicolumn{4}{|c}{$N=100$, $T=200$}& \multicolumn{4}{|c}{$N=200$, $T=200$}& \multicolumn{4}{|c}{$N=200$, $T=500$}\\
$r_2\,$ & $r_3$ & $BT1$ & $BT2$ & $BT3$& $IC$& $BT1$ & $BT2$ & $BT3$& $IC$& $BT1$ & $BT2$ & $BT3$& $IC$\\
\hline
0	&	0	&	0.00	&	0.00	&	0.00	&		0.00	&	0.00	&	0.00	&	0.00	&		0.00	&	0.00	&	0.00	&	0.00	&		0.00	\\
0	&	1	&	0.00	&	0.00	&	0.00	&		0.00	&	0.00	&	0.00	&	0.00	&		0.00	&	0.00	&	0.00	&	0.00	&		0.00	\\
0	&	2	&	0.00	&	0.00	&	-0.01	&		0.00	&	0.00	&	0.00	&	0.00	&		0.00	&	0.00	&	0.00	&	0.00	&		0.00	\\
1	&	0	&	0.99	&	1.00	&	1.00	&		1.00	&	1.00	&	1.00	&	1.00	&		1.00	&	1.00	&	1.00	&	1.00	&		1.00	\\
1	&	1	&	1.00	&	1.00	&	1.00	&		1.00	&	1.00	&	1.00	&	1.00	&		1.00	&	1.00	&	1.00	&	1.00	&		1.00	\\
1	&	2	&	1.00	&	1.00	&	1.00	&		1.00	&	1.00	&	1.00	&	1.00	&		1.00	&	1.00	&	1.00	&	1.00	&		1.00	\\
2	&	0	&	2.00	&	2.00	&	1.99	&		2.00	&	1.98	&	2.00	&	2.00	&		2.00	&	2.00	&	2.00	&	2.00	&		2.00	\\
2	&	1	&	1.98	&	1.99	&	2.00	&		2.00	&	1.98	&	1.99	&	2.00	&		2.00	&	2.00	&	1.99	&	2.00	&		2.00	\\
2	&	2	&	1.98	&	1.99	&	2.00	&		2.00	&	1.97	&	1.98	&	1.99	&		2.00	&	1.99	&	2.00	&	1.99	&		2.00	\\
\hline
\hline
\end{tabular}
}
\end{table}
\clearpage
\subsection{Robustness analysis}

\begin{table}[h!]
\caption{Fraction of wrong detections when estimating $r_1$ }\label{tab:rob1}
\vskip -.5cm
\phantom{Table 13: Fraction of wrong }  $N=200$, $T=200$
\centering
\vskip .2cm
\footnotesize{
\begin{tabular}{lll |p{.09\textwidth}p{.09\textwidth}p{.09\textwidth} | p{.09\textwidth}p{.09\textwidth}p{.09\textwidth}}
\hline
\hline
&&&\multicolumn{3}{c|}{$R_1=N$, $u= \sqrt{2}$, $\delta^*=10^{-5}$}&\multicolumn{3}{c}{$R_1=N$, $u= \sqrt{2}$, $\delta^*=10^{-1}$}\\
$r_1$ & $r_2\,$ & $\wh{r}_1$ & $BT1$ & $BT2$ & $BT3$ & $BT1$ & $BT2$ & $BT3$\\
\hline
1&	0&	0&		0.000&	0.000&	0.000&		0.000&	0.000&	0.000\\
1&	1&	0&		0.000&	0.000&	0.000&		0.000&	0.000&	0.000\\
1&	2&	0&		0.000&	0.000&	0.002&		0.000&	0.002&	0.002\\
0&	0&	1&		0.000&	0.000&	0.000&		0.000&	0.000&	0.000\\
0&	1&	1&		0.022&	0.024&	0.252&		0.002&	0.002&	0.104\\
0&	2&	1&		0.000&	0.000&	0.026&		0.000&	0.000&	0.002\\
\hline
&&&\multicolumn{3}{c|}{$R_1=N$, $u=5$, $\delta^*=10^{-5}$}&\multicolumn{3}{c}{$R_1=N$, $u= 5$, $\delta^*=10^{-1}$}\\
$r_1$ & $r_2\,$ & $\wh{r}_1$ & $BT1$ & $BT2$ & $BT3$ & $BT1$ & $BT2$ & $BT3$\\
\hline
1&	0&	0&		0.000&	0.000&	0.000&		0.000&	0.000&	0.000\\
1&	1&	0&		0.000&	0.000&	0.000&		0.398&	0.412&	0.000\\
1&	2&	0&		0.000&	0.002&	0.002&		0.412&	0.388&	0.002\\
0&	0&	1&		0.000&	0.000&	0.000&		0.000&	0.000&	0.000\\
0&	1&	1&		0.000&	0.000&	0.086&		0.000&	0.000&	0.024\\
0&	2&	1&		0.000&	0.000&	0.000&		0.000&	0.000&	0.000\\
\hline
&&&\multicolumn{3}{c|}{$R_1=\lfloor N/2\rfloor$, $u=\sqrt{2}$, $\delta^*=10^{-5}$}&\multicolumn{3}{c}{$R_1=\lfloor N/2\rfloor$, $u= \sqrt 2$, $\delta^*=10^{-1}$}\\
$r_1$ & $r_2\,$ & $\wh{r}_1$ & $BT1$ & $BT2$ & $BT3$ & $BT1$ & $BT2$ & $BT3$\\
\hline
1&	0&	0&		0.000&	0.000&	0.000&		0.000&	0.000&	0.000\\
1&	1&	0&		0.000&	0.002&	0.000&		0.000&	0.000&	0.000\\
1&	2&	0&		0.000&	0.000&	0.000&		0.000&	0.002&	0.000\\
0&	0&	1&		0.000&	0.000&	0.000&		0.000&	0.000&	0.000\\
0&	1&	1&		0.068&	0.062&	0.326&		0.018&	0.012&	0.192\\
0&	2&	1&		0.002&	0.004&	0.050&		0.000&	0.000&	0.010\\
\hline
&&&\multicolumn{3}{c|}{$R_1=\lfloor N/2\rfloor$, $u=5$, $\delta^*=10^{-5}$}&\multicolumn{3}{c}{$R_1=\lfloor N/2\rfloor$, $u= 5$, $\delta^*=10^{-1}$}\\
$r_1$ & $r_2\,$ & $\wh{r}_1$ & $BT1$ & $BT2$ & $BT3$ & $BT1$ & $BT2$ & $BT3$\\
\hline
1&	0&	0&		0.000&	0.000&	0.000&		0.000&	0.000&	0.000\\
1&	1&	0&		0.000&	0.000&	0.000&		0.134&	0.148&	0.000\\
1&	2&	0&		0.000&	0.002&	0.000&		0.124&	0.162&	0.000\\
0&	0&	1&		0.000&	0.000&	0.000&		0.000&	0.000&	0.000\\
0&	1&	1&		0.006&	0.004&	0.128&		0.000&	0.000&	0.038\\
0&	2&	1&		0.000&	0.000&	0.008&		0.000&	0.000&	0.000\\
\hline
\hline
\end{tabular}
}
\end{table}

\begin{table}[h!]
\caption{Fraction of wrong detections when estimating $r_1$ }\label{tab:rob2}
\vskip -.5cm
\phantom{Table 13: Fraction of wrong }  $N=200$, $T=500$
\centering
\vskip .2cm
\footnotesize{
\begin{tabular}{lll |p{.09\textwidth}p{.09\textwidth}p{.09\textwidth} | p{.09\textwidth}p{.09\textwidth}p{.09\textwidth}}
\hline
\hline
&&&\multicolumn{3}{c|}{$R_1=N$, $u= \sqrt{2}$, $\delta^*=10^{-5}$}&\multicolumn{3}{c}{$R_1=N$, $u= \sqrt{2}$, $\delta^*=10^{-1}$}\\
$r_1$ & $r_2\,$ & $\wh{r}_1$ & $BT1$ & $BT2$ & $BT3$ & $BT1$ & $BT2$ & $BT3$\\
\hline
1	&	0	&	0	&	0.000	&	0.000	&	0.002	&	0.000	&	0.000	&	0.000	\\
1	&	1	&	0	&	0.000	&	0.000	&	0.002	&	0.000	&	0.002	&	0.000	\\
1	&	2	&	0	&	0.000	&	0.002	&	0.000	&	0.000	&	0.000	&	0.000	\\
0	&	0	&	1	&	0.000	&	0.000	&	0.000	&	0.000	&	0.000	&	0.000	\\
0	&	1	&	1	&	0.008	&	0.008	&	0.136	&	0.000	&	0.000	&	0.046	\\
0	&	2	&	1	&	0.000	&	0.000	&	0.008	&	0.000	&	0.000	&	0.000	\\
\hline
&&&\multicolumn{3}{c|}{$R_1=N$, $u=5$, $\delta^*=10^{-5}$}&\multicolumn{3}{c}{$R_1=N$, $u= 5$, $\delta^*=10^{-1}$}\\
$r_1$ & $r_2\,$ & $\wh{r}_1$ & $BT1$ & $BT2$ & $BT3$ & $BT1$ & $BT2$ & $BT3$\\
\hline
1	&	0	&	0	&	0.000	&	0.000	&	0.002	&	0.000	&	0.000	&	0.002	\\
1	&	1	&	0	&	0.000	&	0.000	&	0.002	&	0.000	&	0.000	&	0.002	\\
1	&	2	&	0	&	0.000	&	0.002	&	0.000	&	0.000	&	0.000	&	0.000	\\
0	&	0	&	1	&	0.000	&	0.000	&	0.000	&	0.000	&	0.000	&	0.000	\\
0	&	1	&	1	&	0.000	&	0.000	&	0.032	&	0.000	&	0.000	&	0.006	\\
0	&	2	&	1	&	0.000	&	0.000	&	0.000	&	0.000	&	0.000	&	0.000	\\
\hline
&&&\multicolumn{3}{c|}{$R_1=\lfloor N/2\rfloor$, $u=\sqrt{2}$, $\delta^*=10^{-5}$}&\multicolumn{3}{c}{$R_1=\lfloor N/2\rfloor$, $u= \sqrt 2$, $\delta^*=10^{-1}$}\\
$r_1$ & $r_2\,$ & $\wh{r}_1$ & $BT1$ & $BT2$ & $BT3$ & $BT1$ & $BT2$ & $BT3$\\
\hline
1	&	0	&	0	&	0.000	&	0.000	&	0.000	&	0.000	&	0.000	&	0.000	\\
1	&	1	&	0	&	0.000	&	0.000	&	0.000	&	0.000	&	0.002	&	0.000	\\
1	&	2	&	0	&	0.000	&	0.000	&	0.000	&	0.000	&	0.000	&	0.000	\\
0	&	0	&	1	&	0.000	&	0.000	&	0.000	&	0.000	&	0.000	&	0.000	\\
0	&	1	&	1	&	0.020	&	0.028	&	0.198	&	0.006	&	0.002	&	0.088	\\
0	&	2	&	1	&	0.000	&	0.000	&	0.008	&	0.000	&	0.000	&	0.000	\\
\hline
&&&\multicolumn{3}{c|}{$R_1=\lfloor N/2\rfloor$, $u=5$, $\delta^*=10^{-5}$}&\multicolumn{3}{c}{$R_1=\lfloor N/2\rfloor$, $u= 5$, $\delta^*=10^{-1}$}\\
$r_1$ & $r_2\,$ & $\wh{r}_1$ & $BT1$ & $BT2$ & $BT3$ & $BT1$ & $BT2$ & $BT3$\\
\hline
1	&	0	&	0	&	0.000	&	0.000	&	0.000	&	0.000	&	0.000	&	0.000	\\
1	&	1	&	0	&	0.000	&	0.000	&	0.000	&	0.000	&	0.002	&	0.000	\\
1	&	2	&	0	&	0.000	&	0.000	&	0.000	&	0.000	&	0.000	&	0.000	\\
0	&	0	&	1	&	0.000	&	0.000	&	0.000	&	0.000	&	0.000	&	0.000	\\
0	&	1	&	1	&	0.002	&	0.002	&	0.052	&	0.000	&	0.000	&	0.016	\\
0	&	2	&	1	&	0.000	&	0.000	&	0.000	&	0.000	&	0.000	&	0.000	\\
\hline
\hline
\end{tabular}
}
\end{table}

\begin{table}[h!]
\caption{Fraction of wrong detections when estimating $r_2$ }\label{tab:rob3}
\vskip -.5cm
\phantom{Table 13: Fraction of wrong }  $N=200$, $T=200$,  $r_1=0$ - part 1
\centering
\vskip .2cm
\footnotesize{
\begin{tabular}{lll |p{.09\textwidth}p{.09\textwidth}p{.09\textwidth} | p{.09\textwidth}p{.09\textwidth}p{.09\textwidth}}
\hline
\hline
&&&\multicolumn{3}{c|}{$R_2=N$ ($p=1$), $R_2=\lfloor N/3\rfloor$ ($p>1$) }&\multicolumn{3}{c}{$R_2=N$ ($p=1$), $R_2=\lfloor N/3\rfloor$ ($p>1$) }\\
&&&\multicolumn{3}{c|}{$u= \sqrt{2}$, $\delta^*=10^{-5}$}&\multicolumn{3}{c}{$u= \sqrt{2}$, $\delta^*=10^{-1}$}\\
$r_2$ & $r_3\,$ & $\wh{r}_2$ & $BT1$ & $BT2$ & $BT3$ & $BT1$ & $BT2$ & $BT3$\\
\hline
1	&	0	&	$<r_2$	&	0.000	&	0.000	&	0.000	&	0.000	&	0.000	&	0.000	\\
1	&	1	&	$<r_2$	&	0.000	&	0.002	&	0.000	&	0.002	&	0.000	&	0.000	\\
1	&	2	&	$<r_2$	&	0.000	&	0.000	&	0.000	&	0.000	&	0.000	&	0.000	\\
2	&	0	&	$<r_2$	&	0.002	&	0.002	&	0.002	&	0.024	&	0.002	&	0.002	\\
2	&	1	&	$<r_2$	&	0.014	&	0.004	&	0.004	&	0.040	&	0.012	&	0.006	\\
2	&	2	&	$<r_2$	&	0.002	&	0.002	&	0.000	&	0.020	&	0.006	&	0.002	\\
\hline
&&&\multicolumn{3}{c|}{$R_2=N$ ($p=1$), $R_2=\lfloor N/3\rfloor$ ($p>1$) }&\multicolumn{3}{c}{$R_2=N$ ($p=1$), $R_2=\lfloor N/3\rfloor$ ($p>1$) }\\
&&&\multicolumn{3}{c|}{$u= 5$, $\delta^*=10^{-5}$}&\multicolumn{3}{c}{$u=5$, $\delta^*=10^{-1}$}\\
$r_2$ & $r_3\,$ & $\wh{r}_2$ & $BT1$ & $BT2$ & $BT3$ & $BT1$ & $BT2$ & $BT3$\\
\hline
1	&	0	&	$<r_2$	&	0.000	&	0.000	&	0.000	&	0.002	&	0.002	&	0.000	\\
1	&	1	&	$<r_2$	&	0.002	&	0.002	&	0.000	&	0.018	&	0.018	&	0.002	\\
1	&	2	&	$<r_2$	&	0.000	&	0.000	&	0.000	&	0.006	&	0.008	&	0.000	\\
2	&	0	&	$<r_2$	&	0.042	&	0.002	&	0.002	&	0.116	&	0.020	&	0.014	\\
2	&	1	&	$<r_2$	&	0.132	&	0.032	&	0.012	&	0.406	&	0.176	&	0.042	\\
2	&	2	&	$<r_2$	&	0.072	&	0.024	&	0.008	&	0.342	&	0.164	&	0.076	\\
\hline
&&&\multicolumn{3}{c|}{$R_2=\lfloor N/2\rfloor$ ($p=1$), $R_2=\lfloor N/3\rfloor$ ($p>1$) }&\multicolumn{3}{c}{$R_2=\lfloor N/2\rfloor$ ($p=1$), $R_2=\lfloor N/3\rfloor$ ($p>1$) }\\
&&&\multicolumn{3}{c|}{$u=  \sqrt{2}$, $\delta^*=10^{-5}$}&\multicolumn{3}{c}{$u= \sqrt{2}$, $\delta^*=10^{-1}$}\\
$r_2$ & $r_3\,$ & $\wh{r}_2$ & $BT1$ & $BT2$ & $BT3$ & $BT1$ & $BT2$ & $BT3$\\
\hline
1	&	0	&	$<r_2$	&	0.000	&	0.000	&	0.000	&	0.000	&	0.000	&	0.000	\\
1	&	1	&	$<r_2$	&	0.000	&	0.000	&	0.000	&	0.000	&	0.000	&	0.000	\\
1	&	2	&	$<r_2$	&	0.000	&	0.002	&	0.000	&	0.000	&	0.002	&	0.000	\\
2	&	0	&	$<r_2$	&	0.006	&	0.000	&	0.000	&	0.020	&	0.000	&	0.000	\\
2	&	1	&	$<r_2$	&	0.006	&	0.002	&	0.002	&	0.028	&	0.008	&	0.002	\\
2	&	2	&	$<r_2$	&	0.004	&	0.002	&	0.004	&	0.034	&	0.016	&	0.006	\\
\hline
&&&\multicolumn{3}{c|}{$R_2=\lfloor N/2\rfloor$ ($p=1$), $R_2=\lfloor N/3\rfloor$ ($p>1$) }&\multicolumn{3}{c}{$R_2=\lfloor N/2\rfloor$ ($p=1$), $R_2=\lfloor N/3\rfloor$ ($p>1$) }\\
&&&\multicolumn{3}{c|}{$u=  5$, $\delta^*=10^{-5}$}&\multicolumn{3}{c}{$u=5$, $\delta^*=10^{-1}$}\\
$r_2$ & $r_3\,$ & $\wh{r}_2$ & $BT1$ & $BT2$ & $BT3$ & $BT1$ & $BT2$ & $BT3$\\
\hline
1	&	0	&	$<r_2$	&	0.000	&	0.000	&	0.000	&	0.000	&	0.000	&	0.000	\\
1	&	1	&	$<r_2$	&	0.004	&	0.002	&	0.000	&	0.024	&	0.022	&	0.004	\\
1	&	2	&	$<r_2$	&	0.000	&	0.002	&	0.000	&	0.010	&	0.010	&	0.000	\\
2	&	0	&	$<r_2$	&	0.034	&	0.004	&	0.002	&	0.096	&	0.018	&	0.014	\\
2	&	1	&	$<r_2$	&	0.100	&	0.022	&	0.006	&	0.382	&	0.134	&	0.022	\\
2	&	2	&	$<r_2$	&	0.110	&	0.052	&	0.018	&	0.350	&	0.198	&	0.098	\\
\hline
\hline
\end{tabular}
}
\end{table}

\begin{table}[h!]
\caption{Fraction of wrong detections when estimating $r_2$ }\label{tab:rob4}
\vskip -.5cm
\phantom{Table 13: Fraction of wrong }  $N=200$, $T=200$,  $r_1=0$ - part 2
\centering
\vskip .2cm
\footnotesize{
\begin{tabular}{lll |p{.09\textwidth}p{.09\textwidth}p{.09\textwidth} | p{.09\textwidth}p{.09\textwidth}p{.09\textwidth}}
\hline
\hline
&&&\multicolumn{3}{c|}{$R_2=N$ ($p=1$), $R_2=N$ ($p>1$) }&\multicolumn{3}{c}{$R_2=N$ ($p=1$), $R_2=N$ ($p>1$) }\\
&&&\multicolumn{3}{c|}{$u=  \sqrt 2$, $\delta^*=10^{-5}$}&\multicolumn{3}{c}{$u=\sqrt 2$, $\delta^*=10^{-1}$}\\
$r_2$ & $r_3\,$ & $\wh{r}_2$ & $BT1$ & $BT2$ & $BT3$ & $BT1$ & $BT2$ & $BT3$\\
\hline
1	&	0	&	$<r_2$	&	0.000	&	0.000	&	0.000	&	0.000	&	0.000	&	0.000	\\
1	&	1	&	$<r_2$	&	0.000	&	0.000	&	0.000	&	0.000	&	0.000	&	0.000	\\
1	&	2	&	$<r_2$	&	0.000	&	0.000	&	0.000	&	0.000	&	0.000	&	0.000	\\
2	&	0	&	$<r_2$	&	0.028	&	0.000	&	0.000	&	0.052	&	0.014	&	0.004	\\
2	&	1	&	$<r_2$	&	0.030	&	0.014	&	0.002	&	0.140	&	0.044	&	0.012	\\
2	&	2	&	$<r_2$	&	0.022	&	0.002	&	0.000	&	0.136	&	0.054	&	0.016	\\
\hline
&&&\multicolumn{3}{c|}{$R_2=N$ ($p=1$), $R_2=N$ ($p>1$) }&\multicolumn{3}{c}{$R_2=N$ ($p=1$), $R_2=N$ ($p>1$) }\\
&&&\multicolumn{3}{c|}{$u= 5$, $\delta^*=10^{-5}$}&\multicolumn{3}{c}{$u=5$, $\delta^*=10^{-1}$}\\
$r_2$ & $r_3\,$ & $\wh{r}_2$ & $BT1$ & $BT2$ & $BT3$ & $BT1$ & $BT2$ & $BT3$\\
\hline
1	&	0	&	$<r_2$	&	0.000	&	0.000	&	0.000	&	0.000	&	0.000	&	0.000	\\
1	&	1	&	$<r_2$	&	0.000	&	0.000	&	0.000	&	0.010	&	0.012	&	0.004	\\
1	&	2	&	$<r_2$	&	0.000	&	0.000	&	0.000	&	0.026	&	0.026	&	0.000	\\
2	&	0	&	$<r_2$	&	0.062	&	0.016	&	0.008	&	0.178	&	0.034	&	0.026	\\
2	&	1	&	$<r_2$	&	0.172	&	0.048	&	0.014	&	0.512	&	0.218	&	0.060	\\
2	&	2	&	$<r_2$	&	0.172	&	0.076	&	0.018	&	0.522	&	0.312	&	0.152	\\
\hline
&&&\multicolumn{3}{c|}{$R_2=\lfloor N/2\rfloor$ ($p=1$), $R_2=\lfloor N/2\rfloor$ ($p>1$) }&\multicolumn{3}{c}{$R_2=\lfloor N/2\rfloor$ ($p=1$), $R_2=\lfloor N/2\rfloor$ ($p>1$) }\\
&&&\multicolumn{3}{c|}{$u=  \sqrt 2$, $\delta^*=10^{-5}$}&\multicolumn{3}{c}{$u= \sqrt 2$, $\delta^*=10^{-1}$}\\
$r_2$ & $r_3\,$ & $\wh{r}_2$ & $BT1$ & $BT2$ & $BT3$ & $BT1$ & $BT2$ & $BT3$\\
\hline
1	&	0	&	$<r_2$	&	0.000	&	0.000	&	0.000	&	0.000	&	0.000	&	0.000	\\
1	&	1	&	$<r_2$	&	0.000	&	0.000	&	0.000	&	0.000	&	0.000	&	0.000	\\
1	&	2	&	$<r_2$	&	0.000	&	0.000	&	0.000	&	0.000	&	0.000	&	0.000	\\
2	&	0	&	$<r_2$	&	0.006	&	0.000	&	0.000	&	0.022	&	0.002	&	0.000	\\
2	&	1	&	$<r_2$	&	0.016	&	0.004	&	0.002	&	0.082	&	0.020	&	0.004	\\
2	&	2	&	$<r_2$	&	0.010	&	0.010	&	0.000	&	0.072	&	0.032	&	0.008	\\
\hline
&&&\multicolumn{3}{c|}{$R_2=\lfloor N/2\rfloor$ ($p=1$), $R_2=\lfloor N/2\rfloor$ ($p>1$) }&\multicolumn{3}{c}{$R_2=\lfloor N/2\rfloor$ ($p=1$), $R_2=\lfloor N/2\rfloor$ ($p>1$) }\\
&&&\multicolumn{3}{c|}{$u=  5$, $\delta^*=10^{-5}$}&\multicolumn{3}{c}{$u= 5$, $\delta^*=10^{-1}$}\\
$r_2$ & $r_3\,$ & $\wh{r}_2$ & $BT1$ & $BT2$ & $BT3$ & $BT1$ & $BT2$ & $BT3$\\
\hline
1	&	0	&	$<r_2$	&	0.000	&	0.000	&	0.000	&	0.000	&	0.000	&	0.000	\\
1	&	1	&	$<r_2$	&	0.000	&	0.000	&	0.000	&	0.008	&	0.008	&	0.000	\\
1	&	2	&	$<r_2$	&	0.000	&	0.000	&	0.000	&	0.004	&	0.004	&	0.000	\\
2	&	0	&	$<r_2$	&	0.032	&	0.006	&	0.002	&	0.106	&	0.016	&	0.010	\\
2	&	1	&	$<r_2$	&	0.138	&	0.036	&	0.008	&	0.428	&	0.186	&	0.048	\\
2	&	2	&	$<r_2$	&	0.122	&	0.058	&	0.022	&	0.418	&	0.240	&	0.106	\\
\hline
\hline
\end{tabular}
}
\end{table}

\begin{table}[h!]
\caption{Fraction of wrong detections when estimating $r_2$ }\label{tab:rob5}
\vskip -.5cm
\phantom{Table 13: Fraction of wrong }  $N=200$, $T=500$,  $r_1=0$ - part 1
\centering
\vskip .2cm
\footnotesize{
\begin{tabular}{lll |p{.09\textwidth}p{.09\textwidth}p{.09\textwidth} | p{.09\textwidth}p{.09\textwidth}p{.09\textwidth}}
\hline
\hline
&&&\multicolumn{3}{c|}{$R_2=N$ ($p=1$), $R_2=\lfloor N/3\rfloor$ ($p>1$) }&\multicolumn{3}{c}{$R_2=N$ ($p=1$), $R_2=\lfloor N/3\rfloor$ ($p>1$) }\\
&&&\multicolumn{3}{c|}{$u= \sqrt{2}$, $\delta^*=10^{-5}$}&\multicolumn{3}{c}{$u= \sqrt{2}$, $\delta^*=10^{-1}$}\\
$r_2$ & $r_3\,$ & $\wh{r}_2$ & $BT1$ & $BT2$ & $BT3$ & $BT1$ & $BT2$ & $BT3$\\
\hline
1	&	0	&	$<r_2$	&	0.000	&	0.000	&	0.002	&	0.000	&	0.000	&	0.002	\\
1	&	1	&	$<r_2$	&	0.002	&	0.002	&	0.000	&	0.002	&	0.004	&	0.000	\\
1	&	2	&	$<r_2$	&	0.000	&	0.000	&	0.000	&	0.000	&	0.000	&	0.000	\\
2	&	0	&	$<r_2$	&	0.006	&	0.000	&	0.000	&	0.008	&	0.000	&	0.002	\\
2	&	1	&	$<r_2$	&	0.000	&	0.000	&	0.000	&	0.008	&	0.000	&	0.000	\\
2	&	2	&	$<r_2$	&	0.002	&	0.004	&	0.002	&	0.010	&	0.006	&	0.002	\\
\hline
&&&\multicolumn{3}{c|}{$R_2=N$ ($p=1$), $R_2=\lfloor N/3\rfloor$ ($p>1$) }&\multicolumn{3}{c}{$R_2=N$ ($p=1$), $R_2=\lfloor N/3\rfloor$ ($p>1$) }\\
&&&\multicolumn{3}{c|}{$u= 5$, $\delta^*=10^{-5}$}&\multicolumn{3}{c}{$u=5$, $\delta^*=10^{-1}$}\\
$r_2$ & $r_3\,$ & $\wh{r}_2$ & $BT1$ & $BT2$ & $BT3$ & $BT1$ & $BT2$ & $BT3$\\
\hline
1	&	0	&	$<r_2$	&	0.000	&	0.000	&	0.002	&	0.000	&	0.000	&	0.002	\\
1	&	1	&	$<r_2$	&	0.002	&	0.004	&	0.000	&	0.004	&	0.004	&	0.002	\\
1	&	2	&	$<r_2$	&	0.000	&	0.000	&	0.000	&	0.000	&	0.000	&	0.000	\\
2	&	0	&	$<r_2$	&	0.014	&	0.004	&	0.002	&	0.032	&	0.006	&	0.006	\\
2	&	1	&	$<r_2$	&	0.026	&	0.004	&	0.000	&	0.124	&	0.036	&	0.004	\\
2	&	2	&	$<r_2$	&	0.022	&	0.012	&	0.004	&	0.094	&	0.054	&	0.014	\\
\hline
&&&\multicolumn{3}{c|}{$R_2=\lfloor N/2\rfloor$ ($p=1$), $R_2=\lfloor N/3\rfloor$ ($p>1$) }&\multicolumn{3}{c}{$R_2=\lfloor N/2\rfloor$ ($p=1$), $R_2=\lfloor N/3\rfloor$ ($p>1$) }\\
&&&\multicolumn{3}{c|}{$u=  \sqrt{2}$, $\delta^*=10^{-5}$}&\multicolumn{3}{c}{$u= \sqrt{2}$, $\delta^*=10^{-1}$}\\
$r_2$ & $r_3\,$ & $\wh{r}_2$ & $BT1$ & $BT2$ & $BT3$ & $BT1$ & $BT2$ & $BT3$\\
\hline
1	&	0	&	$<r_2$	&	0.000	&	0.000	&	0.000	&	0.000	&	0.000	&	0.000	\\
1	&	1	&	$<r_2$	&	0.000	&	0.000	&	0.000	&	0.000	&	0.000	&	0.000	\\
1	&	2	&	$<r_2$	&	0.000	&	0.002	&	0.000	&	0.000	&	0.002	&	0.000	\\
2	&	0	&	$<r_2$	&	0.002	&	0.000	&	0.000	&	0.002	&	0.000	&	0.000	\\
2	&	1	&	$<r_2$	&	0.004	&	0.000	&	0.000	&	0.008	&	0.004	&	0.000	\\
2	&	2	&	$<r_2$	&	0.000	&	0.000	&	0.002	&	0.002	&	0.000	&	0.002	\\
\hline
&&&\multicolumn{3}{c|}{$R_2=\lfloor N/2\rfloor$ ($p=1$), $R_2=\lfloor N/3\rfloor$ ($p>1$) }&\multicolumn{3}{c}{$R_2=\lfloor N/2\rfloor$ ($p=1$), $R_2=\lfloor N/3\rfloor$ ($p>1$) }\\
&&&\multicolumn{3}{c|}{$u=  5$, $\delta^*=10^{-5}$}&\multicolumn{3}{c}{$u=5$, $\delta^*=10^{-1}$}\\
$r_2$ & $r_3\,$ & $\wh{r}_2$ & $BT1$ & $BT2$ & $BT3$ & $BT1$ & $BT2$ & $BT3$\\
\hline
1	&	0	&	$<r_2$	&	0.000	&	0.000	&	0.000	&	0.000	&	0.000	&	0.000	\\
1	&	1	&	$<r_2$	&	0.000	&	0.000	&	0.000	&	0.000	&	0.000	&	0.000	\\
1	&	2	&	$<r_2$	&	0.000	&	0.002	&	0.000	&	0.000	&	0.002	&	0.000	\\
2	&	0	&	$<r_2$	&	0.010	&	0.000	&	0.000	&	0.034	&	0.002	&	0.000	\\
2	&	1	&	$<r_2$	&	0.024	&	0.004	&	0.000	&	0.094	&	0.030	&	0.004	\\
2	&	2	&	$<r_2$	&	0.012	&	0.000	&	0.002	&	0.088	&	0.030	&	0.008	\\
\hline
\hline
\end{tabular}
}
\end{table}

\begin{table}[h!]
\caption{Fraction of wrong detections when estimating $r_2$ }\label{tab:rob6}
\vskip -.5cm
\phantom{Table 13: Fraction of wrong }  $N=200$, $T=500$,  $r_1=0$ - part 2
\centering
\vskip .2cm
\footnotesize{
\begin{tabular}{lll |p{.09\textwidth}p{.09\textwidth}p{.09\textwidth} | p{.09\textwidth}p{.09\textwidth}p{.09\textwidth}}
\hline
\hline
&&&\multicolumn{3}{c|}{$R_2=N$ ($p=1$), $R_2=N$ ($p>1$) }&\multicolumn{3}{c}{$R_2=N$ ($p=1$), $R_2=N$ ($p>1$) }\\
&&&\multicolumn{3}{c|}{$u=  \sqrt 2$, $\delta^*=10^{-5}$}&\multicolumn{3}{c}{$u=\sqrt 2$, $\delta^*=10^{-1}$}\\
$r_2$ & $r_3\,$ & $\wh{r}_2$ & $BT1$ & $BT2$ & $BT3$ & $BT1$ & $BT2$ & $BT3$\\
\hline
1	&	0	&	$<r_2$	&	0.000	&	0.000	&	0.000	&	0.000	&	0.000	&	0.000	\\
1	&	1	&	$<r_2$	&	0.000	&	0.000	&	0.000	&	0.000	&	0.000	&	0.000	\\
1	&	2	&	$<r_2$	&	0.000	&	0.000	&	0.002	&	0.000	&	0.000	&	0.002	\\
2	&	0	&	$<r_2$	&	0.004	&	0.000	&	0.000	&	0.018	&	0.000	&	0.000	\\
2	&	1	&	$<r_2$	&	0.008	&	0.002	&	0.000	&	0.020	&	0.008	&	0.000	\\
2	&	2	&	$<r_2$	&	0.004	&	0.000	&	0.000	&	0.034	&	0.008	&	0.000	\\
\hline
&&&\multicolumn{3}{c|}{$R_2=N$ ($p=1$), $R_2=N$ ($p>1$) }&\multicolumn{3}{c}{$R_2=N$ ($p=1$), $R_2=N$ ($p>1$) }\\
&&&\multicolumn{3}{c|}{$u= 5$, $\delta^*=10^{-5}$}&\multicolumn{3}{c}{$u=5$, $\delta^*=10^{-1}$}\\
$r_2$ & $r_3\,$ & $\wh{r}_2$ & $BT1$ & $BT2$ & $BT3$ & $BT1$ & $BT2$ & $BT3$\\
\hline
1	&	0	&	$<r_2$	&	0.000	&	0.000	&	0.000	&	0.000	&	0.000	&	0.000	\\
1	&	1	&	$<r_2$	&	0.000	&	0.000	&	0.000	&	0.002	&	0.002	&	0.000	\\
1	&	2	&	$<r_2$	&	0.000	&	0.000	&	0.002	&	0.000	&	0.002	&	0.002	\\
2	&	0	&	$<r_2$	&	0.022	&	0.000	&	0.000	&	0.054	&	0.006	&	0.000	\\
2	&	1	&	$<r_2$	&	0.030	&	0.012	&	0.002	&	0.152	&	0.044	&	0.012	\\
2	&	2	&	$<r_2$	&	0.038	&	0.010	&	0.002	&	0.160	&	0.074	&	0.026	\\
\hline
&&&\multicolumn{3}{c|}{$R_2=\lfloor N/2\rfloor$ ($p=1$), $R_2=\lfloor N/2\rfloor$ ($p>1$) }&\multicolumn{3}{c}{$R_2=\lfloor N/2\rfloor$ ($p=1$), $R_2=\lfloor N/2\rfloor$ ($p>1$) }\\
&&&\multicolumn{3}{c|}{$u=  \sqrt 2$, $\delta^*=10^{-5}$}&\multicolumn{3}{c}{$u= \sqrt 2$, $\delta^*=10^{-1}$}\\
$r_2$ & $r_3\,$ & $\wh{r}_2$ & $BT1$ & $BT2$ & $BT3$ & $BT1$ & $BT2$ & $BT3$\\
\hline
1	&	0	&	$<r_2$	&	0.002	&	0.000	&	0.000	&	0.002	&	0.000	&	0.000	\\
1	&	1	&	$<r_2$	&	0.000	&	0.002	&	0.000	&	0.000	&	0.002	&	0.000	\\
1	&	2	&	$<r_2$	&	0.000	&	0.000	&	0.000	&	0.000	&	0.000	&	0.000	\\
2	&	0	&	$<r_2$	&	0.002	&	0.000	&	0.000	&	0.010	&	0.000	&	0.002	\\
2	&	1	&	$<r_2$	&	0.004	&	0.004	&	0.002	&	0.012	&	0.004	&	0.002	\\
2	&	2	&	$<r_2$	&	0.000	&	0.000	&	0.000	&	0.006	&	0.002	&	0.000	\\
\hline
&&&\multicolumn{3}{c|}{$R_2=\lfloor N/2\rfloor$ ($p=1$), $R_2=\lfloor N/2\rfloor$ ($p>1$) }&\multicolumn{3}{c}{$R_2=\lfloor N/2\rfloor$ ($p=1$), $R_2=\lfloor N/2\rfloor$ ($p>1$) }\\
&&&\multicolumn{3}{c|}{$u=  5$, $\delta^*=10^{-5}$}&\multicolumn{3}{c}{$u= 5$, $\delta^*=10^{-1}$}\\
$r_2$ & $r_3\,$ & $\wh{r}_2$ & $BT1$ & $BT2$ & $BT3$ & $BT1$ & $BT2$ & $BT3$\\
\hline
1	&	0	&	$<r_2$	&	0.002	&	0.000	&	0.000	&	0.002	&	0.000	&	0.000	\\
1	&	1	&	$<r_2$	&	0.000	&	0.002	&	0.000	&	0.000	&	0.002	&	0.000	\\
1	&	2	&	$<r_2$	&	0.000	&	0.000	&	0.000	&	0.000	&	0.000	&	0.000	\\
2	&	0	&	$<r_2$	&	0.014	&	0.002	&	0.002	&	0.032	&	0.008	&	0.006	\\
2	&	1	&	$<r_2$	&	0.018	&	0.006	&	0.006	&	0.122	&	0.038	&	0.010	\\
2	&	2	&	$<r_2$	&	0.010	&	0.004	&	0.002	&	0.116	&	0.040	&	0.012	\\
\hline
\hline
\end{tabular}
}
\end{table}

\end{document}